\documentclass[11pt, journal, draftclsnofoot, onecolumn]{IEEEtran}
\usepackage{color,amssymb}

\usepackage{tabularx}
\usepackage{graphicx,epsf,epsfig,verbatim,amssymb,amsmath,array,cite,amsthm,subfigure,multicol,multirow}

\usepackage{caption,bbm}
\usepackage{bbm}
\usepackage[T1]{fontenc}
\usepackage{tikz}
\usepackage[scr=dutchcal]{mathalfa}
\let\mathscr\mathbscr

\newcolumntype{x}[1]{>{\centering\arraybackslash}p{#1}}

\newcommand\blfootnote[1]{%
  \begingroup
  \renewcommand\thefootnote{}\footnote{#1}%
  \addtocounter{footnote}{-1}%
  \endgroup
}

\newcommand{\InputAlphabet}{\mathsf{X}}
\newcommand{\OutputAlphabet}{\mathsf{Y}}
\def\ulineInputAlphabet{\underline{\InputAlphabet}}
\def\ulineOutputAlphabet{\underline{\OutputAlphabet}}
\def\ulineoutput{\underline{y}}
\def\ulineinput{\underline{x}}
\def\2IC{$2-$IC}
\def\InputRV{X}
\def\OutputRV{Y}
\def\ulineOutputRV{\underline{\OutputRV}}

\def\ulineInputRV{\underline{\InputRV}}

\def\ulineoutput{\underline{y}}
\def\ulineinput{\underline{x}}

\def\costfn{\kappa}
\def\ulinecostfn{\underline{\costfn}}

\def\TimeSharingRVSet{\mathsf{Q}}
\def\TimeSharingRV{Q}

\def\SemiPrivateRV{U}

\newtheorem{theorem}{Theorem}
\newtheorem{proposition}{Proposition}
\newtheorem{lemma}{Lemma}
\newtheorem{corollary}{Corollary}

\newtheorem{definition}{Definition}

\newtheorem{remark}{Remark}

\newcommand{\bdefi}{\begin{definition}}
\newcommand{\edefi}{\end{definition}}
\newcommand{\beq}{\begin{equation}}
\newcommand{\eeq}{\end{equation}}
\newcommand{\ber}{\begin{eqnarray}}
\newcommand{\eer}{\end{eqnarray}}
\newcommand{\bel}{\begin{lemma}}
\newcommand{\eel}{\end{lemma}}

\newcommand{\mcx}{{\mathsf X}}

\newcommand{\mcr}{{\mathsf R}}

\newcommand{\mcq}{{\mathsf Q}}

\newcommand{\e}{\epsilon}

\newcommand{\nn}{\nonumber}

\newcommand{\mcn}{{\mathsf N}}
\def\triangleq{:{=}~}

\def\eE{{\mathbb E}}

\begin{document}

\title{Lattices from  Linear Codes: 
Source and Channel Networks}
\author{Farhad Shirani$^\dag$  and S. Sandeep Pradhan$^\diamond$
\blfootnote{
      This work was presented in part at IEEE
     International Symposium on Information Theory (ISIT), July
     2018. } \\
$^\dag$North Dakota State University, Fargo, ND $^\diamond$University of  Michigan, Ann Arbor, MI \\
$^\dag$f.shiranichaharsoogh@ndsu.edu, $^\diamond$pradhanv@umich.edu.}
\date{}

\maketitle \thispagestyle{empty} \pagestyle{plain}
\vspace{-0.7in}

\begin{abstract}
 In this paper, we consider the information-theoretic
 characterization of the set of achievable rates and distortions in a broad class of multiterminal communication scenarios with general continuous-valued sources and channels.  A framework is presented which involves fine discretization of the source and channel variables followed by communication over the resulting discretized network. In order to evaluate  fundamental performance limits, convergence results for  information measures are provided under the proposed discretization process. Using this framework, we consider point-to-point source coding
and channel coding with side-information, distributed source coding with distortion constraints, the function reconstruction problems (two-help-one), computation over multiple access channel, and the multiple-descriptions source coding problem. 
We construct lattice-like codes for general sources and channels, and derive inner-bounds to set of achievable rates and distortions in these communication scenarios. 
\end{abstract}

\section{Introduction}
Over the past several decades, information theory has provided a  framework for the study of the fundamental limits of communication --- such as achievable rates and distortions --- and design of source and channel coding strategies in a wide range of communication scenarios. The derivation of the optimal rate-distortion function in point-to-point (PtP) source coding and channel capacity in PtP channel coding for discrete sources and channels laid the groundwork for the development and implementation of practical capacity achieving codes such as low-density parity check (LDPC) codes \cite{gallager1962low}, turbo codes \cite{berrou1996near}, and polar codes \cite{arikan2009channel}, among others, in the following decades. Furthermore, these techniques for PtP data storage and transmission were extended to various multiterminal communication scenarios, where inner and outer bounds for the set of achievable rates and distortions for storage of discrete sources and transmission over discrete channels were derived \cite{ElGamalLec}. Although the derivation of tight performance bounds remains an open problem in many network communication scenarios, the optimal achievable region is known for some special cases of interest such as communication over discrete multiple access channel (MAC) \cite{liao1972multiple}, discrete deterministic and semi-deterministic broadcast channels \cite{gel1980capacity}, and multiple-descriptions source coding in the no-excess rate regime \cite{ahlswede1985rate}.   

Despite the significant progress in characterization of the fundamental limits of communication in network communication scenarios for discrete sources and channels, there is still a lack of a unifying framework for the study of data compression for general continuous sources and data transmission over general continuous channels. Many of the derivations in the discrete case rely on the concept of strong typicality which is based on the frequency of occurrence of symbols in sequences of discrete random variables \cite{csiszarbook}. The notion of strong typicality does not extend naturally to sequences of continuous variables.
Prior works address the issue in an ad-hoc fashion. As a result, the optimal performance in terms of achievable rates and distortions is usually known only for special cases when the underlying distributions of all variables are restricted to be Gaussian variables, e.g.
distributed storage of Gaussian sources \cite{wagner2008rate}, and communication over the Gaussian broadcast channel \cite{weingarten2006capacity}.
 For certain problems, such as PtP source coding, PtP channel coding, and communication over MAC,  performance limits have been derived using techniques using weak typicality \cite{Cover} instead of strong typicality. Weak typicality is based on the empirical entropy of sequences of random variables and is defined for both discrete and continuous variables.
However, weak typicality is not applicable in many multiterminal communication problems such as distributed source coding, and communication over broadcast channels, since for instance, the Markov lemma \cite{Markov} (a crucial step in the derivation
of achievable regions) is not valid for weakly typical
sequences. To address this, Wyner \cite{wyner1978rate} proposed a method for the study of PtP
 source coding with side-information, which can also be used for distributed source coding. Wyner's method involves fine
quantization of the source, the side-information, and the auxiliary
variables to create a finite-alphabet problem, and then
using the achievable results for the finite-alphabet problem to derive
performance limits for the original problem using convergence
properties of mutual information. This idea of `discretizing' the continuous communication system, and then applying discrete coding strategies and analytical techniques has also been recently used in the study of the compute-and-forward communication scenario \cite{pastore2021discretization}, as well as in investigating the correspondence between a set of useful inequalities in terms of entropies of discrete variables and their analogs in terms of differential entropies of continuous variables \cite{makkuva2018equivalence}. Another method to study continuous networks is to modify the notion of weak typicality, and use weak-* typicality instead. In \cite{mitran2015markov}, the Markov lemma has been shown to hold for weak-* typical sequences. The results were applied to source compression in the presence of side-information. The derivations in \cite{wyner1978rate,mitran2015markov} are based on unstructured random code ensembles.  An additional technique which has been considered for compression of linear quadratic
Gaussian  (LQG) sources and channels is to use subtractive dithered lattice
codes \cite{erez2004achieving,zamir2002nested}. The drawback of these lattice codes is that (a) they are
very specific to the LQG nature of the problem, and hence not amenable
to non-Gaussian and  nonlinear problems, and (b) they are based on
point-to-point communication perspective, and hence not general enough to be extended to the multiuser techniques such as joint quantization as seen
in multiple-description coding, and joint source-channel mapping as seen
in transmission of correlated sources over multiple-access channels. 
 
 In this paper we develop a unified framework for studying the performance
limits of communication for general continuous-valued sources and channels in
multi-terminal communication scenarios. The proposed method builds upon Wyner's  fine quantization technique. We derive the covering bounds, packing bounds, and prove Markov lemma when unstructured random code ensembles are used as well as for structured code
ensembles. These tools are used to derive the fundamental limits of communication in PtP communications with side-information, distributed source coding, lossy two-help-one problem, computation over multiple access channel (MAC), and multiple descriptions (MD) source coding. For computation over MAC with more than two users and MD source coding with more than two descriptions, we show that communication schemes using linear coding ensembles outperform those using unstructured random codes. This is inline with similar observations made in computation over MAC and MD compression of discrete sources. 

The rest of the paper is organized as follows: 
Section 
\ref{sec:3} develops a set of useful lemmas which form the framework for the analysis in the rest of the paper. Section \ref{sec:PtP} investigates point-to-point communication with side-information. Section \ref{sec:structured} considers using linear coding ensembles designed for discrete networks for communication over continuous networks. Section \ref{sec:DSC}
investigates distributed compression of continuous sources. Section \ref{sec:LTO} derives the fundamental limits of communication  in the lossy two-help-one problem. Section \ref{sec:IC} studies the computation over MAC problem. Section \ref{sec:MD} investigates the MD source coding problem. Section \ref{sec:conc} concludes the paper.  

\textit{Notation:} 
 We represent random variables by capital letters such as $X, U$ and their realizations by small letters such as $x, u$. Sets are denoted by sans-serif letters such as $\mathsf{X}, \mathsf{U}$. The set of natural numbers, and the real numbers are represented by $\mathbb{N}$, and $\mathbb{R}$ respectively. The Borel sigma-field is denoted by $\mathcal{B}$. Collections of sets are denoted by calligraphic letters such as $\mathcal{X},\mathcal{U}$. The random variable $\mathbbm{1}_{\mathcal{E}}$ is the indicator function of the event $\mathcal{E}$.
 The set of numbers $\{n,n+1,\cdots, m\}, n,m\in \mathbb{N}$ is represented by $[n,m]$. Furthermore, for the interval $[1,m]$, we sometimes use the shorthand notation $[m]$ for brevity. 
 For a given $n\in \mathbb{N}$, the $n$-length vector $(x_1,x_2,\hdots, x_n)$ is written as $x^n$. The function $h(\cdot)$ denotes the differential entropy. For the set $\mathsf{A}\subset \mathbb{R}^n$, we write $cl(\mathsf{A})$ to denote the convex closure.

\section{Preliminaries}
\label{sec:3}

\subsection{Source and Channel Models}
We consider continuous memoryless source and channel networks with real-valued inputs and outputs, and without feedback. Such channel networks (source networks) are completely characterized by their associated channel transition probability (source distribution) and input cost functions (output distortion functions). 

In the most general formulation, the transition probability function is defined as follows.
\begin{definition}[\textbf{Transition Probability}]
A transition probability is a function $P: \mathbb{R}
  \times \mathcal{B} \rightarrow \mathbb{R}$ such that: 
\begin{itemize}
\item For each $x \in \mathbb{R}$, $P(\cdot|x):\mathsf{A} \mapsto P(\mathsf{A}|x)$ is a
  probability measure on $(\mathbb{R},\mathcal{B})$.
\item For each $\mathsf{A} \in \mathcal{B}$, $P(\mathsf{A}|\cdot):x \mapsto P(\mathsf{A}|x)$ is a measurable function.
\end{itemize}
\end{definition}

\begin{remark}
\label{rem:assumptions}
We assume that for any variable $X$ the PDF exists and it approaches  infinity in at most a finite number of points, and that the set of points of discontinuity has (Lebesgue) measure zero. Furthermore, we assume that the variable does not have discrete points, i.e. $\nexists x\in \mathbb{R}, P(X=x)>0)$.
\end{remark}

\begin{definition}[\textbf{Memoryless Channel without Feedback}]
\label{def:ch}
A channel  is characterized by i) a transition probability $P_{Y|X}: \mathbb{R}
  \times \mathcal{B} \rightarrow \mathbb{R}$,  and ii) a continuous cost function $\kappa: \mathbb{R}
\rightarrow \mathbb{R}^+$, where $X$ and $Y$ are the channel input and output, respectively.   
\end{definition}
\begin{remark}
As noted in Definition \ref{def:ch}, we assume that the cost function  is continuous. This smoothness condition ensures that $\kappa(x)$ is bounded for $x\in \mathbb{R}$. 
\end{remark}

We assume that the channel is memoryless
and used without feedback, i.e. The joint probability measure on
$(\mathbb{R}^n,\mathcal{B}^n)$ is given by the unique product measure
\[
P(Y_i\in \mathsf{A}_i, i\in [n]|X^n=x^n)=\prod_{i=1}^n P_{Y|X}(\mathsf{A}_i|x_i),\qquad  \forall \mathsf{A}_1,\mathsf{A}_2,\cdots,\mathsf{A}_n\in \mathcal{B}.
\]
given $x^n$ is transmitted on the channel by using the channel $n$
times.  

\begin{definition}[\textbf{Joint Channel Probability Measure}]
For a channel $(P_{Y|X},\kappa)$, and given probability measure $P_X$ on $(\mathbb{R},\mathcal{B})$, the joint
probability measure $P_{XY}$ on  $(\mathbb{R}^2,\mathcal{B}^2)$  is  
the unique extension of the measure on product sets
\[
P_{XY}(\mathsf{A} \times \mathsf{B})= \int_\mathsf{A} P_{X}(dx) P_{Y|X}(\mathsf{B}|x) = \int_\mathsf{A} P_{X}(dx) \int_\mathsf{B} P_{Y|X}(dy|x),\qquad \mathsf{A},\mathsf{B}\in \mathcal{B}.
\]
\end{definition}
We also consider point-to-point and multiuser source coding scenarios. An information source is captured by its associated probability measure and distortion function as described below.
\begin{definition}[\textbf{Memoryless Source}]
\label{def:source}
A source is characterized by i) a probability measure $P_X:\mathcal{B}\to \mathbb{R}$, and ii) a jointly continuous distortion function $d: \mathbb{R} \times
\mathbb{R} \rightarrow \mathbb{R}^+$ .

\end{definition}

\begin{remark}
As noted in Definition \ref{def:source}, we assume that the distortion function is a jointly continuous. This smoothness condition ensures that $d(x,\hat{x})$ is bounded for all $x,\hat{x}\in \mathbb{R}$.
\end{remark}

\subsection{Source and Channel Discretization and Clipping}
\label{sec:discretization}

We will obtain achievable rate-distortion functions for source coding and
achievable rate-cost functions for channel coding by first
discretizing and clipping the associated  random variables. 
This  approach is described in the following.

\begin{definition}[\textbf{Discretization function}]
\label{Def:Desc}
Let $n \in \mathbb{N}$,
The discretization function $Q_{n}: \mathbb{R} \rightarrow
\mathbb{Z}_{n}$ is defined as 
\[
Q_{n}(s)= \arg \min_{a \in \mathbb{Z}_{n}}   |s-a|, s\in \mathbb{R},
\]
 where $\mathbb{Z}_{n}\triangleq\frac{1}{2^n} \mathbb{Z}$. 
\end{definition}
\begin{definition}[\textbf{Clipping Function}]
\label{Def:Clip}
For a given upper-limit $u\in \mathbb{R}$ and lower-limit $\ell\in \mathbb{R}$, the clipping
function $C_{\ell,u}: \mathbb{R} \rightarrow
\mathbb{R}$ is defined as
\[
C_{\ell,u}(s)=\max\{\min\{u,s\},-\ell\}, s\in \mathbb{R}.
\]
\end{definition}
\begin{definition}[\textbf{Discretization Cells}]
Given a pair of discretization and clipping functions $(Q_n,C_{\ell,u}), n\in \mathbb{N}, \ell,u>0$, the
associated  discrete alphabet 
$\mathbb{Z}_{l,n} \triangleq [-l,u] \bigcap \mathbb{Z}_{n} $,
and associated discretization cells  $\mathcal{A}_{\ell,u,n}(i)$ for
$i=0,1,2,\ldots,  (\lfloor u \ 2^{n} \rfloor - \lceil - \ell 2^n \rceil  - 1)$ are defined as
\begin{align}
\label{eq:disc_cell_1}
&\mathcal{A}_{\ell,u,n}(0)\triangleq \left( -\infty, \lceil -\ell 2^n \rceil \frac{1}{2^n} +
  \frac{1}{2^{n+1}}  \right],
\\&\mathcal{A}_{\ell,u,n}(i)\triangleq\mathcal{A}_{\ell,u,n}(i-1)+\frac{1}{2^n}, i=1,2,
  \ldots, (\lfloor u \ 2^{n} \rfloor - \lceil - \ell 2^n \rceil  - 2),
\\&
\mathcal{A}_{\ell,u,n}(\lfloor u \ 2^{n} \rfloor - \lceil - l 2^n \rceil  - 1)= \left( \lfloor u 2^n \rfloor \frac{1}{2^n} -
  \frac{1}{2^{n+1}},   \infty \right).
  \label{eq:disc_cell_3}
\end{align}
\end{definition}


\begin{remark}
To reduce clutter, we denote $\widehat{S}_{\ell,u,n}=Q_{n}(C_{\ell,u}(S))$ as $\widehat{S}$ when
the subscript is clear from the context. Moreover, we also denote 
$\widetilde{S}_{\ell,u}=C_{\ell,u}(S)$ by $\widetilde{S}$.
\end{remark}


\subsection{Convergence of Distributions, and Information Measures} 
In this section, we introduce several results on convergence of information measures which are used in the subsequent sections. 
\bdefi[\textbf{Convergence of Probability Measures}]
\label{def:distribution_conv}
Consider a sequence of
probability measures $P_n, n\in \mathbb{N}$, defined on the probability space $(\Omega,\mathcal{F})$, 
\begin{itemize}
\item \textbf{Strong Convergence:} $P_n, n\in \mathbb{N}$ is said to
converge strongly to $P$ if 
\[
\lim_{n \rightarrow \infty} P_n(\mathsf{A})=P(\mathsf{A}), \mathsf{A} \in \mathcal{F}.
\] 
\item \textbf{Convergence in Total Variation:}   $P_n, n\in \mathbb{N}$ is said to
converge in total variation to $P$ if 
\[
\lim_{n\to\infty} V(P_n,P)=0,
\] 
where $V(P,Q)\triangleq sup_{\mathsf{A}} |P(\mathsf{A})-Q(\mathsf{A})|$ is the total variation between  $P$ and $Q$. 
\end{itemize}

\edefi

\begin{remark}
It can be noted convergence in total variation guarantees strong convergence, which in turn guarantees convergence in distribution.  
\end{remark}

\begin{lemma}
[\textbf{Lower semi-continuity of Mutual Information}\cite{Pinsker}]
\label{lem:conv_mutual_info_1}
Consider a sequence of pairs of random variables
$(S_n,T_n), n\in \mathbb{N}$, defined on
 $(\mathbb{R}^2, \sigma(\mathcal{B} \times \mathcal{B}))$. If
 $P_{S_n,T_n}$ converges strongly to $P_{S,T}$, then 
\[
I(S;T) \leq \liminf_{n \rightarrow \infty} I(S_n;T_n).
\]
\end{lemma}

Intuitively, if we take $n \rightarrow \infty$, then the discrete random variable
$Q_n(S)$ converges to the continuous random variable $S$ in
distribution, and hence by Lemma \ref{lem:conv_mutual_info_1}, for variables $S$ and $T$, the mutual information $ I(Q_{n}(S); Q_{n}(T))$ converges to $I(S;T)$. This is stated formally in the following lemma.
\begin{lemma}[\textbf{Convergence of Discretized Variables}\cite{gray2011entropy}]
\label{lem:conv_distribution}
For any two random variables $(S,T)$,  the sequence 
$(Q_{n_1}(S),Q_{n_2}(T))$ converges in distribution to $(S,T)$ as $n_1,n_2\to \infty$. Consequently, if $I(S;T)< \infty$, we have 
\[
\lim_{n_1,n_2 \rightarrow \infty} I(Q_{n_1}(S); Q_{n_2}(T))= I(S;T).
\]
\label{thm:quantize0} 
\end{lemma}
\begin{lemma}[\textbf{Convergence of Clipped Variables}]
\label{thm:quantize}
Let $\ell_1,\ell_2,u_1,u_2 >0$, then
 For any two random variables $(S,T)$ with $I(S;T) < \infty$, the
 we have 
\[
\lim_{\ell_1,\ell_2,u_1,u_2 \rightarrow \infty} V(P_{\widetilde{S}_{\ell_1,u_1},\widetilde{T}_{\ell_2,u_2}}, P_{ST})=0
\]
and hence
\[
\lim_{\ell_1,\ell_2,u_1,u_2 \rightarrow \infty} I(\widetilde{S}_{\ell_1,u_1};
\widetilde{T}_{\ell_2,u_2})= I(S;T).
\]
\end{lemma}
\begin{proof}
The variational distance between $P_{ST}$ and
$P_{S_{\ell_1,u_1}T_{\ell_2,u_2}}$ can be given by 
\[
V(P_{\widetilde{S}_{\ell_1,u_1},\widetilde{T}_{\ell_2,u_2}}, P_{ST})= 2[1-P\left[ (-\ell_1 \leq S \leq
  u_1) \cap (-\ell_2 \leq T \leq u_2) \right] ].
\]
Hence we get 
\[
\lim_{\ell_1,\ell_2,u_1,u_2 \rightarrow \infty}
V(P_{\widetilde{S}_{\ell_1,u_1},\widetilde{T}_{\ell_2,u_2}}, P_{ST})=0. 
\]
\end{proof}

\section{Framework for Continuous to Discrete Source and Channel Transformation}

This section introduces the components of the discretization framework which is considered in subsequent sections to study communication over continuous sources and channel networks. We prove convergence of mutual information of sums of discretized random variables to that of their continuous counterparts. Theorems \ref{th:7} and \ref{th:9} are the main results of this section. 

\subsection{Convergence of Cost/Distortion Functions and smoothing of random variables}
The following lemma is used in evaluating the distortion and cost of communication strategies in continuous networks in the subsequent sections.
\begin{lemma}[\textbf{Convergence of Cost Functions and Distortion Functions}]
\label{thm:quantize_dist_cost}
Let $S$ and $T$ be  two random variables.
For any measurable function $\kappa:\mathbb{R}\rightarrow \mathbb{R}^+$
such that   $\mathbb{E}(\kappa(S))< \infty$, there exists two
increasing (and approaching $\infty$) sequences of lengths 
$l_{m},u_{m}$ such that 
\[
\lim_{m \rightarrow \infty} \mathbb{E} \kappa(\widetilde{S}_{l_{m},u_m})  =
\mathbb{E}\kappa(S).
\]
For any measurable  function $d: \mathbb{R}^2 \rightarrow \mathbb{R}^+$
such that   $\mathbb{E}(d(S,T))< \infty$,  there exists four
increasing (and approaching $\infty$) sequences of lengths 
$l_{1n},u_{1n}$, and $l_{2m},u_{2m}$  such that 
\[
\lim_{n \rightarrow \infty} \lim_{m \rightarrow \infty}
\mathbb{E} d(\widetilde{S}_{l_{1n},u_{1n}},\widetilde{T}_{l_{2m},u_{2m}}) =  \mathbb{E}d(S,T).
\]
\end{lemma}
\begin{proof}
Please refer to Appendix \ref{App:lem:6}
\end{proof}

In general, a continuous random variable $U$ may not have a continuous probability density function (PDF). In such scenarios, a useful technique is to `smoothen' the variable using additive noise. That is to construct $\widetilde{U}= U+N_\epsilon$, where $N_{\epsilon}$ is  uniformly distributed over $[-\epsilon,\epsilon], \epsilon>0$, and $\epsilon$ is a (small) positive number. Clearly, $\widetilde{U}$ has a continuous PDF, and it converges to $U$ in distribution.
The following lemma is used in the subsequent sections.
\begin{lemma}[\textbf{Smoothing of Random Variables}]
\label{lem:5}
Consider a bounded continuous random variable $U$ defined on the probability space $([-M,M], \mathscr{B}[-M,M], P_U)$,  such that $h(U)<\infty$ and $M>0$, and let $N_{\epsilon}$ be uniformly distributed over $[-\epsilon,\epsilon], \epsilon>0$. Assume that $U$ and $N_{\epsilon}$ are independent. Then,
\label{lem:smoothing}
\begin{align*}
    \lim_{\epsilon\to 0}I(N_{\epsilon};U+N_{\epsilon})= 0.
\end{align*}
\end{lemma}

\begin{proof}
Please refer to Appendix \ref{App:lem:5}.
\end{proof}
The following generalization of Lemma II.1 in \cite{shirani2018lattices} is used in the sequel in studying the Markov Lemma in distributed source coding under the proposed discretization process.
\begin{lemma}
\label{lem:mc_forced1}
For any quintuple  of random variables $\mathsf{A},\mathsf{B},\mathsf{C},\mathsf{D},E$ with joint distribution
that satisfies the Markov chain $(A,B) - C - (DE)$, consider a pair of random variables
$\widehat{A},\widehat{E}$ that are correlated with $(B,D)$ such that
$P_{BA}=P_{B\widehat{A}}$,  $P_{DE}=P_{D\widehat{E}}$, and $\widehat{A} - B - C - D -\widehat{E}$, then 
\[
I(A;C|B)+I(E;C|D) \geq \frac{1}{2 \ln 2} V^2(P_{CAE},P_{C\widehat{A}\widehat{E}}).
\]
\end{lemma}
\begin{proof}
Please see Appendix \ref{App:mc_forced1}
\end{proof}

\subsection{Discretization of Random Variables and their Sums}
\label{sec:structured}
We develop the framework that is used in the next sections to address structured
code ensembles. This requires evaluating mutual information terms which involve sums of random variables as discussed in the following. 
We consider jointly continuous random variables $X,Y,U,V$ with a joint PDF $f_{XY}f_{U|X}f_{V|Y}$. So that the variables satisfy the Markov chain $U \leftrightarrow X\leftrightarrow Y \leftrightarrow V$. We denote the joint probability measure as
$P_{XYUV}$.   Due to an analogy with distributed source coding, we refer to $(X,Y)$ as the source variables, and $(U,V)$ s the auxiliary variables.
\subsubsection{Discretization of Auxiliary Random Variables}
\label{sec:th:7}
Fix $\ell,\ell', \epsilon>0$, and $n\in \mathbb{N}$. Define the clipped variables ${U}_{\ell},{V}_{\ell'}$
as follows:
\begin{align*}
   & \widetilde{U}_{\ell}=
    \begin{cases}
    U \qquad & \text{ if } U\in [-\ell, \ell],\\
    U' & \text{Otherwise.}
    \end{cases}, \qquad 
    & \widetilde{V}_{\ell'}'=
    \begin{cases}
    V \qquad & \text{ if } V\in [-\ell', \ell'],\\
    V' & \text{Otherwise.}
    \end{cases},
\end{align*}
where $U',V'$ are independent of each other and $U,V$ and generated according to $    f_{U'}(\cdot)\triangleq f_{U|U\in [-\ell,\ell]}(\cdot)$,
and $f_{V'}(\cdot)\triangleq f_{V|V\in [-\ell',\ell']}(\cdot)$, respectively.
We take $\ell,\ell'$ sufficiently large.
Next, define the smoothed random variables $\widetilde{U}_{\ell,\epsilon}, \widetilde{V}_{\ell',\epsilon}$, where 
\begin{align*}
&{\widetilde{U}_{\ell,\epsilon}}\triangleq \widetilde{U}_\ell+{\widetilde{N}_{\ell,\epsilon}},\qquad  {\widetilde{V}_{\ell',\epsilon}}\triangleq \widetilde{V}_{\ell}+{\widetilde{N}'_{\ell',\epsilon}}
,
\\
& 
f_{\widetilde{N}_{\ell,\epsilon}}(\tilde{n})= \frac{1}{2\epsilon}, \quad \tilde{n}\in (-\epsilon,\epsilon),
\qquad 
f_{\widetilde{N}'_{\ell',\epsilon}}(\tilde{n}')= \frac{1}{2\epsilon},\quad  \tilde{n}'\in (-\epsilon,\epsilon),
\end{align*}
and  the variables $\widetilde{N}_{\ell,\epsilon}$ and $\widetilde{N}'_{\ell',\epsilon}$ are mutually independent of each other and of $X,Y,V,U, U_{\ell},V_{\ell'}$. 
Consider discretizing  ${\widetilde{U}_{\ell,\epsilon}}$ and ${\widetilde{V}_{\ell',\epsilon}}$ to $\widehat{U}_{\ell,\epsilon,n}= Q_{n}(\widetilde{U}_{\ell,\epsilon})$ and $\widehat{V}_{\ell,\epsilon,n}= Q_n(\widetilde{V}_{\ell',\epsilon})$, respectively.
Note that the by construction the Markov chain $\widehat{U}_{\ell,\epsilon,n} \leftrightarrow X \leftrightarrow Y \leftrightarrow  \widehat{V}_{\ell,\epsilon,n} $ holds. 
We have the following theorem.

\begin{theorem}
\label{th:7}
For any $\xi>0$ there exists $n,\ell,\ell',\epsilon>0$ such that 
\begin{align}
        \label{eq:5.1.1}&
        |I(X;  \widehat{U}_{n,\ell,\epsilon})
        - I(X;U)|\leq \xi\\
    \label{eq:5.1.2}&
    |I(Y; \widehat{V}_{n,\ell',\epsilon})-I(Y;V)|\leq \xi\\ 
    \label{eq:5.1.3}&
    |I( \widehat{U}_{n,\ell,\epsilon}+ \widehat{V}_{n,\ell',\epsilon}; \widehat{U}_{n,\ell,\epsilon} )- I(U+V;U)|\leq \xi
    \\ 
    \label{eq:5.1.4}&
    |I(\widehat{U}_{n,\ell,\epsilon}+ \widehat{V}_{n,\ell',\epsilon}; \widehat{V}_{n,\ell',\epsilon} )- I(U+V;V)|\leq \xi ,
\\&|I( \widehat{U}_{n,\ell,\epsilon};  \widehat{V}_{n,\ell',\epsilon})-  I(U;V)|\leq \xi,    \label{eq:5.1.5}
\end{align} 
Furthermore, given a pair of jointly continuous distortion functions $d_i:\mathbb{R}^2\to \mathbb{R}^+, i\in \{1,2\}$ and continuous reconstruction functions $g_i:\mathbb{R}^2\to \mathbb{R}$, we have:
\begin{align}
    &|\mathbb{E}(d_1(X,g_1(\widehat{U}_{n,\ell,\epsilon},  \widehat{V}_{n,\ell',\epsilon})))-  \mathbb{E}(d_1(X,g_1(U,V)))|\leq \xi
    \label{eq:5.1.6}
    \\&|\mathbb{E}(d_2(Y,g_2(\widehat{U}_{n,\ell,\epsilon},  \widehat{V}_{n,\ell',\epsilon})))-  \mathbb{E}(d_2(Y,g_2(U,V)))|\leq \xi
    \label{eq:5.1.7}
\end{align}
\end{theorem}

\begin{proof}
Please see Appendix \ref{App:th:7}
\end{proof}




\subsubsection{Discretization of the Source Variables}
\label{sec:th:9}
In the following, we describe the procedure for discretizing the source variables while ensuring that the long Markov chain holds. Let $\ell,\ell'>0$, and 
 $Z$ and $W$ be two random variables
that are independent of the source $(X,Y)$ such that  $Z \in
[-\ell,\ell]$ with probability one, 
$W \in [-\ell',\ell']$ with probability one, 
and the distribution $P_{Z}P_{W}$ is  
given by 
\[
P_{Z}P_{W}(A \times B)= \frac{P_{X}(\mathsf{A} \cap [-\ell,\ell] )P_{Y}(\mathsf{B} \cap
  [-\ell',\ell'])}{P_{X}([-\ell,\ell])P_{Y}( [-\ell',\ell])}
\]
for all events $\mathsf{A}$ and $\mathsf{B}$. in Borel sigma algebra. 
Define the clipped source variables as:
\begin{align}
\widetilde{X}_{\ell}=\left\{
\begin{array}{cc}
X & \mbox{ if } X \in [-\ell,\ell]  \\
Z & \mbox{ otherwise}
\end{array} \right.
\end{align}
and
\begin{align}
\label{eq:Y_clip}
\widetilde{Y}_{\ell'}=\left\{
\begin{array}{cc}
Y & \mbox{ if } Y \in [-\ell',\ell]  \\
W & \mbox{ otherwise}
\end{array} \right.
\end{align}
Furthermore, let $n\in \mathbb{N}$, and define the quantized and clipped source variables $\widehat{X}_{n,\ell}\triangleq Q_{n}(\widetilde{X}_{\ell})$, and $\widehat{Y}_{n,\ell'}\triangleq Q_{n}(\widetilde{Y}_{\ell'})$. 

\begin{theorem}
\label{th:9}
Given a quadruple of random variables $(X,Y,U,V)$, where i) $(X,Y)$ are jointly continuous with joint PDF $f_{X,Y}$, and ii) $U,V$ are discrete random variables defined on finite sets $\mathcal{U}$ and $\mathcal{V}$, respectively, and  iii) the long Markov chain $U - X - Y - V$ holds. Then,  
For any $\xi>0$ there exists $n,\ell,\ell'>0$ and variables $\overline{U}_{n,\ell}$ and $\overline{V}_{n,\ell'}$ defined on $\mathsf{U}\times \mathsf{V}$
such that the long  Markov chain $\overline{U}_{n,\ell} - \widehat{X}_{n,\ell} - \widehat{Y}_{n,\ell'} -\overline{V}_{n,\ell'} $ holds, and the following conditions are satisfied  
\begin{align} 
        \label{eq:9.1.1}&
       |I(\widehat{X}_{n,\ell}; \overline{U}_{n,\ell})
        - I(X;U)|\leq \xi\\
    \label{eq:9.1.2}&
    |I(\widehat{Y}_{n,\ell'}; \overline{V}_{n,\ell'})-I(Y;V)|\leq \xi\\ 
    \label{eq:9.1.3}&
    |I(\overline{U}_{n,\ell}+ \overline{V}_{n,\ell'}; \overline{U}_{n,\ell} )- I(U+V;U)|\leq \xi
    \\ 
    \label{eq:9.1.4}&
    |I(\overline{U}_{n,\ell}+ \overline{V}_{n,\ell'}; \overline{V}_{n,\ell'} )- I(U+V;V)|\leq \xi ,
\\&|I( \overline{U}_{n,\ell};  \overline{V}_{n,\ell'})-  I(U;V)|\leq \xi,   
\label{eq:9.1.5}
\end{align} 
Furthermore, given a pair of jointly continuous distortion functions $d_i:\mathbb{R}^2\to \mathbb{R}^+, i\in \{1,2\}$ and continuous reconstruction functions $g_i:\mathbb{R}^2\to \mathbb{R}$, we have:
\begin{align}
    &|\mathbb{E}(d_1(\widehat{X}_{n,\ell},g_1(\overline{U}_{n,\ell},  \overline{V}_{n,\ell'})))-  \mathbb{E}(d_1(X,g_1(U,V)))|\leq \xi
    \label{eq:9.1.6}
    \\&|\mathbb{E}(d_2(\widehat{Y}_{n,\ell},g_2(\overline{U}_{n,\ell},  \overline{V}_{n,\ell'})))-  \mathbb{E}(d_2(Y,g_2(U,V)))|\leq \xi
    \label{eq:9.1.7}
\end{align}
\end{theorem}

\begin{proof}
Please refer to Appendix \ref{App:th:9}.
\end{proof}

\begin{corollary}
\label{th:10}
Given a triple of random variables  $(Y,U,V)$, For any $\xi>0$ there exists $\epsilon,n,n', \ell,\ell',\ell''>0$ such that the following conditions are satisfied  
\begin{align} 
&
    |I(\widehat{U}_{n,\ell,\epsilon}+ \widehat{V}_{n,\ell',\epsilon},\widehat{Y}_{n',\ell''}; \widehat{U}_{n,\ell,\epsilon} )- I(U+V,{Y};U)|\leq \xi,
\end{align} 
where $\widehat{Y}_{n',\ell''}= Q_{n'}(\widetilde{Y}_{\ell''})$, and $\widetilde{Y}_{\ell''}$ is defined as in \eqref{eq:Y_clip}. 
\end{corollary}
\begin{proof}
First we apply clipping  and discretization
of the random variable $Y$. Then by the data processing inequality and the lower semi-continuity of mutual information, we have \[
|I(U+V,\widehat{Y}_{n',\ell''}; U)- I(U+V,{Y};U)|\leq \xi,
\]
for all sufficiently large $n'$ and $\ell''$.
Next we note that 
\[
I(U+V,\widehat{Y}_{n',\ell''}; U)=I(U+V;V|\widehat{Y}_{n',\ell''})+I(U;\widehat{Y}_{n',\ell''}).
\]
Regarding the first term, since $\widehat{Y}_{n',\ell''}$ is has a finite alphabet, we can 
 apply Theorem \ref{th:9} on each $I(U+V;V|\widehat{Y}_{n',\ell''}=\widehat{y})$ for each value of $\widehat{y}$ and show convergence. Regarding the second term,  $I(\widehat{U}_{n,\ell,\epsilon};\widehat{Y}_{n',\ell''})$ converges to $I(U;\widehat{Y}_{n',\ell''})$ due to lower semi-continuity of mutual information and data processing inequality.  
\end{proof}

\section{Point-to-point Communication with side-information}
\label{sec:PtP}
As a first step, we derive  the  the optimal rate-distortion function for source coding and optimal
capacity-cost functions for channel coding by first
discretizing the associated  random variables, and using transmission
systems designed for discrete sources and channels.  Wyner proved achievability in the source coding with side-information for general probability spaces under two general boundedness and smoothness conditions given in  \cite[p.64]{wyner1978rate}. While the achievability results in the PtP scenario have been shown in prior works, the method developed here provides the foundation for a unified framework which is used in the subsequent sections to address various problems in multiterminal communication.
\begin{figure}[!htb]
    \begin{center}
    \scalebox{2.4}{\includegraphics[width=5.5cm]{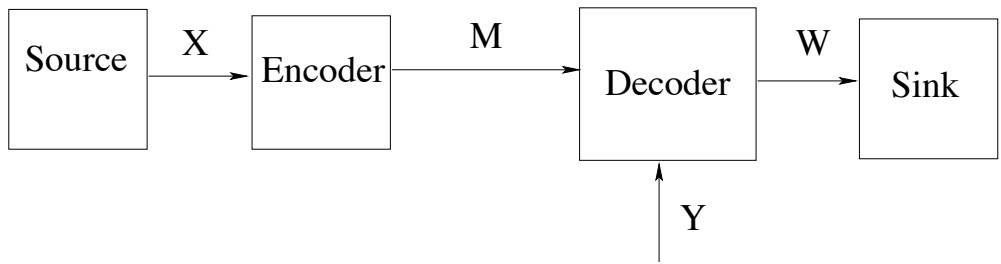}}
    \end{center}
     \caption{Source coding with side
     information.}
     \vspace{-0.4in}
\end{figure}
\subsection{Source Coding With Side-Information at Decoder}

Consider a memoryless 
source $X$ with side-information $Y$ given by $(P_{XY},d)$ comprising of a probability measure
$P_{XY}$ on $\mathbb{R}^2$, with reconstruction alphabet $\mathbb{R}$. The (single-letter)
distortion function $d: \mathbb{R}^3  \rightarrow \mathbb{R}^+$ is assumed to be jointly continuous function, where $d(x,y,\hat{x})$ depends on the input value $x$, the state value $y$, and the reconstruction value $\hat{x}$. 
The single-letter distortion function  induces a
distortion function on $n$-length sequences as follows with a slight
abuse of notation: for all $x^n,y^n \in \mathbb{R}^n\times \mathbb{R}^n$ and $\hat{x}^n \in \mathbb{R}^n$, 
\[
d_n(x^n,y^n,\hat{x}^n)=\frac{1}{n} \sum_{i=1}^n d(x_i,y_i,\hat{x}_i).
\]
In this case, we say that the distortion function on $n$-length
sequence is additive. Observe that the distortion function on
$n$-length sequences is normalized with $n$. 

\bdefi[\textbf{Transmission System}]
An $(n,\Theta)$ transmission system consists of an encoder mapping and a decodes mapping:
\[e:\mathbb{R}^n \rightarrow \{1,2,\ldots,\Theta \}, \mbox{ and }
f:\{1,2,\ldots,\Theta\} \times \mathbb{R}^n \rightarrow \mathbb
{R}^n.\]
A rate distortion pair $(R,D)$ is said to be achievable if
there exists a sequence of $(n,\Theta_n)$ transmission systems  such that 
\[
\limsup_{n \rightarrow \infty} \frac{\log \Theta_n}{n} \leq R,
\qquad  \limsup_{n \rightarrow \infty} \mathbb{E}d_n(X^n,Y^n,f(e(X^n),Y^n)) \leq
D.
\]
Let the operational 
rate-distortion function $R_{op}(D)$ (the Wyner-Ziv rate-distortion function \cite{wyner1973theorem}) denote the infimum of rates $R$ such that $(R,D)$ is
achievable.
\edefi

\begin{remark}
One can equivalently define an achievable pair $(R,D)$ as a pair for which there exists a sequence of $(n,\Theta_n)$ transmission systems with rate $r_n=\frac{\log{\Theta}_n}{n}$ and average distortion $\overline{d}_n=\mathbb{E}d_n(X^n,Y^n,f(e(X^n),Y^n))$ such that $\liminf_{n\to \infty} |R-r_n|^++ |D-\overline{d}_n|^+= 0$.
\end{remark}

 We prove the following theorem for a general probability measure 
$P_{XY}$ and continuous distortion function $d$. 
\begin{theorem}[\textbf{The Wyner-Ziv Rate-Distortion Function}]
For a given source $(P_{XY},d)$,  we have 
$R_{op}(D) \leq \mathsf{R}_{WZ}(D)$, where 
\[
\mathsf{R}_{WZ}(D) \triangleq  \min_{\{P_{U|X},g(\cdot,\cdot)\}} \left[ I(U;X)-I(U;Y) \right],
\]
and the minimization is carried out over all transition probability 
$P_{U|X}$, and continuous functions $g: \mathbb{R}^2 \rightarrow \mathbb{R}$  such that 
$Y \rightarrow X \rightarrow U$,  and $\eE[d(X,Y,g(U,Y))] \leq D$.
 \label{thm:wynerziv_c}
\end{theorem}

\noindent \textit{Proof Outline.} A detailed proof is provided in Appendix \ref{app:wynerziv_c}. We provide an outline of the proof in the following. 
 Let us consider the source $(X,Y)$  and an auxiliary variable $U$ with joint
distribution $P_{XYU}$ such that $Y \leftrightarrow X \leftrightarrow U$. The proof follows in three main steps: i) clipping $X,Y,U$ to produce bounded variables $\widetilde{X},\widetilde{Y},\widetilde{U}$,  ii) discretizing $\widetilde{X},\widetilde{Y},\widetilde{U}$ to produce discrete variables $\widehat{X},\widehat{Y},\widehat{U}$, and iii) source coding with side-information using the Wyner-Ziv scheme for discrete sources and channels \cite{wyner1973theorem} to construct and transmit the quantization $\widehat{U}$ of the source $\widehat{X}$ with side-information $\widehat{Y}$.
To elaborate, in the clipping step, the source and side-information $X,Y$ are clipped to produce bounded random variables $(\widetilde{X},\widetilde{Y})$ which take values in $[\ell_X,u_X]$, where $\ell_X,u_X>0$, such that
$P(\widetilde{X}=X)\to 1$ for asymptotically large $\ell_X,u_X$. Then, we produce the clipped variables $\widetilde{U},\widetilde{Y}$ which takes values from $[\ell_U,u_U]\times [\ell_Y,u_Y]$, where $\ell_U,\ell_Y,u_U,u_Y>0$, such that i) $P(\widetilde{U}=U,\widetilde{Y}=Y)\to 1$ for asymptotically large  $\ell_U,\ell_Y,u_U,u_Y$, and ii) the Markov chain $\widetilde{Y}\leftrightarrow \widetilde{X} \leftrightarrow \widetilde{U}$ is satisfied. 
Note that the clipping limits $\ell_X,\ell_Y,\ell_U, u_X,u_Y,u_U$ are taken to be distinct to avoid complications due to their interrelations when taking their limits to infinity. Using the lower semi-continuity of the mutual information, the data processing inequality, along with Lemma \ref{lem:mc_forced1} we show that the rate expressions in the theorem statement do not change if the original variables $X,Y,U$ are replaced by the clipped variables $(\widetilde{X},\widetilde{Y},\widetilde{U})$ for asymptotically large $\ell_X,\ell_Y,\ell_U, u_X,u_Y,u_U$.  
Next, we discretize $\widetilde{X}$ to produce $\widehat{X}$ using the discretization function introduced in Section \ref{sec:discretization}, with discretization step-size ${2^{-n_X}}$, where $n_X\in \mathbb{N}$. Note that the CDF of $F_{\widehat{X}}(\cdot)\to F_X(\cdot)$ for asymptotically large $n_X$. Next, we produce two conditionally independent discretized auxiliary variables $\widehat{Y}$ and $\widehat{U}$,  conditioned on $\widehat{X}$, with discretization step-sizes ${2^{-n_Y}}$ and ${2^{-n_U}}$, respectively, where $n_Y,n_U\in \mathbb{N}$, such that  $F_{\widehat{X},\widehat{Y}}(\cdot,\cdot)\to F_{X,Y}(\cdot,\cdot)$ and $F_{\widehat{X},\widehat{U}}(\cdot,\cdot)\to F_{X,U}(\cdot,\cdot)$  
 for asymptotically large $n_X,n_Y,n_U$. We use the lower semi-continuity of mutual information, the data processing inequality, along with Lemma \ref{lem:mc_forced1} to show that the rate expressions in the theorem statement do not change if the clipped variables $\widetilde{X},\widetilde{Y},\widetilde{U}$ are replaced by the discretized variables $\widehat{X},\widehat{Y},\widehat{U}$ for asymptotically large $n_X,n_Y,n_U$. The discretization steps ${2^{-n_X}},{2^{-n_Y}},{2^{-n_U}}$ are chosen to be distinct to avoid complications due to their interrelations when taking limits.  
  Furthermore, we show that there exists a reconstruction of the original source $X$ using $\widehat{Y}$ and $\widehat{U}$ such that the resulting expected distortion approaches that of the reconstruction of $X$ using $g(Y,U)$ for asymptotically large  clipping limits $\ell_X,\ell_Y,\ell_U,u_X,u_Y,u_U$ and discretization parameters $n_X,n_Y,n_U$. 
  Based on the above derivations, we propose the following source coding scheme.
Given a source sequence $X^n$ and side-information sequence $Y^n$, the encoder and decoder produce  $\widehat{X}^n$ and $\widehat{Y}$, respectively. 
The sequence $\widehat{X}^n$, whose elements are finite discrete variables, is quantized using the Wyner-Ziv source coding scheme, and $\widehat{U}^n$ is sent to the decoder. Finally, the decoder produces a reconstruction of $X^n$ using $\widehat{Y}^n$ and $\widehat{U}^n$. It is shown that the resulting rates and distortions due to the proposed scheme converge to the ones in the theorem statement as the clipping limits and discretization parameters are taken to be asymptotically large.

\subsection{Channel Coding with Side-Information at Transmitter}

Consider a channel with state $(P_{Y|XS},P_S,\kappa)$ comprising of a transition probability 
$P_{Y|XS}: \mathbb{R}^2 \times \mathcal{B} \rightarrow \mathbb{R}$, a
state random variable $S$ with a probability measure $P_S$, and a cost
function $\kappa:\mathbb{R}^2 \rightarrow \mathbb{R}^+$, $\kappa(x,s)$ depends on the channel input  value $x$ and the channel state value $s$ and is assumed to be jointly continuous. Furthermore,
we assume that the channel state is observable at the
encoder noncausally. 
 This single-letter cost
function induces a cost function on $n$-length sequences as follows
with a slight abuse of notation: for all $x^n \in \mathbb{R}^n$, 
\[
\kappa_n(s^n,x^n)=\frac{1}{n} \sum_{i=1}^n \kappa(s_i,x_i).
\]
In this case we say that the cost function is additive. 
For any transition probability  $P_{X|S}:\mathbb{R} \times \mathcal{B}
\rightarrow \mathbb{R}^+$, we define the joint
probability measure $P_{XSY}$ on the measurable
  space $(\mathbb{R}^3,\mathcal{B}^3)$ as 
the unique extension of the measure on product sets
\[
P_{XSY}(A \times B \times C)= \int_A P_{S}(ds) \int_B P_{X|S}(dx|s)
\int_C P_{Y|XS}(dy|x,s)
\]
We assume that the state is IID, and the channel is stationary, memoryless
and used without feedback.  The time-line of the random variables
associated with this problem for a transmission system of blocklength
$n$ is given by
\[
M,S_1,S_2,\ldots,S_n,X_1,Y_1,X_2,Y_2,\ldots,X_n,Y_n,\widehat{M},
\]
where $(M,\widehat{M})$ are the random message to be transmitted over the channel and its reconstruction, respectively.
The assumptions regarding the relation between these random variables
are the following: for all $i \in \{1,2,\ldots,n\}$, 
(a) $P_{S_i|S^{i-1},M}=P_{S_i}=P_{S}$, (b) $X_i$ is a function of $M$ and
$S^n$,  and (c) $P_{Y_i|M,S^n,X^i,Y^{i-1}}=P_{Y_i|X_i,S_i}=P_{Y|XS}$.

\begin{definition}[\textbf{Transmission System}]
An $(n,\Theta)$ transmission system for communication over a given channel with side information
consists of an encoder mapping and a decoder mapping
\[
e: \{1,2,\ldots,\Theta\} \times \mathbb{R}^n \rightarrow \mathbb{R}^n, \ \ \ 
f: \mathbb{R}^n \rightarrow \{1,2,\ldots,\Theta\}.
\]
The rate-cost pair $(R,\tau)$ is said to be achievable if there exists a sequence of $(n,\Theta_n)$ transmission systems such that:
\begin{align*}
&R\leq \liminf_{n\to \infty}\frac{ \log \Theta_n}{n} ,\qquad  \limsup_{n\to \infty}\frac{1}{\Theta_n} \sum_{i=1}^{\Theta_n} \eE \kappa_n(\left(e(i,S^n)),S^n \right)\leq \tau, 
\\&\lim_{n\to \infty} \frac{1}{\Theta}\sum_{i=1}^{\Theta} 
 \int_{\mathbb{R}^n} P_S^n(ds^n) 
  P_{Y|XS}^n(f(Y^n) \neq i|e(i,s^n), s^n)=0.
\end{align*}
One can define the operational capacity-cost  function (Gelfand-Pinsker capacity-cost function \cite{Gelfand})
with side-information $C_{op}(\tau)$ as the supremum of rates $R$ such that $(R,\tau)$ is achievable.
\end{definition}

\begin{figure}[!t]
    \begin{center}
    \scalebox{2.1}{\includegraphics[width=6cm]{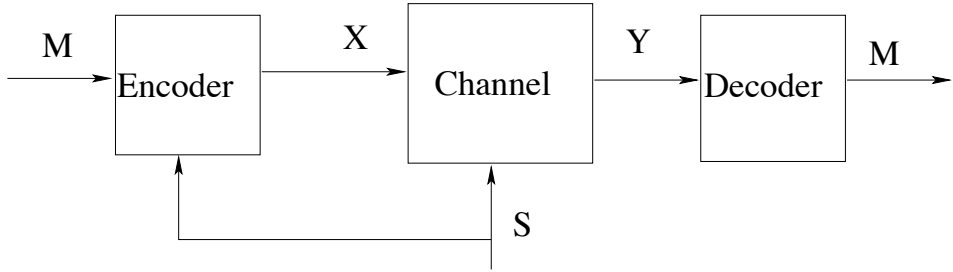}}
    \end{center}
     \caption{A schematic of
    channel coding with side-information.}
\end{figure}

\begin{theorem}[\textbf{The Gelfand-Pinsker Capacity-Cost Function}]
For a given channel with state  $(P_{Y|XS},P_S,\kappa)$, 
we have 
$C_{op}(\tau) \geq \mathsf{R}_{GF}(\tau)$, where 
\[
\mathsf{R}_{GF}(\tau) \triangleq  \sup_{\{P_{U|S},g(\cdot,\cdot)\}} \left[ I(U;X)-I(U;Y) \right],
\]
and the maximization is carried out over all transition probability 
$P_{U|S}$, and continuous functions $g: \mathbb{R}^2 \rightarrow \mathbb{R}$  such that $X=g(U,S)$,  
$U \rightarrow (X,S) \rightarrow Y$,  and $\eE[\kappa(g(U,S),S] \leq \tau$.
 \label{thm:Gelfand}
\end{theorem}
\begin{proof}
A detailed proof is provided in Appendix \ref{App:Gelfand}. We provide an outline in the following.
The proof builds upon the discretization ideas used in the proof of Theorem \ref{thm:wynerziv_c}. That is, it follows in three main steps: i) clipping $Y,S,U$ to produce bounded variables $\widetilde{Y},\widetilde{S},\widetilde{U}$,  ii) discretizing $\widetilde{Y},\widetilde{S},\widetilde{U}$ to produce discrete variables $\widehat{Y},\widehat{S},\widehat{U}$ while ensuring that the Markov chain $\widehat{U} \rightarrow \widehat{X},\widehat{S} \rightarrow \widehat{Y}$ holds, and iii) channel coding with side-information using the Gelfand-Pinsker scheme for discrete sources and channels \cite{Gelfand} to construct the discrete variables $\widehat{U},\widehat{S}$ and transmit $X=g(\widehat{U},\widehat{S})$ and decode the message based on the received discrete variable $\widehat{Y}$. Lastly, convergence results from Section \ref{sec:3} are used to prove that as the discretization grid-size approaches 0 and the clipping limits approach infinity, the resulting rates and distortions converge to the ones given in the theorem statement.  
\end{proof} 
\section{Superposition Coding in Multiterminal Communication}
\begin{figure}[!t]
    \begin{center}
    {\includegraphics[width=4in]{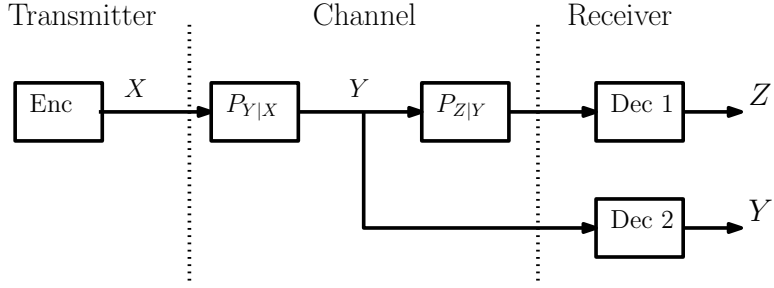}}
    \end{center}
     \caption{A schematic of
    the DBC setup.}
    \label{fig:DBC}
\end{figure}
Superposition coding is one of the key techniques in achieving the best-known inner-bounds to the achievable regions in various network communication scenarios such as the broadcast channel \cite{marton1979coding}, multiple-access channel \cite{Ahlswede}, and interference channel \cite{han1981new} among others. In this section, we demonstrate the implementation of the superposition coding technique using the proposed discretization framework for continuous sources and channels by studying communication over the degraded broadcast channel (DBC). The setup 
is shown in Figure \ref{fig:DBC} and is formally define below.

\bdefi[\textbf{Degraded Broadcast Channel}]
Given a DBC 
$(P_{Y|X}P_{Z|Y},\kappa)$, a transmission system with parameters $(n,\Theta_1,\Theta_2)$ for
reliable communication consists of an encoder mapping and decoder mappings
for 
$e:[\Theta_1]\times [\Theta_2]\rightarrow \mathbb{R}^n_i, i=1,2$, and $f_i: \mathbb{R}^n \rightarrow [\Theta_i], i=1,2
$.
A triple of rates and cost $(R_1,R_2,\tau)$ is
said to be achievable if $\forall \e>0$, and all sufficiently large
$n$, there exists a transmission system with parameters
$(n,\Theta_1,\Theta_2)$ such that for $i=1,2$, 
\[
\frac{1}{n} \log \Theta_i \geq R_i-\e, \ \ 
\frac{1}{\Theta_1\Theta_2}\sum_{j_1=1}^{\Theta_1}\sum_{j_2=1}^{\Theta_2} \kappa(e(j_1,j_2)) \leq \tau+\e,
\]
\[
\sum_{j_1=1}^{\Theta_1} \sum_{j_2=1}^{\Theta_2} \frac{1}{\Theta_1
  \Theta_2} P_{Y,Z|X}^n \left[ f_1(Y^n) \neq j_1 \text{ or } f_2(Z^n)\neq j_2 | X^n=e(j_1,j_2)) \right] \leq \e.
\]
for $i=1,2$.
Let the capacity region $\mathsf{C}(\tau)$ denotes the set of all rate pairs
$(R_1,R_2)$ such that $(R_1,R_2,\tau)$ is achievable. 
\edefi

\begin{definition}
Let $\mathcal{P}(\tau)$ denote the collection of  distributions 
$P_{UX}$  defined on $\mcq \times \mathbb{R}^2$
such that
$\mathbb{E}(\kappa(X)) \leq \tau$. 
For a $P_{QUX} \in \mathcal{P}$, let $\alpha_{SP}(P_{UX})$ denote the set
of rate pairs $(R_1,R_2) \in [0,\infty)^2$ that 
satisfy 
\begin{align*}
R_1 &\leq I(X;Y|U) \\
R_2 &\leq I(U;Z) 
\end{align*}
where the mutual information terms are evaluated  with
$P_{QX}P_{Y|X}P_{Z|Y}$. Let the superposition coding rate region be defined as 
\[
\mathsf{R}_{SP}(\tau)= \mbox{cl} \left( \bigcup_{P_{UX} \in
    \mathcal{P}} \alpha_{SP}(P_{UX}) \right).
\]
\end{definition}

\begin{theorem}
\label{thm:DBC}
The capacity  region $\mathsf{C}(\tau)$ contains the superposition coding rate region
 $\mathsf{R}_{SP}(\tau)$, i.e., $\mathsf{R}_{SP}(\tau) \subseteq \mathsf{C}(\tau)$.
\end{theorem}
\begin{proof}
Please refer to Appendix \ref{App:thm:DBC}. 
\end{proof}

\section{Distributed Source Coding}
\label{sec:DSC}
Next we consider distributed source coding problem consisting of two
correlated and memoryless 
continuous-valued sources $X$ and $Y$, characterized by a probability measure $P_{XY}$
which needs to be compressed distributively into bits to be sent to a
joint decoder.   The joint decoder wishes to reconstruct the sources with respect to 
two separable distortion measures $d_x:\mathbb{R} \rightarrow
\mathbb{R}^+$ and $d_y: \mathbb{R} \rightarrow \mathbb{R}^+$. This is
a well-studied problem \cite{berger1978multiterminal} and we skip the formal
definition for conciseness. Let $\mathcal{R}\mathcal{D}$ denote the
set of all achievable rate and distortion tuples
$(R_1,R_2,D_1,D_2)$. We denote $\mcr(D_1,D_2)$ as the set of all rates
$(R_1,R_2)$ such that $(R_1,R_2,D_1,D_2)$ is achievable. 

\begin{figure}[!htb]
    \begin{center}
    \scalebox{2.4}{\includegraphics[width=5.5cm]{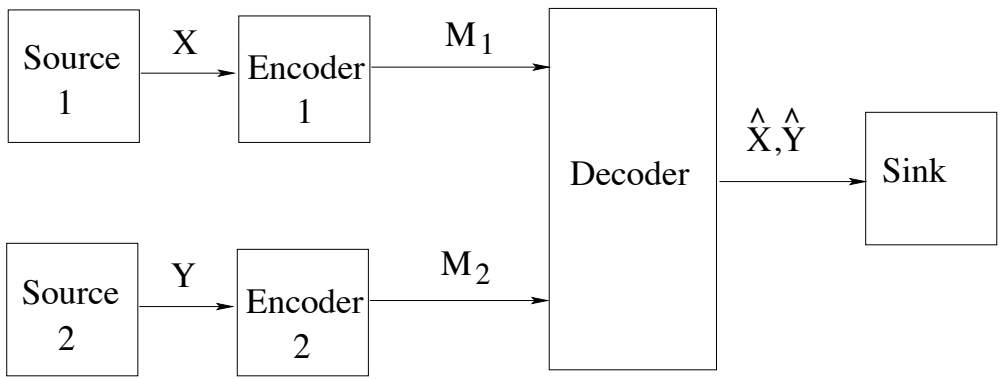}}
    \end{center}
     \caption{ General distributed source
     coding problem.}
\end{figure}

\begin{definition}
Let $\mathcal{P}(D_1,D_2)$ denote the collection of pairs of transition
probabilities $P_{U|X}$ and $P_{V|Y}$, and pairs of continuous
functions $g_i: \mathbb{R}^2 \rightarrow \mathbb{R}$ for $i=1,2$, such
that   $\eE d_x(X,\widehat{X}) \leq D_1 \ $, $\eE d_y(Y,\widehat{Y}) \leq D_2 \ $,
where the expectations are evaluated with the joint measure
$P_{XY}P_{U|X}P_{V|Y}$, i.e., $U - X -Y -V$ form a Markov chain.
For a $(P_{U|X},P_{V|Y},g_1,g_2) \in \mathcal{P}(D_1,D_2)$, let $\alpha(P_{U|X},P_{V|Y},g_1,g_2)$ denote the set
of rate pairs $(R_1,R_2) \in [0,\infty)^2$ that 
satisfy 
\begin{align}
R_1 \geq I(X;U|V), & \ \ 
R_2 \geq I(Y;V|U),  \ \ 
R_1+R_2 \geq I(XY;UV). \nn 
\end{align}
Let the information rate region be defined as \[
\mathsf{R}_{QB}(D_1,D_2) = \mbox{cl} \left( \bigcup_{(P_{U|X},P_{V|Y},g_1,g_2) \in
    \mathcal{P}(D_1,D_2)} \alpha( P_{U|X},P_{V|Y},g_1,g_2) \right)
\]
\end{definition}

\begin{theorem} For a given source $(P_{XY},d_1,d_2)$, we have $\mathsf{R}_{QB}(D_1,D_2) \subseteq \mcr(D_1,D_2)$.
\label{thm:bergertung_cont}
\end{theorem}

\begin{proof}
Please see Appendix \ref{App:bergertung_cont}.
\end{proof}

\section{Lossy Two-Help-One Problem}
\label{sec:LTO}

Next we consider a coding theorem for continuous sources for the
two-help-one problem. Consider a triple of memoryless continuous-valued sources
$(X,Y,Z)$ characterized by a probability measure $P_{XYZ}$. Let
$d:\mathbb{R}^2 \rightarrow \mathbb{R}^+$ be a jointly continuous
distortion function. The sources $X$ and $Y$ act as helpers for the
third source $Z$. The sources need to be compressed distributively
with rates $R_1$, $R_2$ and $R_3$, respectively, into bits to be sent
to a joint decoder. For simplicity we let $R_3=0$. The joint decoder
wishes to reconstruct the source $Z$ with respect to distortion
function $d$. 

\subsection{Problem Formulation and Main Result}
\bdefi
An $(n,\Theta_1,\Theta_2)$ transmission system consists of mappings 
$e_i: \mathbb{R}^n \rightarrow \{1,2,\ldots,\Theta_i\}$, for $i=1,2$,
and $f:\{1,2,\ldots,\Theta_1\} \times \{1,2,\ldots,\Theta_2\}
\rightarrow \mathbb{R}^n$. A triple $(R_1,R_2,D)$ is said to be
achievable if there exists a sequence of $(n,\Theta_{1n},\Theta_{2n})$
transmission systems  such that for $i=1,2$,
\[
\lim_{n \rightarrow \infty} \frac{\log \Theta_i}{n} \leq R_i, \ \ \ 
\lim_{n \rightarrow \infty} \mathbb{E}d_n(Z^n,f(e_1(X^n),e_2(Y^n)))\leq
D.
\]
Let $\mcr(D)$ denote the set of rates $(R_1,R_2)$ such that
$(R_1,R_2,D)$ is achievable. 
\edefi

We provide a coding theorem for  the continuous sources. 
\begin{definition}
Let $\mathcal{P}(D)$ denote the collection of
transition probabilities
$P_{QU_1V_1 U V \widehat{Z}|XY}$ 
such that  (i) $(UU_1) - (XQ) - (YQ) - (VV_1)$ form a Markov chain, 
(ii) $Q$ is independent of $(X,Y)$, (iii) $\widehat{Z}=g(U_1,V_1,U+V)$ for some function
$g$, and (iv) $\eE d(Z,\widehat{Z}) \leq D_1$,
where the expectations are evaluated with distribution  $P_{XYZ}P_{QU_1V_1UV\widehat{Z}|XY}$.
For a $P_{QU_1V_1UV\widehat{Z}|XY} \in \mathcal{P}(D)$, let $\alpha_{F}(P_{QU_1V_1UV\widehat{Z}|XY})$ denote the set
of rate pairs $(R_1,R_2) \in [0,\infty)^2$ that 
satisfy 
\begin{align}
&R_1\geq I(X;UU_1|QV_1)+I(U+V;V|QU_1V_1)-I(U;V|QU_1V_1), \nn  \\
&R_2 \geq I(Y;VV_1|QU_1)+I(U+V;U|QU_1V_1)-I(U;V|QU_1V_1) \nn \\
&R_1+R_2 \geq I(XY;UVU_1V_1|Q)+I(U+V;V|QU_1V_1), \nn \\
& \hspace{0.5in} +I(U+V;U|QU_1V_1)-I(U;V|QU_1V_1)  \nn
\end{align}
where the mutual information terms are evaluated  with
$P_{XY}P_{QU_1V_1UV\widehat{Z}|XY}$. Let the information rate region be defined as 
\[
\mathsf{R}_{F}(D) = \mbox{cl} \left(
  \bigcup_{P_{QU_1V_1UV\widehat{Z}|XY} \in
    \mathcal{P}(D)} \alpha_{F}(P_{QU_1V_1UV\widehat{Z}|XY}) \right)
\]
\end{definition}

\begin{theorem} For a given source
  $(P_{XYZ},d)$
we have $\mathsf{R}_{F}(D) \subseteq \mcr(D)$.
\label{thm:dscfinal_cont}
\end{theorem}
\textit{Proof Outline:} We propose a coding scheme involving two layers to prove the theorem. 
The first is the Berger-Tung unstructured coding layer. The second
is the structured coding layer that uses nested linear codes. First we discretize the sources and the auxilliary variables,
and apply the technique developed in source coding with side-information (Theorem \ref{thm:wynerziv_c}) to come up with a 
discrete version of the problem at hand. The Berger-Tung unstructured coding rates are derived as in Theorem \ref{thm:bergertung_cont}. The structured coding is accomplished using nested linear codes. 
The rates associated with this layer can be understood as follows: for ease of explanation,
assume that the unstructured coding auxiliary variables $U_1,V_1$ and the time-sharing variable $Q$ are
trivial. Then, for the discrete communication system, the rates
$R_1 \geq I(\widehat{X}_{n,\ell};\overline{U}_{n,\ell})+H(\overline{U}_{n,\ell}+\overline{V}_{n,\ell})-H(\overline{U}_{n,\ell}) = I(\widehat{X}_{n,\ell};\overline{U}_{n,\ell})-I(\overline{U}_{n,\ell};\overline{V}_{n,\ell})+I(\overline{V}_{n,\ell};\overline{U}_{n,\ell}+\overline{V}_{n,\ell})$,
and similarly
$R_2 \geq I(\widehat{Y}_{n,\ell'};\overline{V}_{n,\ell})-I(\overline{U}_{n,\ell};\overline{V}_{n,\ell})+I(\overline{U}_{n,\ell};\overline{U}_{n,\ell}+\overline{V}_{n,\ell})$ can be achieved using nested linear codes (over arbitrarily large prime fields) 
along with joint-typical encoding and decoding. 
Achievability follows by Theorems \ref{th:7}, \ref{th:9}, and \ref{thm:bergertung_cont} and by noting that $P(U+V\in [-\ell,\ell])\to 1$ as $\ell\to \infty$, so that the field addition for the discrete variables $\overline{U}_{n,\ell}$ and $\overline{V}_{n,\ell}$ is equal to real addition with probability approaching one as $\ell\to \infty$.

\subsection{Gaussian Lossy Two-Help-One Example}
 Consider a pair of zero-mean jointly Gaussian unit variance correlated sources
$X$ and $Y$ with correlation coefficient $\rho>0$. Let $Z=X-cY$ for
some $c$, and let $d(z,\hat{z})=(z-\hat{z})^2$. Let us evaluate a
subset of the  inner bound to the achievable rate-distortion region using a
specific test channel. Let us denote $\sigma_Z^2=1+c^2-2 \rho c$, and
let $D$ denote the target distortion. Let us choose $\mcq=\phi$, and $U_1=V_1=0$. 
Moreover consider
\[
U=X+Q_1, \ \ \mbox{ and } \ \ V=cY+Q_2,
\]
where $Q_1$ and $Q_2$ are independent zero-mean Gaussian random
variables that are independent of the pair $(X,Y)$. 
We take their  variances to be $q_1$ and $\frac{D
  \sigma_Z^2}{\sigma_Z^2-D}-q_1$.  
With this choice we see that $U+V=Z+Q_1+Q_2$, and we take 
$\widehat{Z}=\mathbb{E}(Z|U+V)=\frac{\sigma_Z^2-D}{\sigma_Z^2} (U+V)$,
which results in $\mathbb{E}d(Z,\widehat{Z})=D$. Now let us see the
achievable rates. 
\[
R_1 \geq \frac{1}{2} \log \frac{\sigma_Z^4}{q_1(\sigma_Z^2-D)}, \ \ \
\mbox{ and } 
R_2 \geq \frac{1}{2} \log \frac{\sigma_Z^4}{D \sigma_Z^2
  -q_1(\sigma_Z^2
-D)}.
\]
Eliminating $q_1$ we see that the rate distortion tuple $(R_1,R_2,D)$
satisfying the following equations is achievable:
\[
2^{-2R_1}+2^{-2R_2} \leq \left( \frac{\sigma_Z^2}{D} \right)^{-1}.
\]

Next let us see what rate and distortion are achievable with
unstructured codes, i.e., using just the first layer of the scheme. 
Let the region $\mathsf{RD}_{\text{in}}$ be defined as follows.
\begin{equation*}
\mathsf{RD}_{\text{in}} = \bigcup_{(q_1,q_2) \in \mathbb{R}^2_{+}}
\left\{ (R_1,R_2,D) \colon R_1 \geq \frac{1}{2} \log
\frac{(1+q_1)(1+q_2) - \rho^2}{q_1(1+q_2)},
\, \, R_2 \geq
\frac{1}{2} \log \frac{(1+q_1)(1+q_2)-\rho^2}{q_2(1+q_1)} \right.
\end{equation*}
\begin{equation} \label{eq:BTrateregion}
\left. R_1+R_2 \geq \frac{1}{2} \log
\frac{(1+q_1)(1+q_2)-\rho^2}{q_1 q_2}, \, \, D \geq \frac{q_1 \alpha
+ q_2 c^2 \alpha + q_1 q_2 \sigma_Z^2}{(1+q_1)(1+q_2)-\rho^2}
\right\}.
\end{equation}
where $\alpha \triangleq 1-\rho^2$. Then the rate distortion tuples $(R_1,R_2,D)$ which
belong to $\mathsf{RD}^{*}_{\text{in}}$ are achievable where $^*$
denotes convex closure. This follows by choosing $U_1=X+Q_{11}$ and 
$V_1=Y+Q_{12}$, where $Q_{11}$ and $Q_{12}$ are zero-mean independent
Gaussian random variables with variances $q_1$ and $q_2$,
respectively, and independent of the pair $(X,Y)$.

 For a given distortion $D$, the minimum sum rate $R_{\text{sum}}
 \triangleq R_1+R_2$ that lies in the region
 $\mathsf{RD}_{\text{in}}^{*}$ of Theorem \ref{thm:bergertung_cont}
 can be evaluated by direct optimization as in \cite{krithivasan2009lattices}.
 Fig.
\ref{fig:rhocrange} is a contour plot that illustrates the resulting rate-distortions in
detail. We observe that the lattice based scheme performs better than the
Berger-Tung based scheme for small distortions provided $\rho$ is
sufficiently high and $c$ lies in a certain interval. The contour labeled $R$ encloses that region in which the
pair $(\rho,c)$ should lie for the lattice binning scheme to achieve
a sum rate that is at least $R$ units less than the sum rate of the
Berger-Tung scheme for some distortion $D$. Observe that we get
improvements only for $c>0$.

\begin{figure}[htp]
\centering
\includegraphics[width = 0.6\textwidth]{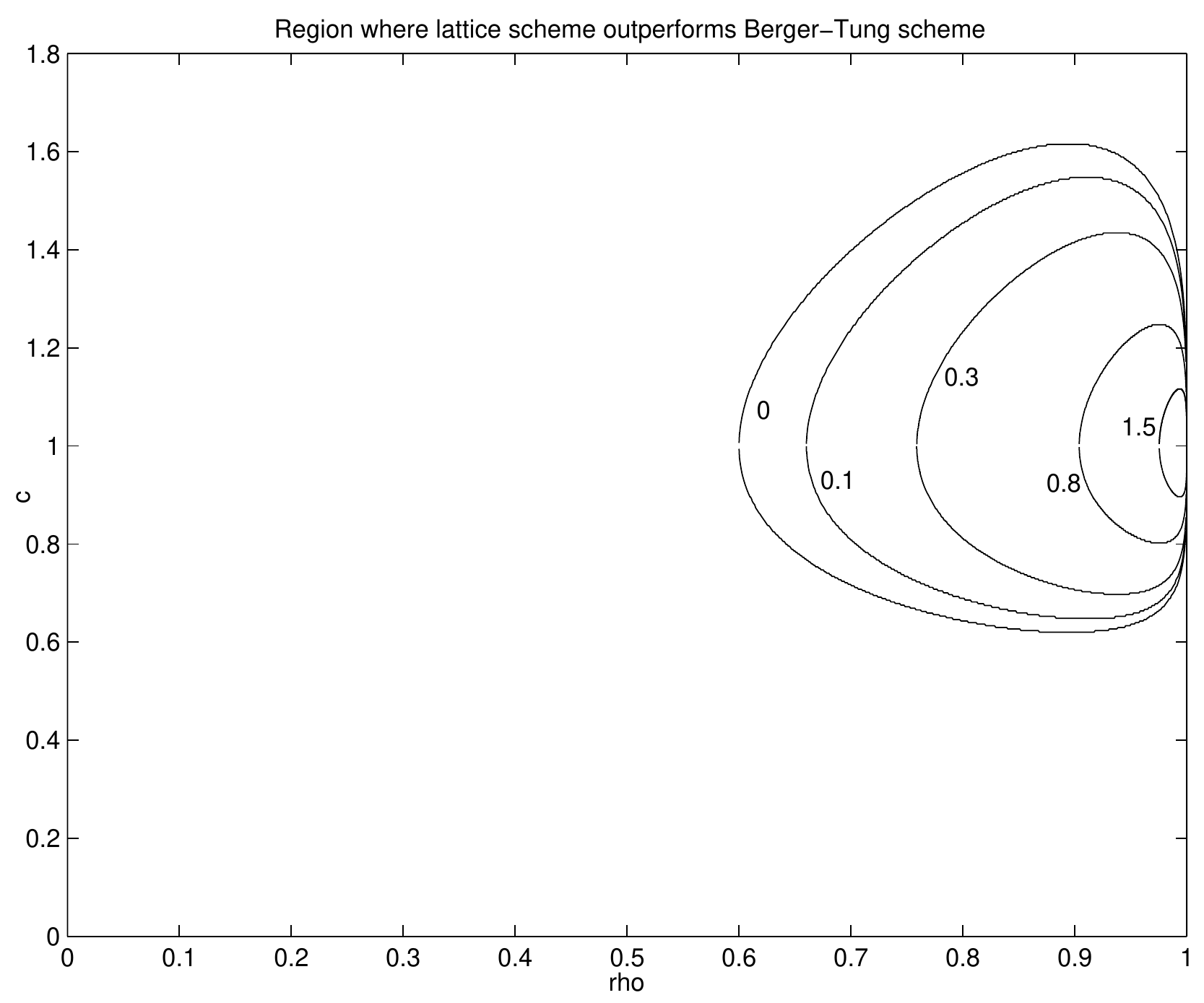}
\caption{Range of $(\rho,c)$ where the lattice scheme performs
better than the Berger Tung scheme for $D \rightarrow 0.$} \label {fig:rhocrange}
\end{figure}

\section{Computation over MAC}
\label{sec:IC}

In this section we consider a coding problem about two-transmitter multiple-access channel. We consider a simple 
formulation which is purely a channel coding problem and captures the
essence of the key concepts. In the standard formulation, the receiver wishes to recover both the messages
reliably. Now consider a variation of this problem, where the decoder
is interested in recovering only a single-letter bivariate function
$g(\cdot,\cdot)$ of the channel inputs sent
by the transmitters reliably.  
Originally, this problem was formulated in \cite{200710TIT_NazGas}, and has been studied extensively  \cite{nazer2008case,
nazer2011compute,nazer2012ergodic,sen_kim}. 
We demonstrate that structured codes can
better facilitate the interaction between the two transmitters to
ensure  that the decoder receivers the desired information while
transmitting information at a larger rate that can be sustained by
unstructured codes. 
Formally, a discrete memoryless stationary two-transmitter multiple-access
channel, used without feedback,  is given by a tuple $(P_{Y|X_1,X_2},g,\kappa_1,\kappa_2)$,   
consisting of
the transition probability $P_{Y|X_1,X_2}: \mathbb{R} \times \mathbb{R} \times \mathcal{B} \rightarrow \mathbb{R}$,
a bivariate function $g:\mathbb{R} \times \mathbb{R}
\rightarrow \mathbb{R}$, and two cost functions $\kappa_1$ and
$\kappa_2$. From now on, we will be exclusively focused on the real addition function, i.e. $g(x,y)=x+y$, for all  $x,y \in \mathbb{R}$, and drop $g$ from the tuple $(P_{Y|X_1,X_2},g,\kappa_1,\kappa_2)$. 

\subsection{Problem Formulation and Main Result}

\bdefi
Given a multiple-access channel
$(P_{Y|X_1,X_2},\kappa_1,\kappa_2)$, a transmission system with parameters $(n,\Theta_1,\Theta_2)$ for
reliable communication consists of a
pair of encoder mappings and decoder mappings
for 
$e_i: \{1,2,\ldots,\Theta_i\} \rightarrow \mathbb{R}^n_i, i=1,2$, and $f: \mathbb{R}^n \rightarrow \mathbb{R}^n.
$
A quadruple  of rates and costs $(R_1,R_2,\tau_1,\tau_2)$ is
said to be achievable if $\forall \e>0$, and all sufficiently large
$n$, there exists a transmission system with parameters
$(n,\Theta_1,\Theta_2)$ such that for $i=1,2$, 
\[
\frac{1}{n} \log \Theta_i \geq R_i-\e, \ \ 
\frac{1}{\Theta_i}\sum_{j=1}^{\Theta_i} \kappa_i(e_i(j)) \leq \tau_i+\e,
\]
\[
\sum_{j=1}^{\Theta_1} \sum_{k=1}^{\Theta_2} \frac{1}{\Theta_1
  \Theta_2} P_{Y_1,Y_2|X_1,X_2}^n \left[ f(Y^n) \neq X_1^n+X_2^n | X^n_1=e_1(j),X_2^n=e_2(k)) \right] \leq \e.
\]
for $i=1,2$.
Let the optimal capacity region $\mathsf{C}(\tau_1,\tau_2)$ denote the set of all rate pairs
$(R_1,R_2)$ such that $(R_1,R_2,\tau_1,\tau_2)$ is achievable. 
\edefi

In the following we provide an achievable rate region that is based on
structured codes.

\begin{definition}
Let $\mathcal{P}(\tau_1,\tau_2)$ denote the collection of  distributions 
$P_{QU_1U_2X_1X_2}$  defined on $\mcq \times \mathbb{R}^4$
such that  (i) $(U_1X_1) - Q - (U_2X_2)$ form a Markov
chain, with $\mathsf{Q}$ being a finite set, and (ii)
$\mathbb{E}(\kappa_i(X_i)) \leq \tau_i$. 
For a $P_{QU_1U_2X_1X_2} \in \mathcal{P}$, let $\alpha_F(P_{QU_1U_2X_1X_2})$ denote the set
of rate pairs $(R_1,R_2) \in [0,\infty)^2$ that 
satisfy 
\begin{align*}
R_1 &\leq I(U_1;Y|U_2Q)+I(X_1+X_2;Y|U_1U_2Q)-I(X_1+X_2;X_2|U_1U_2Q) \\
R_2 &\leq I(U_2;Y|U_1Q)+ I(X_1 +X_2;Y|U_1U_2Q)-I(X_1+X_2;X_1|U_1U_2Q) \\
R_1+R_2 &\leq I(U_1U_2;Y|Q)+2I(X_1\!+\!X_2;Y|U_1U_2Q)-I(X_1\!+\!X_2;X_1|U_1U_2Q)  -I(X_1\!+\!X_2;X_2|U_1U_2Q) 
\end{align*}
where the mutual information terms are evaluated  with
$P_{QU_1U_2X_1X_2}P_{Y|X_1X_2}$. Let the information rate region be defined as 
\[
\mathsf{R}_F(\tau_1,\tau_2)= \mbox{cl} \left( \bigcup_{P_{QU_1U_2X_1X_2} \in
    \mathcal{P}} \alpha_F(P_{QU_1U_2X_1X_2}) \right).
\]
\end{definition}

\begin{theorem}
\label{thm:ic_inner_2_cont}
The operational capacity cost  region $\mathsf{C}(\tau_1,\tau_2)$ contains the information
capacity region $\mathsf{R}_F(\tau_1,\tau_2)$, i.e., $\mathsf{R}_{F}(\tau_1,\tau_2) \subseteq \mathsf{C}(\tau_1,\tau_2)$.
\end{theorem}


\subsection{Gaussian MAC Example}
Consider the MAC given by $Y=X_1+X_2+Z$, where $Z$ is zero-mean Gaussian with variance $N$. We have power constraints on $X_1$ and $X_2$: $\kappa_1(x_1)=x_1^2$ and $\kappa_2(x_2)=x_2^2$, for all
$x_1,x_2 \in \mathbb{R}$. 
Let $\tau_i=P_i$ for $i=1,2$.
The rates achievable using unstructured code ensembles is given by the standard MAC capacity region given by 
\[
\left\{(R_1,R_2): R_1 \leq \frac{1}{2} \log \left(1+\frac{P_1}{N} \right), R_2 \leq \frac{1}{2} \log \left( 1+\frac{P_2}{N}\right),
R_1+R_2 \leq \frac{1}{2} \log \left(1+ \frac{P_1+P_2}{N} \right) \right\}.
\]
This is achieved using independent Gaussian inputs $X_1$ and $X_2$ of variances $P_1$ and $P_2$, respectively. 
Using the same distribution, one can achieve the following rates while employing structured code ensembles.
\[
\left\{(R_1,R_2): R_1 \leq \frac{1}{2} \log \left(\frac{P_1(P_1+P_2+N)}{(P_1+P_2)N} \right), R_2 \leq \frac{1}{2} \log \left(\frac{P_2(P_1+P_2+N)}{(P_1+P_2)N}\right)
 \right\}.
\]
Comparing the sum-rate we see that the structured coding scheme performs better than the unstructured coding scheme when 
\[
\left(1+\frac{P_1}{P_2} \right) \left(1+\frac{P_2}{P_1}\right)
\leq 1+\frac{P_1}{N}+\frac{P_2}{N}.
\]
For the case when $P_1=P_2=P$ is boils down to the condition that $\frac{P}{N} \geq 1.5$. 
\section{Interference Channels}

\subsection{Problem Formulation and Main Result}

A $3$-user interference channel consists of
three inputs $X_1,X_2$ and $X_3$, and three outputs $Y_1,Y_2$ and $Y_3$, and they are related via a 
transition probability 
$P_{Y_1Y_2Y_3|X_1X_2X_3}: \mathbb{R}^3 \times \
\sigma(\mathcal{B}^3) \rightarrow \mathbb{R}$
associated with these variables and a 
triple of cost functions: $\kappa_i:\mathbb{R}\rightarrow \mathbb{R}$
for $i \in [3]$. We assume that the channel is memoryless, stationary and used without feedback. The channel is 
characterized by the pair 
$(P_{\underline{Y}|\underline{X}},\underline{\kappa})$.

We now consider a special class of $3$-user interference channel called 3-to-1 IC that enables us to prove strict
sub-optimality of coding techniques based on unstructured codes. A
3-to-1 IC is an interference channel wherein two of the users enjoy interference
free point-to-point links. Formally, an 
interefence channel 
$(P_{\underline{Y}|\underline{X}},\underline{\kappa})$
is a 3-to-1 IC if (i)
$P_{Y_{2}|\ulineInputAlphabet}(y_{2}|\ulineinput)
\triangleq
\sum_{(y_{1},y_{3}) \in \OutputAlphabet_{1} \times
  \OutputAlphabet_{3}}W_{\ulineOutputRV|\ulineInputRV}(\ulineoutput|\ulineinput)$
is independent of $(x_{1},x_{3}) \in \InputAlphabet_{1} \times
\InputAlphabet_{3}$, and (ii)
$P_{Y_{3}|\ulineInputAlphabet}(y_{3}|\ulineinput)
\triangleq 
\sum_{(y_{1},y_{2}) \in \OutputAlphabet_{1} \times
  \OutputAlphabet_{2}}P_{\ulineOutputRV|\ulineInputRV}(\ulineoutput|\ulineinput)$
is independent of $(x_{1},x_{2}) \in \InputAlphabet_{1} \times
\InputAlphabet_{2}$ for every collection of input and output symbols
$(\ulineinput,\ulineoutput) \in \ulineInputAlphabet \times
\ulineOutputAlphabet$. For a 3-to-1 IC, the channel transition
probabilities factorize
as \[P_{\ulineOutputRV|\ulineInputRV}(\ulineoutput|\ulineinput)
=P_{\OutputRV_{1}|\ulineInputRV}(y_{1}|\ulineinput)
P_{\OutputRV_{2}|\InputRV_{2}}(y_{2}|x_{2})W_{\OutputRV_{3}|\InputRV_{3}}(y_{3}|x_{3}),\]
for some transition probabilities
$P_{\OutputRV_{1}|\ulineInputRV}$, $P_{\OutputRV_{2}|\InputRV_{2}}$
and $P_{\OutputRV_{3}|\InputRV_{3}}$. We also note that
$(X_{1},X_{3})-X_{2}-Y_{2}$ and $(X_{1},X_{2})-X_{3}-Y_{3}$ form Markov
chains for any distribution
$P_{X_{1}}P_{X_{2}}P_{X_{3}}P_{\ulineOutputRV|\ulineInputRV}$.
One can define the operational capacity-cost region
$\mathsf{C}(\underline{\tau})$ for this channel in a
straightforward way.

\bdefi
Given a 3-to-1 IC
$(P_{\underline{Y}|\underline{X}},\underline{\kappa})$, a transmission system with parameters $(n,\underline{\Theta})$ for
reliable communication consists of a
pair of encoder mappings and decoder mappings
for 
$e_i: [\Theta_i] \rightarrow \mathbb{R}^n_i, i \in [3]$, and $f_i: \mathbb{R}^n \rightarrow [\Theta_i] .
$
A tuple  of rates and costs $(\underline{R},\underline{\tau})$ is
said to be achievable if $\forall \e>0$, and all sufficiently large
$n$, there exists a transmission system with parameters
$(n,\underline{\Theta})$ such that for $i\in[3]$, 
\[
\frac{1}{n} \log \Theta_i \geq R_i-\e, \ \ 
\frac{1}{\Theta_i}\sum_{j=1}^{\Theta_i} \kappa_i(e_i(j)) \leq \tau_i+\e,
\]
\[
\sum_{j=1}^{\Theta_1} \sum_{k=1}^{\Theta_2} \sum_{l=1}^{\Theta_3} \frac{1}{\Theta_1
  \Theta_2 \Theta_3} P_{Y_1,Y_2,Y_3|X_1,X_2,X_3}^n \left[ \cup_{i=1}^3 f_i(Y_i^n) \neq X_i^n | X^n_1=e_1(j),X_2^n=e_2(k),X_3^n=e_3(l)) \right] \leq \e.
\]
for $i=1,2$.
Let the optimal capacity region $\mathsf{C}(\underline{\tau})$ denote the set of all rate triple
$\underline{R}$ such that $(\underline{R},\underline{\tau})$ is achievable. 
\edefi

Any  interference channel wherein only one of the users is subjected to
  interference is a $3-$to$-1$ IC by a suitable permutation of the
  user indices.
In the following section, we provide an inner bound to the capacity region by using structured code ensembles.

\begin{definition}
 \label{Eqn:TestChannelsCodingOver3To1ICUsingNestedCosetCodes}
Given a $3-$IC $(W_{\ulineOutputRV|\ulineInputRV},\ulinecostfn)$,
let $\mathcal{P}_{f}(\underline{\tau)})$
denote the collection of
 distributions 
$P_{\TimeSharingRV\SemiPrivateRV_{2}\SemiPrivateRV_{3}\ulineInputRV}$ where 
$P_{\TimeSharingRV\SemiPrivateRV_{2}\SemiPrivateRV_{3}\ulineInputRV} \in \mathcal{P}_{u}(\underline{\tau})$
defined over $\TimeSharingRVSet \times
\mathbb{R}^4$, such
that $(U_1,X_1)$, $(U_2,X_2)$ and $X_3$ are conditionally mutually independent given $Q$, where $\mathcal{Q}$ is a finite set.  
For 
$(P_{\TimeSharingRV\SemiPrivateRV_{2}\SemiPrivateRV_{3}\ulineInputRV}) \in
\mathcal{P}_{f}(\underline{\tau})$, let
$\alpha_F(P_{\TimeSharingRV\SemiPrivateRV_{2}\SemiPrivateRV_{3}
\ulineInputRV})$ be defined as the set of rate triples $(R_{1},R_{2},R_{3})
\in [0,\infty)^{3}$ that satisfy
\begin{eqnarray}
R_1 &\leq&  I(X_1;Y_1|U_{2} + U_{3},\TimeSharingRV ),\nonumber\\
R_1 &\leq&  
 I(X_1, U_{2} + U_{3};Y_1|\TimeSharingRV )-I(U_2+U_3;U_j|Q) :j=2,3 \nonumber\\R_j &\leq&
I(U_{j},X_j;Y_j|\TimeSharingRV ):j=2,3,   \nonumber\\
R_1+R_j &\leq& I(X_j;Y_j|\TimeSharingRV, U_{j}) +I(X_1,U_{2}+
U_{3};Y_1|\TimeSharingRV )
-I(U_2+U_3;U_{\bar{j}}|Q)
:j=2,3,\nonumber
\end{eqnarray}
where the mutual information terms are evaluated with
$P_{\TimeSharingRV\SemiPrivateRV_{2}\SemiPrivateRV_{3}\ulineInputRV}W_{\underline{Y}|\underline{X}}$, and 
$\bar{j}=5-j$.  Let
the information rate-cost region be defined as \begin{equation}
 \label{Eqn:AchievableRateRegion3to1ICUsingNestedCosetCodes}
\mathsf{R}_{F}(\underline{\tau}) \triangleq  cl \left(
\underset{\substack{P_{\TimeSharingRV\SemiPrivateRV_{2}\SemiPrivateRV_{3}
\ulineInputRV} \in \mathcal{P}_{f}(\underline{\tau})}
}{\bigcup}\alpha_{F}(P_{\TimeSharingRV\SemiPrivateRV_{2}\SemiPrivateRV_{3}
\ulineInputRV} \right).\nonumber
\end{equation}
\end{definition}
\begin{theorem}
\label{Thm:AchievableRateRegionFor3To1ICUsingCosetCodes}
For a given $3-$IC,  $(P_{\ulineOutputRV|\ulineInputRV},\underline{\kappa})$, the operational capacity-cost  region $\mathsf{C}(\underline{\tau})$
contains the information rate-cost region,  
$\mathsf{R}_{F}(\underline{\tau})$, 
i.e., $\mathsf{R}_{F}(\underline{\tau}) \subseteq \mathsf{C}(\underline{\tau})$.
\end{theorem}

\subsection{Gaussian Example}

Consider the linear quadratic Gaussian example where the transition probability of the interference channel is characterized in a canonical form 
with a symmetry about user 2 and user 3 as: $Y_1=X_1+aX_2+aX_3+Z_1$,
$Y_2=X_2+Z_2$, and $Y_3=X_3+Z_3$, where 
$Z_i$ is zero-mean Gaussian noise with variance given by $N_i$, and $a$ is a free parameter. For ease of comparison we choose $N-1=N_2=N_3=1$. The cost constraints are given by $\kappa_i(x)=x^2$ for $i\in[3]$. We choose
cost levels at $\tau_i=P$ for $i \in[3]$. We will also focus on the operating point where users $2$ and $3$ wish to attain their respective PtP capacities, i.e., 
$R_2=R_3=\frac{1}{2}\log (1+P)$. 
Now let us find the capacity of the first user. We choose $Q$ to be the trivial random variable, and let 
$U_i=X_i$, for $i=2,3$.  
Since users $2$ and $3$ wish to attain their respective capacities, we have to choose $X_i$ to be 
zero-mean Gaussian with variance $P$ for 
$i=2,3$. Now the bounds on $R_1$ are given by
\begin{align}
R_1 &\leq \frac{1}{2}\log (1+P)  +
\min\left\{0,  \frac{1}{2}\log \left[\frac{1}{2}+ \frac{aP}{P+1} \right]
-\frac{1}{2}\log (1+P)\right\}.
\end{align}
Now let us find the condition under which even user 1 can attain the respective capacity as if the channel is interference-free. This requires 
\[
\frac{P+1}{a^2} \leq \frac{2P}{1+2P}.
\]
Solution exists for $a>1.714$. 
For the unstructured coding scheme, the first receiver has to decode $X_2$ and $X_3$ before decoding $X_1$. Using the standard results on the multiple-access channels, we see that 
following conditions need to be satisfied.
\begin{align}
\frac{1}{2} \log (P+1) &\leq \frac{1}{2} \log (a^2P+1) \\
 \log (P+1) &\leq \frac{1}{2} \log \min\{ (2a^2P+1), (P+a^2P+1)\} \\
\frac{3}{2} \log (P+1) &\leq \frac{1}{2} \log (a^2P+P+1). 
\end{align}
Solutions exists for $a>1$. The range of $a$ and $P$ for which solutions exists are shown in Figure \ref{fig:IC-Gaussian} for the two cases. 
One can clearly see that structured coding offers larger range of $(a,P)$ such that all users attain their respective interference-free capacities.
\begin{figure}[htp]
\centering
\includegraphics[width = 0.45\textwidth]{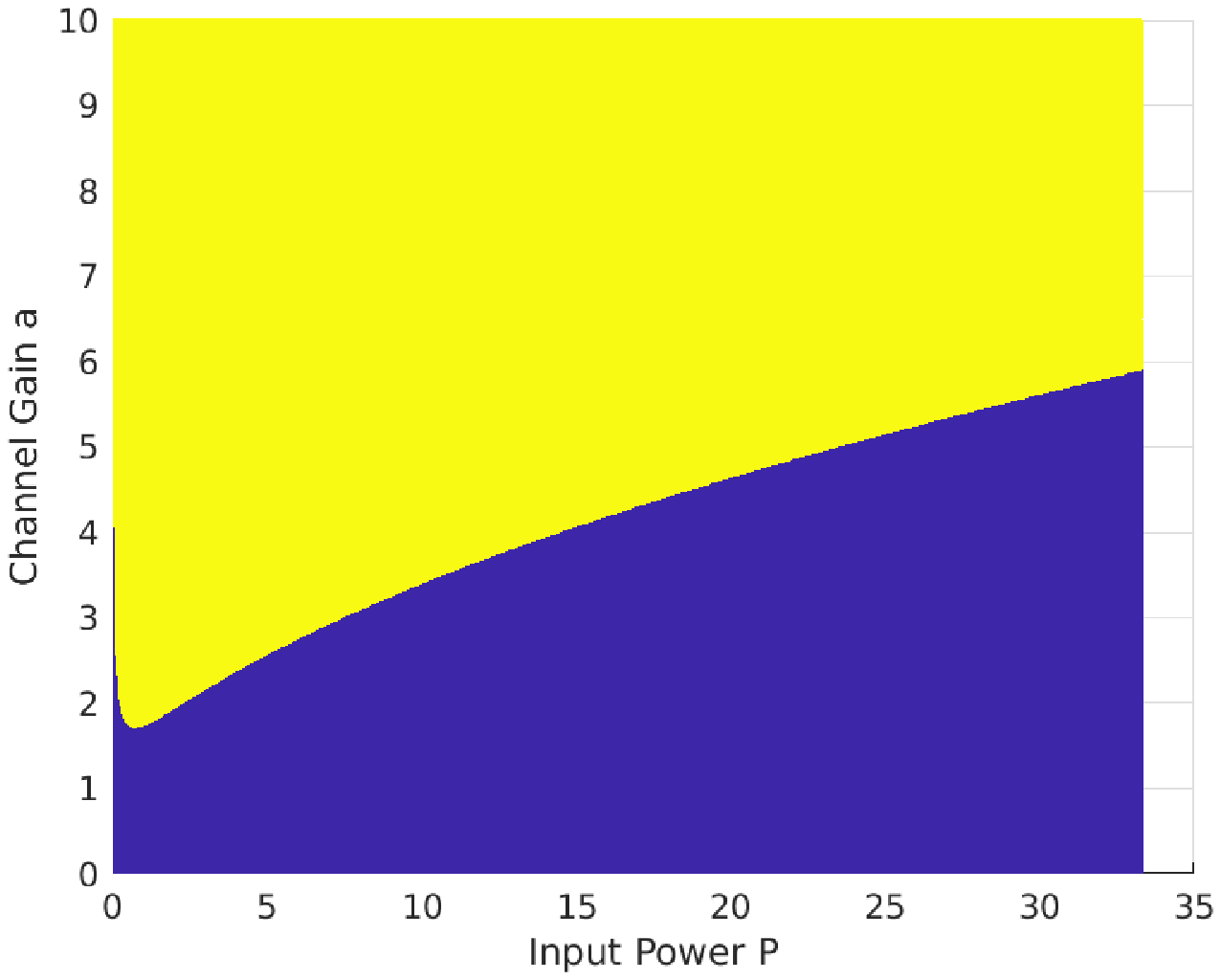}
\includegraphics[width = 0.45\textwidth]{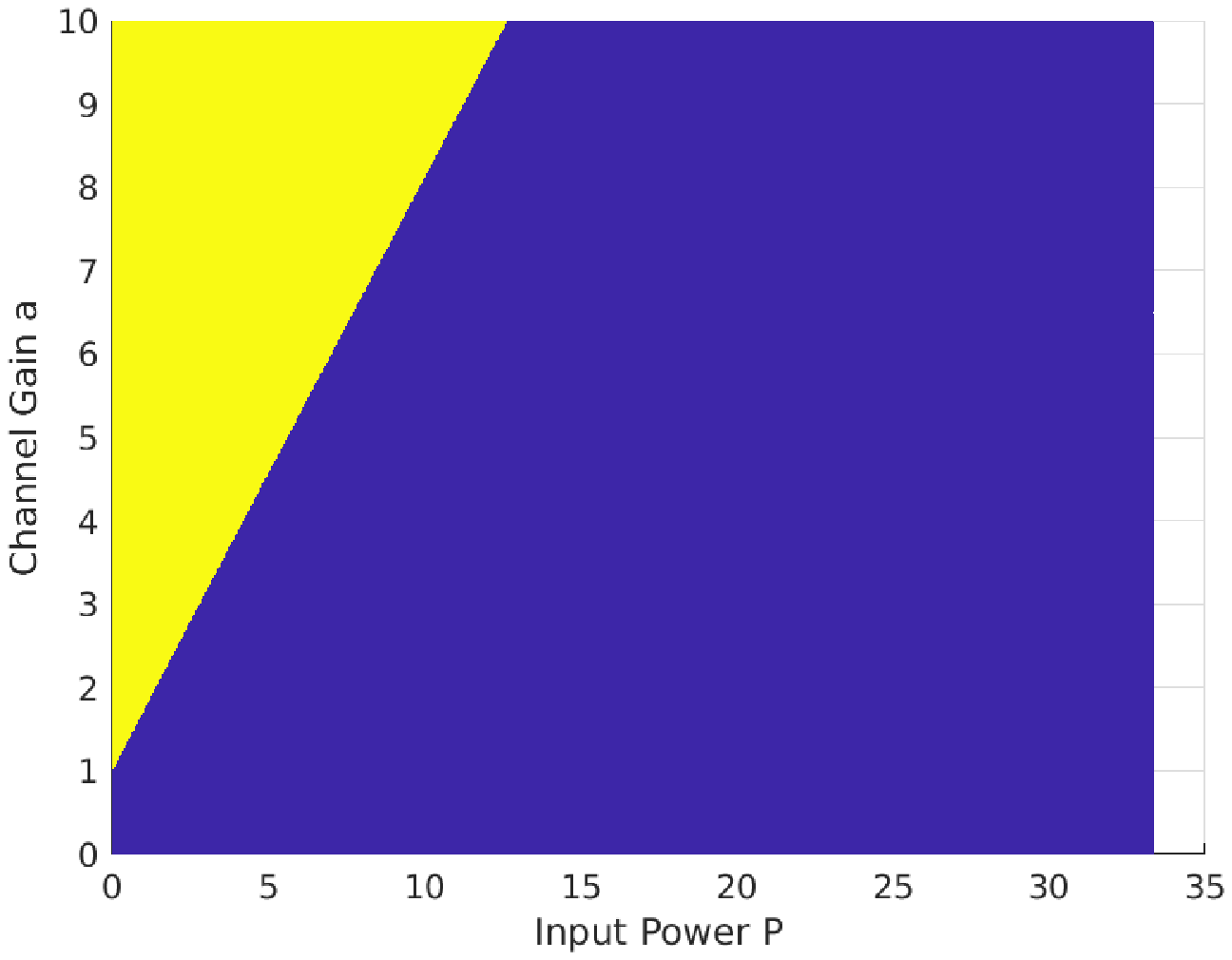}
\centering
\hspace{2in} (a) \hspace{2in} (b) 
\caption{Range of $(a,P)$ where all users can attain their respective interference-free capacities is shown in Maize.
(a) Structure Coding Scheme 
(b) Unstructured Coding Scheme.} 
\label {fig:IC-Gaussian}
\end{figure}

\section{Multiple Descriptions Source Coding}
\label{sec:MD}
In this section, we consider the multiple descriptions source coding problem, where given a source $X$, the encoder wishes to construct a set of $\ell\geq 2$ descriptions of the source, such that given each subset of descriptions, the source can be reconstructed with a desired distortion. The scenario has been studied extensively in the discrete case \cite{zhang1987new, el1980multiple,akyol2012combinatorial,pradhan2004n, shirani2018achievable}. We derive an achievable rate-distortion for general continuous sources, and demonstrate that structured codes achieve a larger rate-distortion region compared to unstructured codes in an example with Gaussian sources and test-channels. 
\subsection{Problem Formulation and Main Result}

\bdefi
Let $\ell\geq 2$, $\mathsf{L}=[\ell]$,  $\mathcal{L}=2^\mathsf{L}-\Phi$, $n \in \mathbb{N}$, and $\Theta_i\in \mathbb{N}, i\in \mathsf{L}$. A coding system with parameters $(n,\Theta_i:i \in \mathsf{L})$ for 
multiple description  coding of a given source 
$(P_X,d_\mcn: \mcn \in \mathcal{L})$, 
consists of 
$\ell$ encoder mappings and $2^\ell-1$ decoder mappings: 
\[
e_i: \mcx^n \rightarrow [\Theta_i], \ \ \ \  
f_{\mcn}: \prod_{i \in \mcn}[\Theta_i] \rightarrow \mcx^n,
\]
where $i\in \mathsf{L}$, and $\mcn \in \mathcal{L}$.
A of rate-distortion tuple $(R_i:i \in \mathsf{L},D_\mcn:\mcn \in \mathcal{L}) \in 
(\mathbb{R}^+)^{\ell+2^{\ell}-1}$ is
said to be achievable if for all $\e>0$,  and for all sufficiently large
$n$, there exists a coding system with parameters 
$(n,\Theta_i:i\in \mathsf{L})$ such that

\[
\frac{1}{n} \log \Theta_i \leq R_i+\e, \ \ \mathbb{E}[d_\mcn(X^n,f_{\mcn}((e_i(X^n))_{i \in \mcn}))]  \leq D_\mcn+\e, \qquad i\in \mathsf{L},  \mcn \in \mathcal{L}.
\]
The operational  rate-distortion region 
 $\mathsf{R}_{op}(D_\mcn:\mcn \in \mathcal{L})$ is given by the set of all achievable rate-distortion tuples 
$(R_i:i \in \mathsf{L},D_\mcn:\mcn \in \mathcal{L})$.
\edefi

We use the discretization techniques developed in prior sections, along with the Sperner Set Coding (SSC) strategy in \cite{shirani2018achievable} to derive an achievable rate-distortion region for multiple descriptions coding with continuous sources. 
We explain the random variables decoded at each decoder for the SSC strategy for the three-descriptions problem, i.e. $\ell=3$. The complete explanation of the scheme is provided in \cite{shirani2018achievable}. Let $\mathsf{S}_{L}$ be the Sperner set for
$\ell=3$. It is known that $\mathbf{S}_\mathsf{L}$ has 17 elements. Let $U_{\mathcal{M}},\mathcal{M}\in \mathbf{S}_{\mathsf{L}}$ be a vector of random variables with whose joint distribution with source $X$ is given by $P_{X,U_{\mathcal{M}}, \mathcal{M}\in \mathbf{S}_{\mathsf{L}}}$.
The SSC strategy generates 17 independent random codebooks $\mathsf{C}_{\mathcal{M}},\mathcal{M}\in \mathbf{S}_{\mathsf{L}}$, with blocklength $n$, where $\mathsf{C}_{\mathcal{M}}$ is generated based on the single letter distribution $P_{U_{\mathcal{M}}}$ for $\mathcal{M}\in \mathbf{S}_{\mathsf{L}}$. Each codebook is binned independently $\ell$ times, once per description. Given a source sequence $X^n$, the encoder finds codewords $U^{n}_{\mathcal{M}}\in \mathsf{C}_{\mathcal{M}}, \mathcal{M}\in \mathbf{S}_{\mathsf{L}}$ which are jointly typical with each other and the source sequence and sends the corresponding bin numbers on each description. 
Each decoder decodes a subset of the codewords. To elaborate, decoder 
$\mathsf{N}\in 2^{\mathsf{L}}-\Phi$ recovers $U_{\mathcal{M}}$ based on the received bin numbers if $\mathsf{N}\subseteq \mathcal{M}$.  
Let  ${\mathbf{M}}_\mathsf{N}$ be the set of indices of random variables whose corresponding codewords are decoded at decoder $\mathsf{N}$, and let $\widetilde{{\mathbf{M}}}_\mathsf{N}$ be the indices of those which are decodable if we have access to strict subsets of the descriptions received by $\mathsf{N}$.  For instance, the random variables decoded at decoders $\{1\}$ and $\{2,3\}$ are given below, where we have removed the outer most curly brackets for ease of exposition: 
\begin{align*}
 &\text{decoder $\{1\}$: } U_{\{1\},\{2\},\{3\}}, U_{\{1\},\{2\}}, U_{\{1\},\{3\}}, U_{\{1\},\{2,3\}}, U_{\{1\}}\\
&\text{decoder $\{2,3\}$: } U_{\{1\},\{2\},\{3\}}, U_{\{1,2\},\{1,3\},\{2,3\}}, U_{\{1\},\{2\}}, U_{\{1\},\{3\}}, U_{\{2\},\{3\}},\\&\qquad\quad\qquad U_{\{1\},\{2,3\}}, U_{\{2\},\{1,3\}} U_{\{3\},\{1,2\}}, U_{\{1,2\},\{2,3\}}, U_{\{1,3\},\{2,3\}}, U_{\{2\}}, U_{\{3\}}, U_{\{2,3\}}
\end{align*}
So, as an example ${\mathbf{M}}_{\{1\}}=\Bigg\{{\Big\{\{1\},\{2\},\{3\}}\Big\}, \Big\{\{1\},\{2\}\Big\}, \Big\{\{1\},\{3\}\Big\}, \Big\{\{1\},\{2,3\}\Big\}, \Big\{\{1\}\Big\}\Bigg\}$ which are all the codebooks decoded at decoder $\{1\}$.
 Also, $\widetilde{\mathbf{M}}_{\{2,3\}}=\Bigg\{\Big\{\{1\},\{2\},\{3\}\Big\}$, $ \Big\{\{1\},\{2\}\Big\}$, $\Big\{\{1\},\{3\}\Big\}$, $\Big\{\{2\},\{3\}\Big\}$, $\Big\{\{2\},\{1,3\}\Big\}$, $\Big\{\{3\},\{1,2\}\Big\}$, $\Big\{\{2\}\Big\}$, $\Big\{\{3\}\Big\}\Bigg\}$,  and these are all of the indices of random variables whose corresponding codewords are decoded at decoders $\{2\}$ and $\{3\}$.
 The following theorem provides an achievable region for the multiple descriptions problem using the discretization process developed in the previous sections along with the SSC strategy with unstructured random codes developed in \cite{shirani2018achievable} for discrete sources and test-channels.

\begin{definition}
Given a source $(P_X,d_\mcn:\mcn \in \mathcal{L})$,  let $\mathcal{P}(D_\mcn:\mcn \in \mathcal{L})$ denote the collection of pairs $(P,g_\mathcal{L})$ of (a)   joint distribution $P$ on random variables $X$ and $U_{\mathcal{M}},\mathcal{M}\in \mathbf{S}_{\mathsf{L}}$ with $X$-marginal distribution $P_X$ and (b) a set of reconstruction functions $g_{\mathcal{L}}\triangleq (g_\mathsf{N}:\mathsf{U}_{\{\mathsf{N}\}}\to\mathsf{X}, \mathsf{N}\in\mathcal{L})$ such that 
$\mathbb{E}d_{\mathsf{N}}(X,g_{\mathsf{N}}(U_{\{\mathsf{N}\}})) \leq D_\mcn$,
$\forall \mcn \in \mathcal{L}$, 
where the expectations are evaluated with the distribution $P$.  For a $(P,g_{\mathcal{L}}) 
\in 
\mathcal{P}(D_\mcn:\mcn \in 
\mathcal{L})$,
define  $\alpha_{SS}(P, g_\mathcal{L}) $  as the set of rate tuples
$(R_i:i \in \mathsf{L})$ satisfying the following constraints for some non-negative real numbers $(\rho_{\mathcal{M},i},r_{\mathcal{M}})_{i\in \widetilde{\mathcal{M}},\mathcal{M}\in \mathbf{S}_\mathsf{L} }$ :  
\begin{align}
\label{cov11} \overline{I}(U_{{\mathbf{M}}})+ I(U_{{\mathbf{M}}};X)    &\leq   \sum_{\mathcal{M}\in      \mathbf{M}}r_{\mathcal{M}}   ,\forall \ \ \mathbf{M}\subset \mathbf{S}_\mathsf{L},
\\\label{pack11} \sum_{\mathcal{M}\in \mathbf{M}_\mathsf{N}\backslash({\mathbf{L}\cup \widetilde{\mathbf{M}}_\mathsf{N}})} \!\!\!\! \!\!\!\!(r_\mathcal{M}-\sum_{i\in \widetilde{\mathcal{M}}}\rho_{\mathcal{M},i})
&\leq \overline{I}(U_{\mathbf{M}_\mathsf{N}\backslash({\mathbf{L}\cup \widetilde{\mathbf{M}}})})+I(U_{_{\mathbf{M}_\mathsf{N}\backslash({\mathbf{L}\cup \widetilde{\mathbf{M}}})}}; U_{\mathbf{L}\cup \widetilde{\mathbf{M}}})  , \forall \mathbf{L}\subset \mathbf{M}_\mathsf{N}, \forall \mathsf{N}\in \mathcal{L},
\\R_i &=\sum_{\mathcal{M}} \rho_{\mathcal{M},i},  \label{eq:RDSSC}
\end{align}
 where we have defined $\overline{I}(Z^k)\triangleq \sum_{j=1}^k I(Z_k;Z^{k-1})$ for a random vector $Z^k$, 
 $\mathbf{M}_\mathsf{N}$ is the set of all codebooks decoded at decoder $\mathsf{N}$, that is $\mathbf{M}_\mathsf{N} \triangleq  \{\mathcal{M}\in \mathbf{S}_\mathsf{L}| \exists \mathsf{N}'\subset \mathsf{N},  \mathsf{N}'\in \mathcal{M}\}$, and $\widetilde{\mathbf{M}}_\mathsf{N}$ denotes the set of all codebooks decoded at decoders $\mathsf{N}_p\subsetneq \mathsf{N}$ which receive subsets of descriptions received by $\mathsf{N}$, that is  
$\widetilde{\mathbf{M}}_\mathsf{N} \triangleq  \bigcup_{\mathsf{N}_p \subsetneq \mathsf{N}}{\mathbf{M}}_{\mathsf{N}_p}$. The mutual information terms are evaluated with the distribution $P$. 
 Define the Sperner Set Coding rate-distortion region as 
\begin{align*}
 \mathsf{R}_{SS}(D_\mcn:\mcn \in \mathcal{L})\triangleq cl \left( \underset{{(P,g_\mathcal{L}) \in \mathcal{P}(D_\mcn:\mcn \in \mathcal{L})}}\bigcup \ \  \alpha_{SS}(P, g_{\mathcal{L}}) \right).
\end{align*}
\end{definition}

\begin{theorem} \label{thm:SSC}
\label{thm:RDMD}
Given a source $(P_X,d_\mcn:\mcn \in \mathcal{L})$, the operational rate-distortion region contains the information rate-distortion region, i.e., 
$\mathsf{R}_{SS}(D_\mcn:\mcn \in \mathcal{L}) \subseteq \mathsf{R}_{op}(D_\mcn:\mcn \in \mathcal{L})$.
\end{theorem}

\textit{Proof Outline.} Given the random variables $X$ and $U_{\mathcal{M}},\mathcal{M}\in \mathbf{S}_{\mathsf{L}}$ described in the theorem statement,   the transmission system first discretizes the source using techniques developed in the prior sections and then uses the  
discrete SSC strategy introduced in Theorem 3 \cite{shirani2018achievable} to achieve the rate-distortion vector in \eqref{cov11}, \eqref{pack11}, and \eqref{eq:RDSSC}. The mutual-information terms in 
\eqref{cov11}, \eqref{pack11} for the discretized variables converge to that of the continuous variables as the clipping limits are increased asymptotically and the quantization step approaches zero by similar arguments as in the prior sections. 

 The following theorem provides an achievable region for the multiple descriptions problem using the discretization process developed in the previous sections along with the SSC strategy with both unstructured and structured random codes developed in \cite{shirani2018achievable} for discrete sources and test-channels. 
  
\begin{definition}
 Given a source 
 $(P_X,d_\mcn:\mcn \in \mathcal{L})$, let 
 $\mathcal{P}(D_\mcn:\mcn \in \mathcal{L})$ denote the collection of pairs $(P,g_{\mathcal{L}})$ of (a) 
 joint distribution $P$
 on random variables 
 $X$, $U_{\mathcal{M}},\mathcal{M}\in \mathbf{S}_\mathsf{L}$, $V_{ \mathcal{A}_{in}}, V_{ \mathcal{A}_{out}}, V_{ \mathcal{A}_{sum}}$, 
 with $X$-marginal distribution $P_X$, 
 where  $\mathcal{A}_j\in {{\mathbf{S}}}_\mathsf{L},j \in \{in,out,sum\}$, are three distinct families, and 
all the auxiliary random variables take values in  $\mathbb{R}$, 
 and (b) a set of reconstruction functions $g_{\mathcal{L}}=\{g_\mathsf{N}:\mathsf{U}_{\{\mathsf{N}\}}\to\mathsf{X}, \mathsf{N}\in\mathcal{L}\}$, 
 such that 
 $V_{\mathcal{A}_{sum}}=V_{\mathcal{A}_{in}}+V_{\mathcal{A}_{out}}$, and 
 $\mathbb{E}d_\mcn (X,g_{\mcn}(U_{\{\mcn\}})) \leq D_{\mcn}$ $\forall \mcn \in \mathcal{L}$, where the expectations are evaluated with $P$.  For a $(P,g_{\mathcal{L}})
 \in \mathcal{P}(D_{\mcn}:\mcn \in \mathcal{L})$, 
 define $\alpha_F(P,g_{\mathcal{L}})$ as the set of rate tuple $(R_i,i\in \mathsf{L})$ satisfying the following constraints 
  for some non-negative real numbers  $(\rho_{\mathcal{M},i},r_{o,\mathcal{M}})_{i\in \widetilde{\mathcal{M}},\mathcal{M}\in {\mathbf{S}_\mathsf{L}} }$, $r'_{ \mathcal{A}_{in}}, \rho_{ \mathcal{A}_{in},i},  i\in \widetilde{\mathcal{A}}_{in}$, $r'_{ \mathcal{A}_{out}}, \rho_{ \mathcal{A}_{out},i},  i\in \widetilde{\mathcal{A}}_{out}$: 
  
\noindent i) Covering Constraints:
for all $ \mathbf{M} \subset \mathbf{S}_\mathsf{L}, \mathbf{E}\subset \mathbf{A}$ and $\alpha,\beta\in \mathbb{F}_p^+$,
\begin{align}
 \label{sec1cov1}
 \sum_{\mathcal{M} \in \mathbf{M}}r_{\mathcal{M}}+\sum_{\mathcal{E} \in \mathbf{E}}{r'_{\mathcal{E}}}\geq  & \overline{I}(U_{{\mathbf{M}}},V_{\mathbf{E}})+ I(U_{{\mathbf{M}}}V_{\mathbf{E}};X)  , \quad  \\
\label{sec1cov2}
 \sum_{\mathcal{M} \in \mathbf{M}}r_{\mathcal{M}}+r'_{\mathcal{A}_{out}}\geq &\overline{I}(U_{{\mathbf{M}}},V_{\mathcal{A}_{out}})+ I(U_{{\mathbf{M}}}W_{\mathcal{A}_{out},\alpha,\beta};X) -I(W_{\mathcal{A}_{sum},\alpha,\beta}; V_{\mathcal{A}_{in}}|U_{\mathbf{M}})
 \nn\\& 
 +I(V_{\mathcal{A}_{in}};V_{\mathcal{A}_{out}}|U_{\mathbf{M}}).
 \end{align}
 ii) Packing constraints: 
 for all $ \overline{\mathbf{L}}\subset {\overline{\mathbf{M}}}_\mathsf{N},
 \mathcal{A}_{sum}\notin {\mathbf{M}}_\mathsf{N}$,
 \begin{align}
 \label{sec1pack1}
 \hspace{-0.3in} \sum_{\mathcal{M}\in {{\mathbf{M}}}_\mathsf{N}\backslash\widetilde{\mathbf{M}}_\mathsf{N}\cup {\mathbf{L}}}\!\!\!\!\!\!\! r_{\mathcal{M}} -\!\! \sum_{j\in \widetilde{\mathcal{M}}} \rho_{\mathcal{M},j}&
 +  \sum_{\substack{\mathcal{E}\in {\mathbf{M}}_\mathsf{N}\backslash\widetilde{\mathbf{M}}_\mathsf{N}\cup \overline{\mathbf{L}}\\\bigcap\{\mathcal{A}_i|i\in\{in,out\}\} }}\!\!\!\!\!\!\!\!\!\!r'_{\mathcal{E}}-\sum_{j\in \widetilde{\mathcal{E}}}\rho_{o,\mathcal{E},j}, 
 \leq \nn
\overline{I}([UVW]_{\overline{{\mathbf{M}}}_\mathsf{N}\backslash\widetilde{\mathbf{M}}_\mathsf{N}\cup {\mathbf{L}}})
\\& \qquad\qquad \qquad \qquad \qquad \qquad+
I({[UVW]_{\overline{{\mathbf{M}}}_\mathsf{N}\backslash\widetilde{\mathbf{M}}_\mathsf{N}\cup {\mathbf{L}}}};{[UVW]_{\widehat{\mathbf{M}}_\mathsf{N}\cup \overline{\mathbf{L}}}}),
 \end{align}
and  for all $\overline{\mathbf{L}}\subset {\overline{\mathbf{M}}}_\mathsf{N},  \mathcal{A}_{sum}\in {\mathbf{M}}_\mathsf{N}, \mathcal{A}_{in}\notin {\mathbf{M}}_\mathsf{N} ,\mathcal{A}_{out}\notin {\mathbf{M}}_\mathsf{N}$
 \begin{align}
  \label{sec1pack2}
\hspace{-0.3in} \sum_{\mathcal{M}\in {{\mathbf{M}}}_\mathsf{N}\backslash\widetilde{\mathbf{M}}_\mathsf{N}\cup {\mathbf{L}}}\!\!\!\!\!\!\! (r_{\mathcal{M}} -\!\! \sum_{j\in \widetilde{\mathcal{M}}}& \rho_{\mathcal{M},j})
+ r'_{{\mathcal{A}}_{out}}-\sum_{j\in \widetilde{\mathcal{A}}_{sum}}\rho_{o,\mathcal{A}_{sum},j}\leq \nn\overline{I}(U_{\mathbf{M}_\mathsf{N}},V_{\mathcal{A}_{out}}) + I(U_{\mathbf{M}_\mathsf{N}}W_{\mathcal{A}_{sum},1,1}; [UVW]_{\widehat{\mathbf{M}}_\mathsf{N}\cup \overline{\mathbf{L}}})
\\&\qquad \qquad \qquad- I(W_{\mathcal{A}_{sum},1,1};V_{in}|U_{\mathbf{M}_\mathsf{N}})
+ I(V_{in};V_{out}|U_{\mathbf{M}_\mathsf{N}})
, 
\end{align}
where (a) $R_i =\sum_{\mathcal{M}} \rho_{\mathcal{M},i}$, (b) $\mathbf{A} \triangleq \{\mathcal{A}_{in},\mathcal{A}_{out}\}$, (c) $\overline{{\mathbf{M}}}_\mathsf{N}\triangleq({\mathbf{M}}_\mathsf{N}, \{ \mathcal{A}_j, j\in \{in,out\}|\mathcal{A}_j\in \mathbf{M}_\mathsf{N} \},\{(\mathcal{A}_{sum}, 1,1)| \mathcal{A}_{sum}\in \mathbf{M}_\mathsf{N}\})$, (d) $\widehat{\mathbf{M}}_\mathsf{N}\triangleq \bigcup_{\mathsf{N}'\subsetneq \mathsf{N}}{\overline{\mathbf{M}}}_{\mathsf{N}'}$, (e) $r'_{\mathcal{A}_{in}}\leq r'_{\mathcal{A}_{out}}
$, and (f) $W_{\mathcal{A}_3, \alpha, \beta}\triangleq\alpha V_{ \mathcal{A}_{in}}+\beta V_{ \mathcal{A}_{out}}$\footnote{The collection $\{\mathcal{A}_3, \alpha, \beta\}$ is used as the subscript for $W$ since the random variable is defined using $\alpha$ and $\beta$.}.
The mutual information terms are evaluated with the distribution $P$. 
 Define  information rate-distortion region as 
\begin{align*}
 \mathsf{R}_F(D_\mcn:\mcn \in \mathcal{L})\triangleq cl \left( \underset{{(P,g_\mathcal{L}) \in \mathcal{P}(D_\mcn:\mcn \in \mathcal{L})}}\bigcup \ \  \alpha_F(P, g_{\mathcal{L}}) \right).
 \end{align*}
\end{definition}

\begin{theorem}
\label{thm:linoneadd}
For a given source $(P_X,d_\mcn:\mcn \in \mathcal{L})$,
the operational rate-distortion region contains the information rate-distortion region, i.e., 
$\mathsf{R}_{F}(D_\mcn:\mcn \in \mathcal{L}) \subseteq \mathsf{R}(D_\mcn:\mcn \in \mathcal{L})$.
\end{theorem}

\textit{Proof Outline.} Given the random variables $X$ and $U_{\mathcal{M}},\mathcal{M}\in \mathbf{S}_{\mathsf{L}}, V_{\mathcal{A}_{in}}, V_{\mathcal{A}_{out}}, W_{\mathcal{A}_{sum}}$ described in the theorem statement,   the transmission system first discretizes the source using techniques developed in the prior sections and then uses the  
discrete SSC strategy introduced in  \cite[Theorem 5]{shirani2018achievable}. Next, we rewrite the entropy terms in   \cite[Theorem 5]{shirani2018achievable}  in terms of mutual information terms using the fact that for any triple $U,V,X$ the following holds
\begin{align*} 
&H(\alpha U+\beta V|X)-H(U|X)
=H(\alpha U+\beta V|X)-H(U,V|X)+H(V|X,U)
\\
&=H(\alpha U+\beta V|X)-H(\alpha U+\beta V,V|X)+H(V|X,U)
=-H(V|X,\alpha U+\beta V)+H(V|X,U)\\
&=I(\alpha U+\beta V;V|X)-I(U;V|X).
\end{align*} 
The resulting mutual-information terms in the covering bounds given in \eqref{sec1cov1},\eqref{sec1cov2} and packing bounds given in \eqref{sec1pack1},\eqref{sec1pack2} for the discretized variables converge to that of the continuous variables as the clipping limits are increased asymptotically and the quantization step approaches zero by similar arguments as in the prior sections.

\subsection{Improvements for Using Nested Lattice Quantizers}
We proceed to show through an example that using nested lattice codes gives gains in terms of achievable rate-distortion. 
\\\textit{Example 1:}
 The set-up is shown in Figure \ref{fig_Gau}. Here $X$ and $Z$ are independent zero-mean, unit-variance, Gaussian sources. The distortion function for the individual decoders is mean squared error. Decoder 1 and 2 want to reconstruct $X$ and $Z$, respectively, with mean squared error less than or equal to $P$, and Decoder 3 wants to reconstruct $Y=X+Z$ with distortion $2P$. Each of the joint decoders wish to reconstruct $X$ and $Z$ with distortion mean square error less than or equal to $P$ 


\begin{proposition}
\label{th:MD:Ex1}
The rate triple $(R_1,R_2,R_3)=(\frac{1}{2}\log(\frac{1}{P}),\frac{1}{2}\log(\frac{1}{P}),\frac{1}{2}\log(\frac{2}{P}))$ is achievable using the SSC strategy with structured codes, i.e. $(R_1,R_2,R_3)\in \mathsf{R}_F$. 
\end{proposition}
\begin{figure}[!t]
\centering
\includegraphics[width=3.1in,height=3.1in]{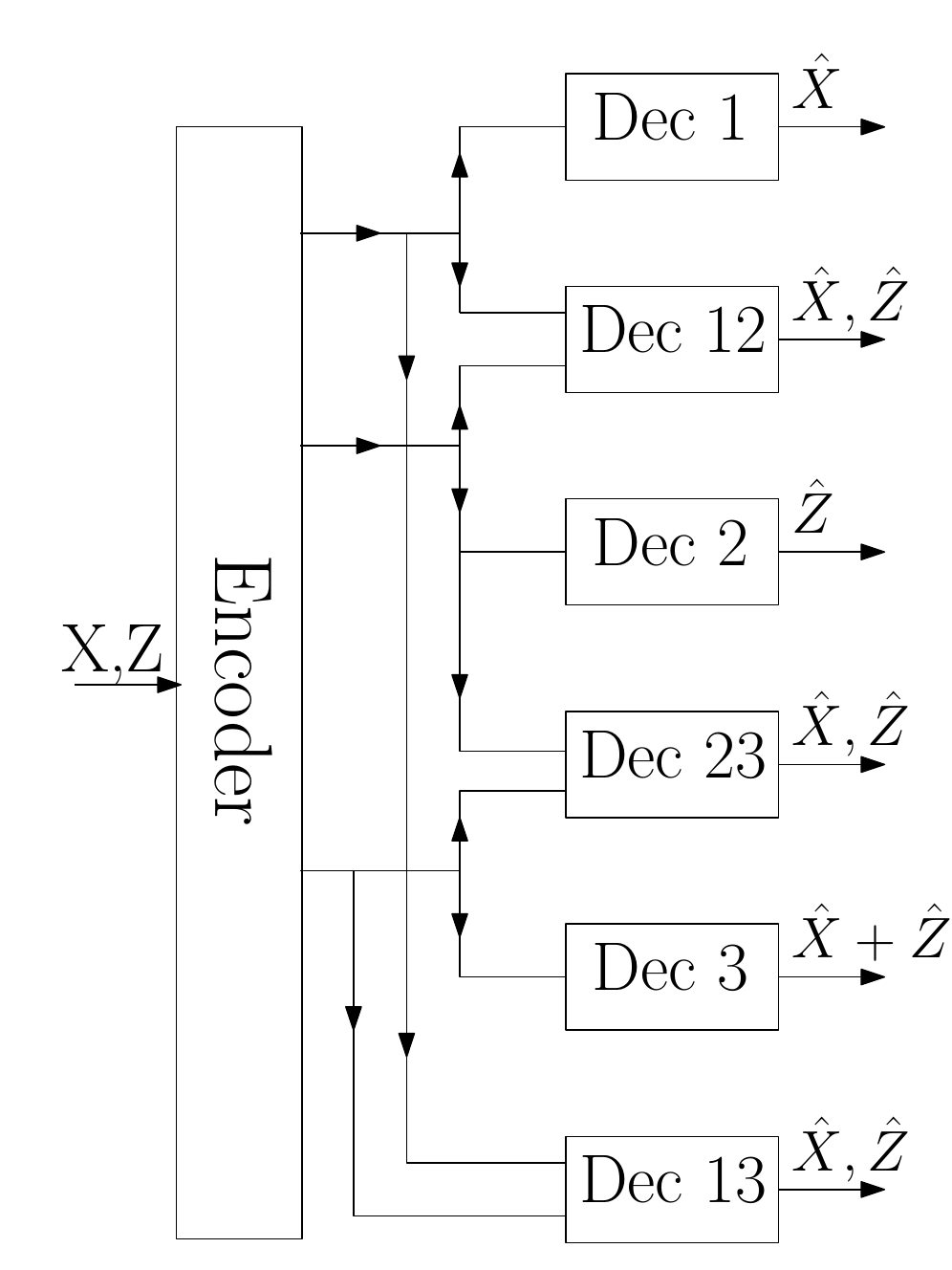}
\caption{Example for Vector Gaussian Source}
\label{fig_Gau}
\end{figure}
\begin{proof}
The proof follows by taking $\mathcal{A}_{in}=\{1\}$, $\mathcal{A}_{out}=\{2\}$, and $\mathcal{A}_{sum}=\{3\}$ with
$V_{\{1\}}=X+N_P$, $V_{\{2\}}=Z+N'_P$ and $W_{\{3,1,1\}}=V_{\{1\}}+V_{\{2\}}$ in Theorem \ref{thm:linoneadd}, where $N_P,N'_{P}$ are independent Gaussian variables with zero mean and variance $P$, and all other variables are taken to be trivial. 
\end{proof}

\begin{proposition} Let $P<1/2$, then the rater triple $(R_1,R_2,R_3)=(\frac{1}{2}\log(\frac{1}{P}),\frac{1}{2}\log(\frac{1}{P}),\frac{1}{2}\log(\frac{2}{P}))$ is not achievable with SSC strategy using unstructured codes, i.e. $(R_1,R_2,R_3)\notin \mathsf{R}_{SS}$.
\end{proposition}
\begin{proof}
The achievable rates are shown in Figure \ref{Fig:MD:Ex1}. The details of the proof are provided in Appendix \ref{App:th:MD:Ex1}.
\end{proof}
 \begin{figure}[t]
\centering
\includegraphics[width =0.9 \textwidth]{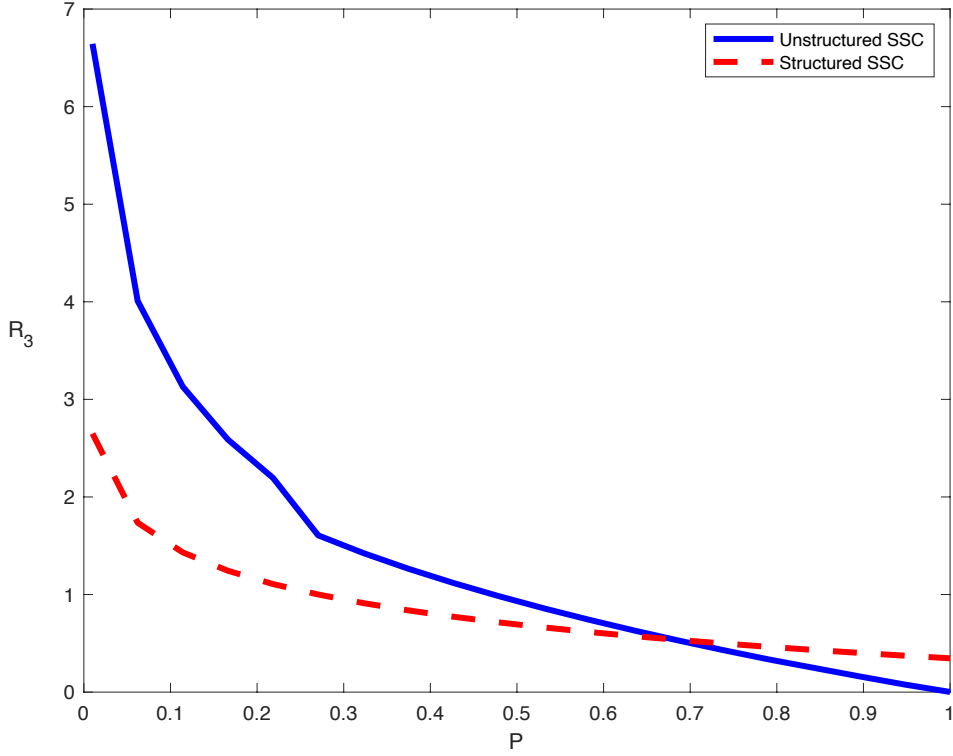}
\caption{Achievable rates $R_3$ in Example 1 using the SSC strategy with unstructured codes (blue curve) and structured codes (red curve). } 
\label {Fig:MD:Ex1}
\end{figure}
 Note that in the above example, as a result of the independence between the two source components $X$ and $Z$, the extra covering bound in \eqref{sec1cov2} is redundant. However, this is not always the case, to illustrate this point we investigate the following example.
\\\textit{Example 2:}
   Assume the source $X$ is a scalar zero-mean unit-variance Gaussian. Consider the random variables $U$ and $V$ which are jointly Gaussian with $X$ and have the following covariance matrix $(\frac{1}{2}<P<\frac{2}{3})$:
\[ \mbox{Cov}([X, U, V])= \left[ \begin{array}{ccc}
1 & 1-P & 1-P \\
1-P & 1-P & 0\\
1-P & 0 & 1-P \end{array} \right].\]

 We intend to transmit $U$ on the first description (i.e. $V_{\{1\}}=U$), $V$ on the second description (i.e. $V_{\{2\}}=V$) and $U+V$ on the third description (i.e. $W_{\{3,1,1\}}=U+V$). In this case the covering bound \eqref{sec1cov2} is not redundant. To see this note that the covering bound is non-redundant if $I(U+V;V|X)-I(U;V|X)<0$. Simplifying the inequality, the bound is non-redundant if $Var(V|X,U)<Var(V|X,U+V)$. Also,
 \begin{align*}
& Var(V|X,U)=\frac{P(1-P)}{2},Var(V|X,U+V)=\frac{1-P}{2}\\
 &\Rightarrow Var(V|X,
 U)<Var(V|X,U+V)\iff P<\frac{2}{3}
 \end{align*}
 
 Which shows the bound is non-redundant in this setting. We calculate the achievable rates using Fourier–Motzkin elimination gives:
 \vspace{-.03in}
 \begin{align*}
 R_1=&R_2=\max\big{\{}\frac{I(UV;X)}{2},I(U;X)-I(\alpha U+\beta V;V|X)+I(U;V|X)\big{\}}
 \\&\hspace{-0.3in}R_3=R_1-H(U)+H(U+V)
 \end{align*}
 We have:
 \begin{align*}
 &I(UV;X)=\frac{1}{2}\log(\frac{1}{2P-1}), \hspace{0.1in} I(U;X)=\frac{1}{2}\log\frac{1}{P},
 \\& I(U;V|X)-I(\alpha U+\beta V;V|X)=\frac{1}{2}\log (\frac{\alpha^2P}{\alpha^2+\beta^2-(\alpha+\beta)^2(1-P)})
 \end{align*}
 Hence $R_1=R_2=\max(\frac{1}{2}\log(\frac{P^2}{P+(1-P)^2}), \frac{1}{4}\log(\frac{1}{2P-1})), R_3=R_1+\frac{1}{2}$. And the distortions are $D_1=D_2=P, D_3=2P, D_{12}=D_{13}=D_{23}=2P-1$.
 
\section{Conclusion}
\label{sec:conc} 
A new framework for handling continuous source and channel networks was introduced. The framework
involves fine discretization of the source and channel variables followed by communication over the resulting discretized network. Convergence results for  information measures under the proposed discretization process were provided, and these results were used to derive the fundamental limits of communication in point-to-point source coding
and channel coding with side-information, distributed source coding with distortion constraints, the function reconstruction problems (two-help-one), computation over MAC, and the multiple-descriptions source coding problem.

\begin{appendices}

\section{Proof of Lemmas in Section \ref{sec:3}}
\subsection{Proof of Lemma \ref{lem:5}}
\label{App:lem:5}
Let $U_{\epsilon}\triangleq U+N_{\epsilon}$. 
\begin{align*}
  &I({N}_{\epsilon};{U}_{\epsilon})=
  h({U}_{\epsilon}) -h(U)=\int f_{U}(u)\log{f_{U}(u)}du- \int f_{{U}_{\epsilon}}(u)\log{f_{{U}_{\epsilon}}(u)}du.
\end{align*}

The integral $\int f_{{U}_{\epsilon}}(u)\log{f_{U_{\epsilon}}(u)}du$ converges to $\int f_{U}(u))\log{f_{U}(u)}du$ as $\epsilon \to 0$ by 
Fatou's Lemma \cite{royden1988real}. To see this, take $g_\epsilon(u)\triangleq  f_{U_\epsilon}(u)\log{f_{U_\epsilon}(u)}+1, u\in  [-M-\epsilon,M+\epsilon]$ and $g(x)\triangleq f_{U}(u)\log{f_{U}(u)}+1, u\in  [-M,M]$,
and note that $g_{\epsilon}(u)$ is  non-negative for $u \in [-M-\epsilon,M+\epsilon]$  since $f_{U_\epsilon}(u)\log{f_{U_\epsilon}(u)}\geq -\frac{1}{e}\log{e}> -1, u\in \mathbb{R}$. So, by Fatou's lemma $\int_{u\in\mathbb{R}} g(u)du\leq \liminf_{\epsilon\to 0} \int_{u\in \mathbb{R}} g_{\epsilon}(u)du$ since by construction $g_{\epsilon}(u)$ converges to $g(u)$ in a pointwise manner almost everywhere as $\epsilon\to 0$. The last statement follows by noting that $f_{U_\epsilon}(\cdot)$ converges to $f_{U}(\cdot)$ in  a pointwise manner almost everywhere as $\epsilon\to 0$ using the assumptions made in Remark \ref{rem:assumptions}.
This implies that $\int f_{U}(u)\log{f_{U}(u)}du\leq $ $ \liminf_{\epsilon\to 0}\int f_{{U}_{\epsilon}}(u)\log{f_{U_{\epsilon}}(u)}du$.
Also, note that $\int f_U(u)\log{f_{U}(u)}du\geq $ $
\int$ $ f_{U_\epsilon}(u)$ $\log{f_{U_\epsilon}(u)}du$ $, \epsilon>0$ since $I(N_{\epsilon};U_{\epsilon})\geq 0$. So,  $\int f_{U}(u)\log{f_{U}(u)}du\geq $ $ \limsup_{\epsilon\to 0}$ $ \int f_{U_\epsilon}(u)\log{f_{U_{\epsilon}}(u)}du$. 
Consequently, $I(N_{\epsilon};U_{\epsilon})\to 0$ as $\epsilon\to 0$.

\subsection{Proof of Lemma \ref{thm:quantize_dist_cost}}
\label{App:lem:6}
Let us fix $\e>0$ and let 
$\alpha_1=\liminf_{s \rightarrow \infty} \kappa(s)$ and 
$\alpha_2=\liminf_{s \rightarrow \infty} \kappa(-s)$.
Since $\mathbb{E}(\kappa(S))< \infty$,  we have  
\begin{equation}
\int_{[-\ell,u]^c} \kappa(s) dP_S(s) \leq \e, \ \  \mbox{ and } \ \ 
\int_{[-\ell,u]^c}  dP_S(s) \leq \e 
\label{eq:dist_converge1}
\end{equation}
for all $u > U(\e)$ and $\ell > L(\e)$ for some $U(\e)$ and $L(\e)$.

First consider the case when $\alpha_1<\infty$ and $\alpha_2<\infty$. In this case,
there exists $U_{1}(\e)$ such that for all $u > U_{1}(\e)$, we have 
\[
\alpha_1 + 2 \e > \alpha_1 + \e  \geq \left( \inf_{\tilde{u}>u} \kappa(\tilde{u})
\right)  \geq \alpha_1 - \e.
\]
This implies that for all $u > U_{1}(\e)$ there exists $u^*>u$ such
that $\alpha_1 + 2\e > \kappa(u^*)$. Hence, by choosing $u >
\max\{U_1(\e),U(\e)\}$, we obtain a $u^*>u$ such that
\begin{align}
\kappa(u^*) P(S \geq u^*)  &\leq  (\alpha_1+2 \e)P(S \geq u^*) \nn 
=(\alpha_1-\e+3\e) \int_{u^*}^{\infty} dP_S(s) \nn 
\leq \int_{u^*}^{\infty} \kappa(s) dP_S(s) + 3\e^2. \nn 
\end{align}
Similarly, there exists $\ell^*$ such that 
\[
\kappa(-\ell^*) P(S \leq -\ell^*) \leq  \int_{-\infty}^{-\ell^*} \kappa(s) dP_S(s) + 3\e^2.
\]
Now using the above results, consider 
\begin{align}
\mathbb{E}\kappa(\widetilde{S}_{\ell^*,u^*}) &= \int_{-\ell^*+}^{u^*-} \kappa(s) dP_S(s)
+ \kappa(u^*)P(S \geq u^*)+ 
\kappa(-\ell^*) P(S \leq \ell^*) \nn
\leq  6 \e^2 + \mathbb{E} \kappa(S). \nn
\end{align}
Moreover,
\begin{align}
\mathbb{E}\kappa(\widetilde{S}_{\ell^*,u^*}) &\geq \int_{-\ell^*+}^{u^*-} \kappa(s) dP_S(s) \nn 
\geq \int \kappa(s) dP_S(s) -2 \e = -2\e+ \mathbb{E} \kappa(S). \nn
\end{align}
Next consider the case when $\alpha_1=\infty$. 
Define 
\[
u^*\triangleq \arg \min_{\{u \geq U(\e)\}} \kappa(u).
\]
We have 
\[
\kappa(u^*) P(S\geq u^*) \leq \int_{u^*}^{\infty} \kappa(s) dP_S(s).
\]
Using a similar argument for $\alpha_2$ we get the desired result
(second statement) for the cost function. 

The third statement follows by similar arguments. Note that there
exists a $T$-measurable function $d_1$ such that 
$
\mathbb{E}d(S,T)= \mathbb{E}d_1(T)
$.
This follows because $\int \int_{\mathsf{A}} d(s,t) d P_{ST}(s,t)$ is a finite
measure on $\mathsf{A}$ for every $\mathsf{A} \in \mathcal{B}$, and is absolutely
continuous with respect to $P_T$, and $d_1$ is the corresponding
Radon-Nikodym derivative.  It follows by similar arguments as in the 
proof of the second statement for the cost function that
there exists a sequence of lengths $\ell_{2m},u_{2m}$ such that 
\[
\lim_{m \rightarrow \infty}
\mathbb{E}d(S,\widetilde{T}_{\ell_{2m},u_{2m}})=\mathbb{E}d(S,T).
\]
Similarly, it follows that 
there exists a sequence of lengths $\ell_{1n},u_{1n}$ such that 
\[
\lim_{n \rightarrow \infty}
\mathbb{E}d(\widetilde{S}_{\ell_{1n},u_{1n}},\widetilde{T}_{\ell_{2},u_{2}})=\mathbb{E}d(S,\widetilde{T}_{\ell_{2},u_{2}}),
\]
for every $\ell_{2},u_2$. This completes the proof.
\qed
\subsection{Proof of Lemma \ref{lem:mc_forced1}}
\label{App:mc_forced1}
\begin{align*}
    &I(A;C|B)+I(E;C|D) \\&=\sum_{a,b,c,d,e} P_C(c)P_{B|C}(b|c) P_{D|C}(d|c) P_{A|BC}(a|b,c) P_{E|DC}(e|d,c) \log{\frac{P_{A|BC}(a|b,c)P_{E|DC}(e|d,c)}{P_{A|B}(a|b)P_{E|D}(e|d)}}
    \\&= \sum_{b,c,d} P_{C}(c)P_{B|C}(b|c) P_{D|C}(d|c) D(P_{A|BC}(\cdot|b,c)P_{E|DC}(\cdot|d,c)|| {P_{A|B}(\cdot|b)P_{E|D}(\cdot|d)})
    \\&\stackrel{(a)}\geq \sum_{c} P_C(c) D(P_{AE|C}(\cdot|c)||P_{\widehat{A},\widehat{E}|C}(\cdot|c))
    \\& \stackrel{(b)}{\geq} \frac{1}{\ln{2}}\sum_{c} P_C(c) V^2(P_{A,E|C}(\cdot|c), P_{\widehat{A},\widehat{E}|C}(\cdot|c))
    \\& \stackrel{(c)}{\geq} \frac{1}{2\ln{2}}(\sum_{c} P_C(c) V(P_{A,E|C}(\cdot|c), P_{\widehat{A},\widehat{E}|C}(\cdot|c)))^2
    = \frac{1}{2\ln{2}} V^2(P_{A,E,C}, P_{\widehat{A},\widehat{E},C}),
\end{align*}
where (a) follows from the convexity of relative entropy, (b) follows from 
Pinsker's inequality and (c) follows from Jensen's inequality. 

\section{Proof of Theorem \ref{thm:wynerziv_c}}
\label{app:wynerziv_c}
Let us consider a joint probability measure
$P_{XYU}$ and $g(\cdot,\cdot)$ that satisfies the conditions given in the theorem. 
Define the induced distortion function $d':\mathbb{R}^3
\rightarrow \mathbb{R}^+$ as 
\[
d'(x,y,u)\triangleq d(x,y,g(y,u)).
\]
{Note that $d'(\cdot,\cdot,\cdot)$ is continuous everywhere.}


\noindent \textbf{Step 1 (Clipping $X$):}
First, we wish to clip $X$ to produce $\widetilde{X}$, so that the support of $\widetilde{X}$ is restricted to $[-\ell_X,u_X]$. To avoid producing discrete components in $\widetilde{X}$, and preserve the Markov chain $\widetilde{Y}\leftrightarrow \widetilde{X} \leftrightarrow \widetilde{U}$, we follow the clipping procedure below.  

Let $Z$ and $W$ be two random variables
that are independent of the source $(X,Y)$ such that  $Z$ has support 
$[-\ell_X,u_X]$, and the distribution $P_{ZW}$ is  
given by 
\[
P_{ZW}(\mathsf{A} \times \mathsf{B})\triangleq \frac{P_{XY}( \mathsf{A} \cap [-\ell_X,u_X]  \times \mathsf{B})}{P(X \in [-\ell_X,u_X])}
\]
for all events $\mathsf{A}$ and $\mathsf{B}$ in the Borel sigma-field. 
We define 
\[
(\widetilde{X},\overline{Y})\triangleq\Bigg\{
\begin{array}{cc}
(X,Y) & \mbox{ if } X\in [-\ell_X,u_X] \\
(Z,W) & \mbox{ otherwise}
\end{array},
\]
where $\overline{Y}$ is an auxiliary variable which is defined so that  $P_{\widetilde{X},\overline{Y}}(\mathsf{A}\times \mathsf{B})=P_{X,Y}(\mathsf{A}\times \mathsf{B}| X\in [-\ell_X,u_X])$.
The information that whether 
$X \in [-\ell_X,u_X]$ or not will  be communicated to the decoder from the encoder. 
Similarly, let $\overline{U}$ be a random variable that is correlated with
$(X,Y,Z,W)$ such that the distribution of $(\widetilde{X},\overline{U})$ 
is given by 
\[
P_{\widetilde{X},\overline{U}}(\mathsf{A} \times B)=\frac{P_{X,U}(\mathsf{A} \cap [-\ell_X,u_X] \times \mathsf{B})}{P(X \in
  [-\ell_X,u_X])}
\]
for all event $\mathsf{A}$ and $\mathsf{B}$, and $(X,Y,\overline{Y}) - \widetilde{X} - \overline{U}$. 
One can now check that the triple $(\widetilde{X},\overline{Y},\overline{U})$
has the distribution 
\[
P_{\widetilde{X},\overline{Y},\overline{U}}(\mathsf{A} \times \mathsf{B} \times \mathsf{C})=
\frac{P_{XYU}(\mathsf{A} \cap [-\ell_X,u_X] \times \mathsf{B} \times \mathsf{C})}{P(X \in
  [-\ell_X,u_X])}.
\]
for all $\mathsf{A},\mathsf{B}$, and $\mathsf{C}$, and $\overline{Y} - \widetilde{X} - \overline{U}$.
It should be noted that
$(\widetilde{X},\overline{Y},\overline{U})$ depends on $\ell_X,u_X$, however, this
is not made explicit to keep the notation simple. 
Next, we show that $I(\widetilde{X};\overline{U}) \approx I(X;U)$, 
$I(\overline{Y};\overline{U}) \approx I(Y;U)$, and
$\mathbb{E}d(\widetilde{X},\overline{Y},\overline{U}) \approx
\mathbb{E}d(X,Y,U)$ for sufficiently large $\ell_X,u_X$. 

Fix an $\e>0$. Consider 
\begin{align}
I(\widetilde{X};\overline{U}) &= \int_{-\ell_X}^{u_X}  \int \frac{d
  P_{X,U}(\tilde{x},\overline{u})}{P_X([-\ell_X,u_X])} \log \frac{ d P_{X,U}}{
  d (P_{\overline{U}} P_X)} (\tilde{x},\overline{u}) \nn \\
&\overset{(a)}{=} \int_{-\ell_X}^{u_X}  \int \frac{d
  P_{X,U}(\tilde{x},\overline{u})}{P_X([-l_X,u_X])} \log \frac{ d P_{X,U}}{
  d (P_{U} P_X)} (\tilde{x},\overline{u}) 
- D(P_{\tilde{U}} \| P_U)  \nn \\ 
&\leq \frac{1}{P_X([-l_X,u_X])} \int_{-\ell_X}^{u_X} \int d
  P_{X,U}(\tilde{x},\overline{u}) 
    \log \frac{ d P_{X,U}}{
  d (P_{U} P_X)} (\tilde{x},\overline{u}) \nn \\ 
&\rightarrow  I(X;U), \nn  
\end{align}
as $\ell_X,u_X \rightarrow \infty$, where in
(a) we note that $P_{\overline{U}} \ll P_{U}$. Also, note that $P_{\widetilde{X},\overline{U}}$ converges strongly to $P_{X,U}$. So,
by the lower semi-continuity of mutual information, it follows that $I(\widetilde{X};\overline{U})=I(X;U)$.
Next, we observe that 
\begin{align}
P_{\overline{Y},\overline{U}}(\mathsf{A} \times \mathsf{B}) &=
P_{\tilde{X}\overline{Y}\overline{U}}([-\ell_X,u_X] \times \mathsf{B} \times \mathsf{C})  \nn = \frac{P_{X,Y,U}([-\ell_X,u_X] \times \mathsf{B} \times \mathsf{C})}{P_X([-\ell_1,u_X])} 
\rightarrow P_{Y,U}(\mathsf{A} \times \mathsf{B}) \nn
\end{align}
as $\ell_X,u_X \rightarrow \infty$, and hence by the lower semi-continuity of mutual information, we have 
\[
\lim_{\ell_X,u_X \rightarrow \infty} I(\overline{Y};\overline{U}) \geq I(Y;U).
\]
Moreover, 
\[
\mathbb{E}d(\widetilde{X},\overline{Y},g(\overline{Y},\overline{U}))=\int_{\tilde{x}=-\ell_X}^{u_X} \int \int
d(\tilde{x},\overline{y},g(\overline{y},\overline{u})) \frac{d
  P_{X,Y,U}}{P_X([-\ell_X,u_X])} (\tilde{x},\overline{y},\overline{u})
\rightarrow \mathbb{E}d(X,Y,g(Y,U))
\]
as $\ell_X,u_X \rightarrow \infty$.

\noindent \textbf{Step 2 (Clipping $\overline{Y}$ and $\overline{U}$):} In this step, we clip $\overline{Y}$ and $\overline{U}$ with parameters $\ell_Y,u_Y$
and $\ell_U,u_U$, respectively, to produce $\widetilde{Y}$ and $\widetilde{U}$. The Markov chain  $\widetilde{U}\leftrightarrow\widetilde{X} \leftrightarrow
\widetilde{Y}$ holds. Using arguments similar to the proof of Lemma
\ref{thm:quantize_dist_cost}, we see that there exists sequence
$\ell_{Y}(n),\ell_{U}(m)$ and $u_{Y}(n),u_{U}(m)$ such that 
\begin{align}
\lim_{n,m \rightarrow \infty}
\mathbb{E}d(\widetilde{X},\widetilde{Y},g(\widetilde{Y},\widetilde{U}))
&= \lim_{n,m \rightarrow \infty}
\mathbb{E}d'(\widetilde{X},\widetilde{Y},\widetilde{U})
\nn \\
&= \mathbb{E}{d'}(\widetilde{X},\overline{Y},\overline{U})=\mathbb{E}d(\widetilde{X},\overline{Y},g(\overline{Y},\overline{U})). \nn
\end{align}
Moreover, 
\[
\lim_{m \rightarrow \infty} I(\widetilde{X};\widetilde{U})=I(\widetilde{X};\overline{U}), \ \ 
\lim_{n,m \rightarrow \infty} I(\widetilde{Y};\widetilde{U})=I(\overline{Y};\overline{U}).
\]

\noindent \textbf{Step 3 (Discretizing  X):} 
Next, we discretize $\widetilde{X}$ into
$\widehat{X}$ with step size $2^{-n_X}$ and enforce the Markov chain. 
Using 
\[
I(\widetilde{X},\widetilde{Y};\widetilde{U}|\widehat{X})=I(\widetilde{X},\widetilde{Y};
\widetilde{U})
-I(\widehat{X};\widetilde{U}), 
\]
and  Theorem \ref{thm:quantize0}, we have
\beq
\label{eq:step3_inter}
   \lim_{n_X \rightarrow \infty}  I(\widetilde{X},\widetilde{Y};\widetilde{U}|\widehat{X})
=I(\widetilde{Y};\widetilde{U}|\widetilde{X})=0.
\eeq

Define $\widetilde{U}'$ as a random variable having the same alphabet as
$\widetilde{U}$, and  that is jointly correlated with
$(\widehat{X},\widetilde{X},\widetilde{Y})$ according to the probability
distribution that satisfies the Markov chain 
$\widetilde{Y} 
\leftrightarrow\widetilde{X} \leftrightarrow \widehat{X} \leftrightarrow \widetilde{U}'$, and 
the pair $(\widehat{X},\widetilde{U}')$ has the same distribution as 
the pair $(\widehat{X},\widetilde{U})$. 
Using (\ref{eq:step3_inter}), and Lemma \ref{lem:mc_forced1} with the triple
$C = (\widetilde{X},\widetilde{Y})$, $B=\widehat{X}$, $A=\widetilde{U}$, and
$\widehat{A}=\widetilde{U}'$, and $(D,E,\widehat{E})$ taken to be trivial, we have\footnote{Note that $\widetilde{U}'$ depends on $n_X$.}
\begin{align}
\label{eq:var:App:A}
   \lim_{n_X \rightarrow \infty}   V \left( P_{\widetilde{X},\widetilde{Y},\widetilde{U}}
  ,  P_{\widetilde{X},\widetilde{Y},\widetilde{U}'} \right)=0,
\end{align}
and hence, we have 
\[
\lim_{n_X  \rightarrow \infty}
I(\widehat{X};\widetilde{U}')
\stackrel{(a)}{=}
\lim_{n_X  \rightarrow \infty}
I(\widehat{X};\widetilde{U})
\stackrel{(b)}{=}
I(\widetilde{X};\widetilde{U}), 
\]
\[ 
\lim_{n_X  \rightarrow \infty}
I(\widetilde{Y};\widetilde{U}')\stackrel{(c)}{\geq}
I(\widetilde{Y};\widetilde{U}),
\]
where in (a) we have used the fact that by construction $P_{\widehat{X},\widetilde{U}'}=P_{\widehat{X},\widetilde{U}}$, in (b) we have used Theorem \ref{thm:quantize0}, and (c) follows from lower semi-continuity of mutual information and \eqref{eq:var:App:A}.

Since convergence in variational distance implies convergence in
distribution, using the Portmanteau theorem \cite{achim2014probability,billingsley2013convergence}, we have 
\begin{align}
\lim_{n_X  \rightarrow \infty}
\mathbb{E}d(\widehat{X},\widetilde{Y},g(\widetilde{Y},\widetilde{U}'))
&= \lim_{n_X  \rightarrow \infty} \mathbb{E}
{d}'(\widehat{X},\widetilde{Y},\widetilde{U}') 
=\mathbb{E}\tilde{d}(\widetilde{X},\widetilde{Y},\widetilde{U})
=\mathbb{E}d(\widetilde{X},\widetilde{Y},g(\widetilde{Y},\widetilde{U})),
\end{align}
where we have used 

\begin{align*}
      \lim_{n_X \rightarrow \infty}   V \left( P_{\widehat{X},\widetilde{Y},\widetilde{U}'}
  ,  P_{\widetilde{X},\widetilde{Y},\widetilde{U}} \right)\leq 
    \lim_{n_X \rightarrow \infty}   V \left( P_{\widehat{X},\widetilde{Y},\widetilde{U}'}
  ,  P_{\widetilde{X},\widetilde{Y},\widetilde{U}'} \right)
  + 
  V \left( P_{\widetilde{X},\widetilde{Y},\widetilde{U}'}
  ,  P_{\widetilde{X},\widetilde{Y},\widetilde{U}} \right)
  =0,
\end{align*}
where the first term, $V \left( P_{\widehat{X},\widetilde{Y},\widetilde{U}'}
  ,  P_{\widetilde{X},\widetilde{Y},\widetilde{U}'} \right)$, goes to $0$ as $n_X\to \infty$ since  convergence in total variation is implied by weak convergence (Portmonteau theorem \cite{billingsley2013convergence}), and the latter is shown below:
\begin{align*}
  \big|F_{\widehat{X},\widetilde{Y},\widetilde{U}'}(a,b,c)- & F_{\widetilde{X},\widetilde{Y},\widetilde{U}'}(a,b,c)\big|
    = 
    \big|P( \widehat{X}\leq a, \widetilde{Y}\leq b,\widetilde{U}'\leq c)
    -     P( \widetilde{X}\leq a, \widetilde{Y}\leq b,\widetilde{U}'\leq c)\big|
    \\& =  \big|   P( \widehat{X}\leq a-\delta, \widetilde{Y}\leq b,\widetilde{U}'\leq c)+  P( a-\delta < \widehat{X}\leq a, \widetilde{Y}\leq b,\widetilde{U}'\leq c)
 \\&  \hspace{0.1in}  - P( \widetilde{X}\leq a-\delta, \widetilde{Y}\leq b,\widetilde{U}'\leq c)
 -  P( a-r< \widetilde{X}\leq a, \widetilde{Y}\leq b,\widetilde{U}' \leq c)\big|
 \\& \leq
 P( a-\delta \leq \widetilde{X}\leq a-\delta+2^{-n_X},\widetilde{Y}\leq b,\widetilde{U}'\leq c)
\\&  \leq P( \widetilde{X}\in [a-\delta,a-\delta+2^{-n_X}])\to 0, \text{ as } n_X\to \infty, 
\end{align*}
where $\delta\triangleq a-Q_{n_{X}}(a)-\frac{1}{2^{n^X}}$, where $Q_{n_X}(\cdot)$ is defined in Definition \ref{Def:Desc}.


\noindent \textbf{Step 4 (Discritizing $\widetilde{Y}$ and $\widetilde{U}'$) :} Next, we discretize $\widetilde{Y}$ and
$\widetilde{U}'$ into $\widehat{Y}$ and
$\widehat{U}$ with step sizes $2^{-n_Y}$ and $2^{-n_U}$, respectively. We can use the  above arguments
 to show
convergence of mutual information and expected distortions as $n_Y$ and $n_U$ become large. 
It is worth noting here that
$g(\widehat{Y},\widehat{U})$ is
bounded and finite-valued, and so no need to discretize it. 


In other words,  for every $\e>0$, there exist infinitely many
$\ell_X,u_X,\ell_Y,u_Y,\ell_U,u_U$ such that for all sufficiently large
$n_X,n_Y,n_U$, we have 
\begin{equation}
I(\widehat{X};\widehat{U})-I(\widehat{Y};\widehat{U})
\leq I(X;U)-I(Y;U)+2\e,
\label{eq:side_info_1}
\end{equation}
and
\begin{equation}
\mathbb{E}d(\widehat{X},\widehat{Y},g(\widehat{Y},\widehat{U}))
\leq \mathbb{E}d(X,Y,g(Y,U))+\e.
\label{eq:side_info_2}
\end{equation}

\noindent \textbf{Step 5 (Wyner-Ziv Source Coding Scheme):} 
Now we can use  Wyner-Ziv theorem  \cite{Ziv} to show the existence of a transmission system.  
In particular, we can show that the rate-distortion pair
given  by 
\[
(I(\widehat{X};\widehat{U})-I(\widehat{Y};\widehat{U}),
\mathbb{E}d(\widehat{X},\widehat{Y},g(\widehat{Y},\widehat{U}))
\]
 is achievable for the finite-alphabet source $(\widehat{X},\widehat{Y},d')$. 
In other words, for all $\e>0$, and for all sufficiently large blocklength $m$, there exists a transmission system
$\mbox{TS}_d$ with parameter $(m,\Theta)$ for compressing the finite-alphabet source
such that 
\begin{equation}
\frac{1}{m} \log \Theta \leq
I(\widehat{X};\widehat{U})-I(\widehat{Y};\widehat{U})+\e,
\label{eq:side_info_3}
\end{equation}
and
\begin{equation}
\frac{1}{m}\sum_{i=1}^m \mathbb{E}d'(\widehat{X}_{i},\widehat{Y}_{i},\widehat{U}'_i)
\leq
\mathbb{E}d'(\widehat{X},\widehat{Y},\widehat{U})
+\e.
\label{eq:side_info_4}
\end{equation}
 Let 
$e(\cdot)$ be the encoder of the Wyner-Ziv coding scheme,  and $f(\cdot)$ be the decoder.

For the source $(X,Y,d)$ we obtain an $(m,\Theta')$  transmission system 
$\mbox{TS}_c$ as follows. We assume that the encoder and decoder share
common randomness.  From $(X,Y)$ we create
$(\widehat{X},\widehat{Y})$ and use $\mbox{TS}_d$.
Let $V=\mathbbm{1}_{\{X \in [-\ell_X,u_X]\}}$. The encoder of
$\mbox{TS}_c$ sends information to the decoder in two parts.
The first part is $e(\widehat{X}^m)$ and the second part is a  compressed (almost
lossless) version of $V^m$.  If the sequence $V^m$ is typical, the
encoder sends the index of the sequence in the typical set, otherwise,
it sends the index $0$. This extra piece of information requires
$h_b(P_X([-\ell_X,u_X])+\e$ bits per sample. The decoder is constructed
as follows. Let $\check{U}^m$ denote the reconstruction vector.
If the decoder receives index $0$ in the second part, then
it uses a constant $c$ as a reconstruction, i.e.,
$\check{U}_i=c$ for all $i$. Otherwise, it can
reconstruct $V^m$ reliably. If $V_i=1$, then the reconstruction is
$\check{U}_i=\widehat{U}'_i$, otherwise it is a  constant $c$, such that $\mathbb{E}d'(X,Y,c)<\infty$,
so that $\check{U}_i=c$. 

Fix $\e>0$. Assume that the ten parameters
$\ell_X,\ell_Y,\ell_U,u_X,u_Y,u_U,n_X,n_Y,n_U$ and $m$  are such that
\begin{itemize}
\item (\ref{eq:side_info_1},\ref{eq:side_info_2},\ref{eq:side_info_3},\ref{eq:side_info_4}) are
satisfied.  
\item $h_b(P_X([-\ell_X,u_X])) \leq \e$.
\item $P(\mathsf{A}^c) \mathbb{E}d'(X_i,Y_i,c|\mathsf{A}^c) \leq \e$, for all $1 \leq i
  \leq m$, where $\mathsf{A}$ denote the event
  that $V^m$ is typical. 
\item $P_X([-\ell_X,u_X]^c) \mathbb{E}d'(X,Y,c|X \not \in [-\ell_X,u_X]) \leq \e$
\item 
$|d'(x_1,y_1,b)-d'(x_2,y_2,b)| \leq \e$, 
for all (a)  $x_1,x_2 \in [-\ell_X,u_X]$, (b) $y_1,y_2 \in [-\ell_Y,u_Y]$,
(c)  $b \in [-\ell_U,u_U]$, (d) $|x_1-x_2| \leq \frac{1}{2^{n_X}}$, and 
(e) $|y_1-y_2| \leq \frac{1}{2^{n_Y}}$.
\end{itemize}
Note that the third and the fourth bullet follows because of the
assumption made on the distortion function that 
$\mathbb{E}d'(X,Y,c)<\infty$. In particular for the third bullet, note
that 
\[
P(A^c) \mathbb{E}d'(X_i,Y_i,c|\mathsf{A}^c) =\int_{\mathsf{A}^c} dP_{X^nY_i}(x^n,y_i)
d'(x_i,y_i,c) \leq \e
\]
as $P(\mathsf{A}^c)$ can be made arbitrarily small.
Note that the last bullet follows from the 
uniform continuity of the distortion function 
over $[-\ell_X,u_X] \times [-\ell_Y,u_Y] \times
[-\ell_U,u_U]$ for a fixed $\ell_X,\ell_Y,\ell_U$ and $u_X,u_Y,u_U$.

We see that 
\[
\frac{1}{m} \log \Theta' \leq I(X;U)-I(Y;U)+3 \e.
\]
Next, we find the expected distortion. Let $\mathsf{B}_i$ denote the event $X_i
\in [-\ell_X,u_X]$. Let $P_{X_i,Y_i|\mathsf{A},\mathsf{B}_i}$ denote the
probability distribution of $(X_i,Y_i)$ given the event $\mathsf{A}$ and $\mathsf{B}_i$. Consider for any $i \in
\{1,2,\ldots,m\}$, 
\begin{align}
\mathbb{E}d'(X_i,Y_i,\check{U}_i) &\leq  P(\mathsf{A}^c) \mathbb{E}d'(X_i,Y_i,c|\mathsf{A}^c)+
P(\mathsf{B}_i^c) \mathbb{E}d(X_i,Y_i,c|\mathsf{B}_i^c)  
+ P(\mathsf{A} \cap \mathsf{B}_i)
\mathbb{E}d(X_i,Y_i,\check{U}_i|\mathsf{A},\mathsf{B}_i) \nn \\
&\overset{(a)}{\leq} 2\e+
P(\mathsf{A} \cap \mathsf{B}_i) \mathbb{E}d'(X_i,Y_i,\widehat{U}'_i|\mathsf{A} \cap \mathsf{B}_i) \nn \\
&\leq  2\e+ P(\mathsf{B}_i) \mathbb{E}d'(X_i,Y_i,\widehat{U}'_i|\mathsf{B}_i) \nn \\ 
&\overset{(b)}{=} 2\e+ P(\mathsf{B}_i) \sum_{i,j,b} P_i(b|\zeta(i),\zeta(j)) \int_{\mathsf{A}(i) \times
  \mathsf{A}(j) } d'(x,y,b) \frac{dP_{X_iY_i}(x,y)}{P(\mathsf{B}_i)}  \nn \\
&\overset{(c)}{\leq} 2\e+P(\mathsf{B}_i) \left[
  \mathbb{E}d'(\widehat{X}_{i},\widehat{Y}_{i},\widehat{U}'_i) +
  \e \right] \nn 
\end{align}
where we have following arguments:
(a) follows from third and fourth bullets from the previous
page. In  (b) we have denoted the conditional probability of
$\widehat{U}'_i$ given $\widehat{X}_{i},\widehat{Y}_{i}$ as
$P_i$. 
(c) follows from the fifth bullet from the previous page.
Finally, we have 
\begin{align}
\frac{1}{m} \sum_{i=1}^m \mathbb{E}d'(X_i,Y_i,\widehat{U}'_i)
&\leq 2\e+\left[
  \mathbb{E}d'(\widehat{X},\widehat{Y},\widehat{U}) +
  2\e \right]  
  \leq 5\e+ \mathbb{E}d'(X,Y,U). \nn 
\end{align}
This completes the proof.\qed
\section{Proof of Theorem \ref{thm:Gelfand}}
\label{App:Gelfand}
 Consider a joint probability measure
$P_{XSU}$ and $g(\cdot,\cdot)$ that satisfy the conditions given in
the theorem statement. 

\noindent \textbf{Step 1 (Clipping U):} Define the induced cost function $\tilde{\kappa}:\mathbb{R}^2
\rightarrow \mathbb{R}^+$ as 
\[
\tilde{\kappa}(u,s)\triangleq\kappa(g(s,u),s).
\]
Note that $\tilde{\kappa}$ is continuous everywhere.
Since $X$ is a function of $(U,S)$, let us average it out and look
at just the triple $(S,U,Y)$ with the induced joint measure
$P_SP_{U|S}P_{Y|US}$ given by 
\[
P_{SUY}(\mathsf{A} \times \mathsf{B} \times \mathsf{C}) = \int_\mathsf{A} P_S(ds) \int_\mathsf{B} P_{U|S}(du|s)
\int_\mathsf{C} P_{Y|US}(dy|u,s)
\]
where $P_{Y|US}:  \mathcal{B}\times \mathbb{R}^2  \rightarrow
\mathbb{R}^+$ is a transition probability and is given by 
$P_{Y|US}(\mathsf{A}|u,s)=P_{Y|XS}(\mathsf{A}|g(u,s),s)$ for all $\mathsf{A} \in \mathcal{B}$. 

To address the difficulties that arise due
to tail probabilities, we first clip the input of the channel. The original distribution is $P_{SUY}=P_SP_{U|S}P_{Y|US}$.  
Let $\mathsf{E}_S=(-\infty,-\ell_S] \cup
[u_S,\infty)$ and $\mathsf{E}_U=(-\infty,-\ell_U] \cup
[u_U,\infty)$.
Let $\widetilde{U}=C_{\ell_U,u_U}(U)$ if $S \in
[-\ell_S,u_S]$, and $\widetilde{U}=c$ otherwise, where $c$ is a constant such that
$\mathbb{E}\tilde{\kappa}(S,c) < \infty$. 
 In other words,
$S$ and $\widetilde{U}$ have the  following joint
probability distribution:
\begin{align}
& P(S \in \mathsf{A}, \widetilde{U} \in \mathsf{B}) = \int_\mathsf{A} P_S(ds) \int_\mathsf{B}
P_{\widetilde{U}|S}(du|s) \nn \\
&= \int_{\mathsf{A} \cap \mathsf{E}_S^c} P_S(ds) \left[ \int_{\mathsf{B} \cap  \mathsf{E}_U^c}
  P_{U|S}(du|s)+ \mathbbm{1}_{\{-\ell_U \in \mathsf{B}\}} P_{U|S}((-\infty,-\ell_U]|s) +
\mathbbm{1}_{\{u_U \in \mathsf{B}\}} P_{U|S}([u_U,\infty)|s) \right]  \nn \\
& + \int_{\mathsf{A} \cap \mathsf{E}_S} P_S(ds) \mathbbm{1}_{\{c \in \mathsf{B}\}}, \nn
\end{align}
 for all borel sets $\mathsf{A}$ and $\mathsf{B}$.
 
Note that $\widetilde{U}$ depends on $\ell_S,u_S,\ell_U,u_U$, but this is not
made explicit for ease of notation.
We use
$\widetilde{U}$ as the input to the channel $P_{Y|SU}$. Let ${Y}'$
denote the corresponding channel output. In other words,
for any triple of borel sets $\mathsf{A}$, $\mathsf{B}$ and $\mathsf{C}$, we have 
\[
P(S \in \mathsf{A}, \widetilde{U} \in \mathsf{B}, {Y}' \in \mathsf{C})=\int_\mathsf{A} P_S(ds) \int_\mathsf{B} P_{\widetilde{U}|S}(du|s)
\int_\mathsf{C} P_{Y|US}(dy|u,s).
\]
Let $P_{S \widetilde{U} {Y}'}$ denote the joint probability distribution
  of $S,\tilde{U}$ and ${Y}'$. 
Next we show convergence of mutual information terms and the cost term
as $\ell_S,\ell_U$ and $u_S,u_U$ become large.   For any triple of borel sets $\mathsf{A}$, $\mathsf{B}$
and $\mathsf{C}$, we have 
\begin{align*}
&|P_{SUY}(\mathsf{A},\mathsf{B},\mathsf{C})-P_{S \widetilde{U} Y'}(\mathsf{A},\mathsf{B},\mathsf{C})| \nn \\&
\leq  |P_{SUY}((\mathsf{A} \cap \mathsf{E}_S),\mathsf{B},\mathsf{C})- P_{S\tilde{U}Y'}((\mathsf{A}
\cap \mathsf{E}_S),\mathsf{B},\mathsf{C})|    +
 |P_{SUY}((\mathsf{A} \cap \mathsf{E}_S^c),\mathsf{B},\mathsf{C})- P_{S\widetilde{U}Y'}((\mathsf{A}
\cap \mathsf{E}_S^c),\mathsf{B},\mathsf{C})| \nn
\\&\stackrel{(a)}{\leq} |P_{SUY}((\mathsf{A} \cap \mathsf{E}_S),\mathsf{B},\mathsf{C})- P_{S\widetilde{U}{Y}'}((\mathsf{A}
\cap \mathsf{E}_S),\mathsf{B},\mathsf{C})|    +
\\& \hspace{0.5in} |P_{SUY}((\mathsf{A} \cap \mathsf{E}_S^c),(\mathsf{B} \cap \mathsf{E}_U),\mathsf{C})- P_{S\widetilde{U}{Y}'}((\mathsf{A}
\cap \mathsf{E}_S^c),(\mathsf{B} \cap \mathsf{E}_U),\mathsf{C})| \nn
\\
&\leq P_{SUY} ((\mathsf{A} \cap \mathsf{E}_S) ,\mathsf{B},\mathsf{C})+
\\&\hspace{0.5in}
P_{S\widetilde{U}{Y}'}((\mathsf{A} \cap \mathsf{E}_S),\mathsf{B},\mathsf{C}) + P_{SUY}((\mathsf{A} \cap \mathsf{E}_S^c),(\mathsf{B}
\cap \mathsf{E}_U),\mathsf{C}) + P_{S\widetilde{U}Y'}((\mathsf{A}
\cap \mathsf{E}_S^c),(\mathsf{B} \cap \mathsf{E}_U),\mathsf{C})
\nn \\
&\leq 2P_S(\mathsf{E}_S) +2 P_U(\mathsf{E}_U) \rightarrow 0 \mbox{ as } \ell_S,u_S,\ell_U,u_U
\rightarrow \infty, \nn
\end{align*}
where in (a) we have used the fact that if $S,U\in \mathsf{E}^c_S\times  \mathsf{E}_U^c$, then $S=\widetilde{S}$ and $U=\widetilde{U}$ to conclude $P_{SUY}((\mathsf{A} \cap \mathsf{E}_S^c),(\mathsf{B} \cap \mathsf{E}^c_U),C)=P_{S\widetilde{U}Y'}((\mathsf{A}
\cap \mathsf{E}_S^c),(\mathsf{B} \cap \mathsf{E}^c_U),\mathsf{C})$.  
Using the lower semi-continuity of mutual information (
Lemma \ref{lem:conv_mutual_info_1}), we see that for any $\e>0$, we
have  
\beq
  I(\widetilde{U};{Y}') \geq I(U;Y)-\e.
\label{eq:ccsi_cont_cond_1}
\eeq
for sufficiently large $\ell_S,u_S,\ell_U,u_U$.
For the mutual information between $\widetilde{U}$ and the state
information $S$, we have for any $\e>0$,
\begin{align}
\nn I(\widetilde{U};S) &= I(\widetilde{U};S, \mathbbm{1}_{\{S \in \mathsf{E}_S^c\}})
=I(\widetilde{U};\mathbbm{1}_{\{S \in \mathsf{E}_S^c\}})+P(S \in \mathsf{E}_S)I(\widetilde{U};S|S
\in \mathsf{E}_S)+
\\& 
\nn P(S \in \mathsf{E}_S^c)I(\widetilde{U};S|S \in \mathsf{E}_S^c) \nn
\\&
\nn \overset{(a)}{\leq} H(P_S(\mathsf{E}_S))+ 0 + P(S \in \mathsf{E}_S^c) I(U;S|S \in \mathsf{E}_S^c) \nn 
\\&
\nn 
\leq H(P_S(\mathsf{E}_S)) +  P(S \in \mathsf{E}_S^c)I(U;S|S \in \mathsf{E}_S^c)+ P(S\in \mathsf{E}_S)I(U;S|S\in \mathsf{E}_S)
\\&
\nn 
= H(P_S(\mathsf{E}_S)) + I(U;S|\mathbbm{1}_{\{S \in \mathsf{E}_S^c\}})
\\&\leq
\epsilon+I(U;S),
\label{eq:ccsi_cont_cond_2}
\end{align}
for all sufficiently large $\ell_S,u_S$,
where (a) follows by noting that $\widetilde{U}$ is constant for $S\in \mathsf{E}_S$.
Consider the cost function:
\begin{align}
\mathbb{E}\tilde{\kappa}(S,\widetilde{U}) &=\int_{\mathsf{E}_S^c} P_S(s) \int
\tilde{\kappa}(s,u) P_{\widetilde{U}|S}(du|s) + \int_{\mathsf{E}_S} P_S(s)
\tilde{\kappa}(s,c) \nn \\
&\overset{(a)}{\leq} \int_{\mathsf{E}_S^c} P_S(s) \int \tilde{\kappa}(s,u)
P_{U|S}(du|s) + \e + \int_{\mathsf{E}_S} P_S(s)
\tilde{\kappa}(s,c) \nn \\ 
&\overset{(b)}{\leq} \int_{\mathsf{E}_S^c} P_S(s) \int \tilde{\kappa}(s,u)
P_{U|S}(du|s) + \e + \e \nn \\
&\leq \mathbb{E}\tilde{\kappa}(S,U)+2\e, 
\label{eq:ccsi_cont_cond_3}
\end{align}
where (a) follows for all sufficiently large $\ell_U,u_U$, 
by using the argument considered in
Theorem  \ref{thm:quantize_dist_cost},  
and (b) follows for all sufficiently large $\ell_S,u_S$ using the fact
that $\mathbb{E}\tilde{\kappa}(S,c)<\infty$. 

In summary we modified $P_{U|S}$ to get $P_{\widetilde{U}|S}$, and in the
process obtained $P_{S\widetilde{U}{Y}'}$ from $P_{SUY}$ without
changing the distribution of the state $P_S$ and the transition
probability of the channel $P_{Y|SU}$. 

\noindent \textbf{Step 2 (Discretizing U and Y):}
We discretize $\widetilde{U}$ with parameters $n_U$ to get $\widehat{U}$, and clip and discretize ${Y}'$ with
parameters $l_Y,u_Y,n_Y$. When $\widetilde{U}=c$ is observed, we let
$\widehat{U}=*$, which is a special symbol. 
Using this we create a discrete state
information and a discrete channel. Consider the induced distribution
on the triple  $(S,\widehat{U}, \widehat{Y})$
given as follows for all $\mathsf{A} \in \mathscr{B}$, and all $k$, and $j \neq
*$, 
\begin{align}
P(S \in \mathsf{A}, \widehat{U}=j, \widehat{Y}=k) 
&= \int_{\mathsf{A} \cap \mathsf{E}_S^c} P_S(ds) \int_{\mathsf{A}(j) \backslash \{c\}} P_{\widetilde{U}|S}(du|s) \int_{\mathsf{A}(k)}
P_{Y|US}(dy|u,s) \nn \\
&= \int_{\mathsf{A} \cap \mathsf{E}_S^c} P_S(ds) Q(j,s)   \int_{\mathsf{A}(j) \backslash \{c\}} \frac{P_{\widetilde{U}|S}(du|s)}{Q(j,s)} \int_{\mathsf{A}(k)}
P_{Y|US}(dy|u,s), \nn 
\end{align}
where $\mathsf{A}(\cdot)$ is as in Equations \eqref{eq:disc_cell_1}-\eqref{eq:disc_cell_3} and $Q(j,s)$ is the conditional probability of  $\widehat{U}$
given $S$, and is given more precisely as 
\[
Q(j,s)=\int_{\mathsf{A}(j) \backslash \{c\} } P_{\widetilde{U}|S}(du|s),
\]
and $\frac{P_{\widetilde{U}|S}(\cdot|s)}{Q(j,s)}$ denotes the conditional
distribution of $\widetilde{U}$ given $(S,\widehat{U})$.
Moreover since for all $\mathsf{A} \in \mathcal{B}$, and all $j \neq *$, we have 
\[
 P(S \in \mathsf{A}, \widehat{U}=j)=  \int_{\mathsf{A} \cap \mathsf{E}_S^c} P_S(ds) Q(j,s),
\]
one can now look at
\begin{align}
\widetilde{P}(k|s,j) \triangleq 
\int_{\mathsf{A}(j) \backslash \{c\}} \frac{P_{\widetilde{U}|S}(du|s)}{Q(j,s)}  \int_{\mathsf{A}(k)} P_{Y|US}(dy|u,s)
\label{eq:tilde_transition_probability}
\end{align}
as a transition probability (channel) with input $\widehat{U}$ and
$S$ and  output $\widehat{Y}$.
For the case $j=*$, we have
\begin{align}
P(S \in \mathsf{A}, \widehat{U}=j, \widehat{Y}=k) 
&= \int_{\mathsf{A}} P_S(ds) P_{\widetilde{U}|S}(\{c\}|s) \int_{\mathsf{A}(k)}
P_{Y|US}(dy|c,s) \nn \\
&= \int_{\mathsf{A}} P_S(ds) Q(j,s)   \int_{\mathsf{A}(k)}
P_{Y|US}(c,s,dy), \nn 
\end{align}
where 
\[
Q(j,s)=P_{\widetilde{U}|S}(\{c\}|s).
\]
Note that $Q(*,s)=1$ if $S \in \mathsf{E}_S$. Moreover, 
note that conditioned on  $\widehat{U}=*$, and $S=s$, we have $\widetilde{U}=c$
with probability $1$, and
\begin{align}
\widetilde{P}(k|s,*) \triangleq   \int_{\mathsf{A}(k)} P_{Y|US}(dy|c,s). 
\label{eq:tilde_transition_probability1}
\end{align}

When $S \in \mathsf{E}_S^c$, we see that $\widehat{U}$ is a quantized
version of $\widetilde{U}$ except that the number $c$ is mapped to a separate
symbol $*$. When $S \in \mathsf{E}_S$, then $\widetilde{U}=c$ and
$\widehat{U}=*$ with probability $1$. 
Hence, we have  
\beq
I(\widehat{U};S) \leq I(\widetilde{U};S)
\label{eq:ccsi_cont_cond_4}
\eeq
and  for any $\e>0$,
\begin{equation}
I(\widehat{U};\widehat{Y}) \geq  I(\widetilde{U};\widetilde{Y})-\e, \ \
\label{eq:ccsi_cont_cond_5}
\end{equation}
for infinitely many  $l_Y,u_Y$ (and approaching $\infty$) and all sufficiently large $n_U,n_Y$ using Theorems
\ref{thm:quantize0} and \ref{thm:quantize}. 

We further define  $\tilde{\kappa}(s,*) \triangleq \tilde{\kappa}(s,c)$ for
all $s$. The expected cost satisfies: 
\begin{align}
\mathbb{E} \tilde{\kappa}(S,\widehat{U}) &= \int_{\mathsf{E}_S^c} P_S(ds)
\sum_{j} Q(j,s) \tilde{\kappa}(s,j) + \int_{\mathsf{E}_S} P_S(ds)
\tilde{\kappa}(s,c) \nn \\
&\overset{(a)}{\leq}  \int_{\mathsf{E}_S^c} P_S(ds)
\int \tilde{\kappa}(s,u) P_{\widetilde{U}|S}(s,u) + \e+\e, 
\label{eq:ccsi_cont_cond_5a}
\end{align}
for any $\e>0$ and
 all sufficiently large $\ell_S,u_S,n_U$, where (a) follows by noting
(i) that $P_{S\widehat{U}}$ converges in distribution to $P_{S
  \widetilde{U}}$, and $\tilde{\kappa}$ is bounded and continuous
over $\mathsf{E}_S^c \times \mathsf{E}_U^c$, and (ii) $\mathbb{E}\tilde{\kappa}(s,c) < \infty$.

\noindent \textbf{Step 3 (Discretizing S):}
Now we discretize the state $S$ with parameters $n_S$ to get $\widehat{S}$. 
If $S \in \mathsf{E}_S$, then let $\widehat{S}=*$, otherwise,
$\widehat{S}=Q_{n_S}(C_{\ell_S,u_S}(S))$. 
Consider $\widehat{S}$ and $\widehat{U}$. Clearly
these two random variables are bounded, and this implies that 
$\tilde{\kappa}(\widehat{S}$,$\widehat{U})$ is bounded as well.
Note that the joint distribution of the tuple
$(S,\widetilde{U},\widehat{S},\widehat{U})$ can be characterized
using $P_{S,\widetilde{U}}$ and the functions
$\widehat{S}=Q_{n_S}(C_{\ell_S,u_S}(S))$, if $S \in \mathsf{E}_S^c$ and
$\widehat{S}=*$ otherwise,
 and $\widehat{U}=Q_{n_U}(\widetilde{U})$ if $\widetilde{U} \neq c$, and 
$\widehat{U}=*$ otherwise. 
In general, the triple
$(S,\widehat{S},\widehat{U})$ does  not satisfy the Markov
chain $S - \widehat{S} - \widehat{U}$, which is needed for the encoding operation described in the sequel. 
Toward addressing this, we consider the following triple of
random variables $(S,\widehat{S}, \widehat{U}')$, where 
we have $S - \widehat{S} - \widehat{U}'$, and
$(\widehat{S},\widehat{U}')$
has the same distribution as $(\widehat{S},\widehat{U})$.
Recall that conditioned on the event $\widehat{S}=*$, the random
variables $S$ and $\widehat{U}$ are independent (in fact
$\widehat{U}=*$ with probability $1$). Hence on this
event there will be no difference between $\widehat{U}$ and
$\widehat{U}'$. In other words, we have 
\[
P(S \in \mathsf{E}_S, \widehat{S}=*,\widehat{U}'=*)=P(S \in \mathsf{E}_S).
\]

Since $I(\widehat{S};\widehat{U})=I(\widehat{S};\widehat{U}')$, using the data processing inequality,
we have 
\begin{equation}
I(\widehat{S};\widehat{U}')= I(\widehat{S};\widehat{U}) \leq  I(S;\widehat{U}). \ \
\label{eq:ccs_cont_cond_6}
\end{equation}
For the cost function we define \[\tilde{\kappa}(*,*) \triangleq \frac{\int_{\mathsf{E}_S} P_S(ds)
\tilde{\kappa}(s,c)}{\int_{\mathsf{E}_S} P_S(ds)}.\]
Now the cost function satisfies: 
\begin{align}
\mathbb{E}\tilde{\kappa}(\widehat{S};\widehat{U}') &=\sum_{i \neq *} \sum_{j} P(\widehat{S}=i,\widehat{U}=j) \tilde{\kappa}(i,j) 
+ P(\widehat{S}=*, \widehat{U}=*)\tilde{\kappa}(*,*) \nn
\\&
=\sum_{i \neq *} \sum_{j} P(\widehat{S}=i,\widehat{U}=j) \tilde{\kappa}(i,j) 
+ \int_{\mathsf{E}_S} P_S(ds) \tilde{\kappa}(s,c) \nn \\  
&\overset{(a)}{\leq} \int_{\mathsf{E}_S} P_{S}(ds) \sum_j Q(j,s) \tilde{\kappa}(s,j) 
+\e+  \int_{\mathsf{E}_S} P_S(ds) \tilde{\kappa}(s,c),  \nn \\
&= \mathbb{E} \tilde{\kappa}(S,\widehat{U})+ \e,
\end{align}
where in (a) we have used the uniform continuity of $\tilde{\kappa}(\cdot,\cdot)$ over $\mathsf{E}_S$  and the fact that $\widehat{S}$ converges to $S$ in distribution using Lemma 
\ref{lem:conv_distribution}.

At this point we have the triple $(S,$ $\widehat{S}$,
$\widehat{U}')$ along with the Markov chain $S -
\widehat{S} - \widehat{U}'$. Let us denote the
conditional distribution of $\widehat{U}'$ given $S$ as 
$\widehat{Q}(j,s)$ for all $j$ and $s$. 
Consider the following
channel with input
$(S,\widehat{U}')$ and output
$\widehat{Y}'$ and transition probability $\widetilde{P}$ given
in \eqref{eq:tilde_transition_probability} and \eqref{eq:tilde_transition_probability1}. In other words, the joint
distribution of the triple
$(S,\widehat{U}',\widehat{Y}')$ is given by
for all $\mathsf{A} \in \mathcal{B}$ and all $j \neq *$ and $k$, 
\begin{align*}
&P(S \in \mathsf{A}, \widehat{U}'=j, \widehat{Y}'=k) 
=
\\&
=\int_{\mathsf{A} \cap \mathsf{E}_S^c} P_S(ds) P( \widehat{U}'=j, \widehat{Y}'=k|s)
\\& 
=\int_{\mathsf{A} \cap \mathsf{E}_S^c} P_S(ds) \widehat{Q}(j,s)  \widetilde{P}( \widehat{Y}'=k|s,\widehat{U}'=j)
\\&
=\int_{\mathsf{A} \cap \mathsf{E}_S^c} P_S(ds) \widehat{Q}(j,s) \int_{\mathsf{A}(j) \backslash \{c\}}
{P_{\widetilde{U}|S}(du|s, U\in \mathsf{A}(j))}   {P}_{\widehat{Y}'|U,S}(k|u,s)
\\&
=\int_{\mathsf{A} \cap \mathsf{E}_S^c} P_S(ds) \widehat{Q}(j,s) \int_{\mathsf{A}(j) \backslash \{c\}}
\frac{P_{\widetilde{U}|S}(du|s)}{Q(j,s)}  \int_{\mathsf{A}(k)} P_{Y|US}(dy|u,s),
\end{align*}
and 
\[
P(S \in \mathsf{A}, \widehat{U}'=*, \widehat{Y}'=k) 
=  \int_\mathsf{A} P_S(ds) \widehat{Q}(*,s) \int_{\mathsf{A}(k)} P_{Y|US}(dy|c,s),
\]
In summary, we took $P_{S \widehat{U}, \widehat{Y}}$
and modified it to obtain $P_{S \widehat{U}'
  \widehat{Y}'}$. 
Now we make the observation that 
\[
I(S;\widehat{U}|\widehat{S}) =
I(S;\widehat{U})-I(\widehat{S};\widehat{U}) < \e
\]
for all sufficiently large $\ell_S,u_S,n_S$ due to lower semi-continuity of mutual information and data processing inequality. Using Lemma
\ref{lem:mc_forced1}, and taking $(S,\widehat{S},\widehat{U})$ as the variables $(C,B,A)$, respectively, and taking $D,E$ to be trivial in the lemma, 
we note that 
\[
D(P_{S \widehat{U}} \| P_{S \widehat{U}'}) \leq \e, 
\]
and consequently, 
\[
D(P_{\widehat{U} \widehat{Y}} \|
P_{\widehat{U}',\widehat{Y}'})  \leq D(P_{S \widehat{U} \widehat{Y}} \|
P_{S \widehat{U}',\widehat{Y}'}) =
D(P_{S \widehat{U}}|| 
P_{S \widehat{U}'})+
D(P_{\widehat{Y}|S \widehat{U} } \|
P_{\widehat{Y}'|S \widehat{U}'})\stackrel{(a)}{=}
D(P_{S
  \widehat{U}} \| P_{S \widehat{U}'}) \leq \e, 
\]
where in (a) we have used the fact that the channel from $(\widehat{U},S)$ to $\widehat{Y}$ is the same as the one from $(\widehat{U}',S)$ to $\widehat{Y}'$, and the last inequality holds 
for all sufficiently large $\ell_S,u_S,n_S$. 
Hence, using lower semi-continuity of mututal information, we get for any $\e>0$, 
\beq
I(\widehat{U}';\widehat{Y}') \geq
I(\widehat{U};\widehat{Y}) -\e,
\label{eq:ccsi_cont_cond_7}
\eeq
for all sufficiently large $\ell_S,u_S,n_S$. 

\noindent \textbf{Step 4 (Discrete Coding Scheme):} 
We give a description of the coding scheme that is
based on the above observations and the direct coding theorem of
channel coding with side information in the discrete case.  
We can use  the Gelfand-Pinsker Theorem \cite{Gelfand} to show the existence of a transmission system.
Recall that we now have a discrete distribution that relates three 
random variables $\widehat{S}$, $\widehat{U}'$ and
$\widehat{Y}'$. Moreover we have 
$P(\widehat{U}'=*|\widehat{S}=*)=1$.
We use the Gelfand-Pinsker Theorem \cite{Gelfand} and show that the rate-cost pair
given  by 
\[
(I(\widehat{U}';\widehat{Y}')-I(\widehat{S};\widehat{U}'),\mathbb{E}\tilde{\kappa}(\widehat{S},\widehat{U}'))
\]
 is achievable for the synthesized discrete channel. 
In other words, for every $\e>0$, for all sufficiently large $m$, there exists a transmission system (say
$\mbox{TS}_d$) with parameter $(m,\Theta)$ for transmission of information such that 
\beq
\frac{1}{m} \log \Theta \geq
I(\widehat{U}';\widehat{Y}')-I(\widehat{S};\widehat{U}')-\e,
\label{eq:ccsi_cont_cond_8}
\eeq
and 
\begin{equation}
\frac{1}{m}\sum_{i=1}^m \mathbb{E}\tilde{\kappa}(\widehat{S}_{i},\widehat{U}_{i})
\leq \mathbb{E}\kappa(\widehat{S},\widehat{U}')+\e,
\label{eq:ccsi_cont_cond_9}
\end{equation}
where $\widehat{U}^{'m}\triangleq e_1(M,\widehat{S}^m)$ denote the input vector of the discrete
channel, $M$ is the random message, and $e_1$ is the first part of the
encoder of $\mbox{TS}_d$\footnote{Note that $\widehat{S}^m$ is an IID extension of $\widehat{S}$ , whereas $\widehat{U}^{'m}$ is not an IID extension of $\widehat{U}'$ as it is a codeword and has memory. }. It should be noted that we
only need the first part of the encoder of $\mbox{TS}_d$ which produces
the $U$-codeword. The second part which produces the channel input $X$
is not needed here.  It is worth mentioning here that with arbitrarily
high probability the transmission system $\mbox{TS}_d$ will ensure
that $\widehat{U}^{'m}$ is jointly typical with $\widehat{S}^m$. 
When they are jointly typical, if for any $1 \leq i \leq m$, 
we find that $\widehat{S}_{i}=*$, then we also find that $\widehat{U}'_{i}=*$. 
In the complement event (where the encoder declares an encoding error
event), $\mbox{TS}_d$ does not care what the channel is.  

For the continuous-valued channel we obtain an $(m,\Theta)$  transmission system 
$\mbox{TS}_c$ by using $\mbox{TS}_d$. The encoder of  $\mbox{TS}_c$
has three stages. The first stage clips and quantizes $S^m$ to obtain
$\widehat{S}^m$.  The second stage uses the encoder of
$\mbox{TS}_d$ to generate $\widehat{U}^{'m}$. The third stage generates
$\widetilde{U}'^m$ from $S^m$ and $\widehat{U}^m$ using a random
transformation corresponding to the $m$-product of the following
synthetic channel: $\frac{P_{\widetilde{U}|S}(\cdot|s)}{Q(j,s)}$.
Note that $\widetilde{U}'$ has the same distribution as 
$\widetilde{U}$ created earlier.
In other words given $S_i=s$ and $\widehat{U}'_{i}=j$, and $j \neq *$, the probability of
the event $\widetilde{U}'_{i} \in \mathsf{A}$ is given by 
\[
\int_{\mathsf{A} \cap (\mathsf{A}(j) \backslash \{c\})} \frac{P_{\widetilde{U}|S}(du|s)}{Q(j,s)},
\]
for any borel set $\mathsf{A}$. Given $S_i=s$ and $\widehat{U}'_{i}=*$,  we have 
$\widetilde{U}'_{i}=c$ with probability $1$.
The channel input is given by ${X}^m=g^m(S^m,\widetilde{U}'^m)$.
When an encoding error is declared by
$\mbox{TS}_d$ we need to make sure we let $\widehat{U}^{'m}=[c,c,\ldots,c]$, i.e., 
all-$c$ vector. 
Let the channel output be given by $\widetilde{Y}^m$.

Next let us consider the decoder. The decoder of $TS_c$ has again two
stages. The first stage discretizes the channel output to produce
$\widehat{Y}^m$. The second stage is the decoder of
$\mbox{TS}_d$. We need to make sure that the decoder of $\mbox{TS}_d$
sees the channel that was promised. First of all the state
$\widehat{S}^m$ is memoryless and, of course,  each sample has the promised
distribution $P_{\widehat{S}}$. The channel seen by the
$\mbox{TS}_d$ is given by for all $i \neq *$, $j \neq *$, and all $k$
\begin{align}
&P(\widehat{Y}=k|\widehat{S}=i,\widehat{U}'=j) \nn \\
&= \int_{\mathsf{A}_{\ell_S,u_S,n_S}(i) \backslash \mathsf{E}_S }
\frac{P_S(ds)}{P_S(\mathsf{A}_{\ell_S,u_S,n_S}(i) \backslash \mathsf{E}_S)}
\int_{\mathsf{A}_{\ell_U,u_U,n_U}(j) \backslash \{c\}} \frac{P_{\widetilde{U}|S}(s,du)}{Q(j,s)}
\int_{\mathsf{A}_{\ell_Y,u_Y,n_Y}(k)} P_{Y|US}(u,s,dy),
\end{align}
\begin{align}
P(\widehat{Y}=k|\widehat{S}=i,\widehat{U}'=*)  &= \int_{\mathsf{A}_{\ell_S,u_S,n_S}(i) \backslash \mathsf{E}_S }
\frac{P_S(ds)}{P_S(\mathsf{A}_{\ell_S,u_S,n_S}(i) \backslash \mathsf{E}_S)}
\int_{\mathsf{A}_{\ell_Y,u_Y,n_Y}(k)}P_{Y|US}(dy|c,s),
\end{align}
\begin{align}
P(\widehat{Y}=k|\widehat{S}=*,\widehat{U}'=*) &= \int_{\mathsf{E}_S }
\frac{P_S(ds)}{P_S(\mathsf{E}_S)}
\int_{\mathsf{A}_{\ell_Y,u_Y,n_Y}(k)} P_{Y|US}(dy|c,s).
\end{align}
Hence the decoder of $\mbox{TS}_d$ sees the intended channel. The
message is reconstructed reliably at the decoder. 
From this we see that the rate given by 
$I(\widehat{U}';\widehat{Y})-I(\widehat{S};\widehat{U}')$ 
is achievable. 

\noindent \textbf{Step 5 (Checking the cost constraint):}
Fix $\e>0$. Assume that the ten parameters
$\ell_S,\ell_U,\ell_Y,u_S,u_U,u_Y,$ $n_S,n_U,n_Y$ and $m$  are such that
\begin{itemize}
\item  (\ref{eq:ccsi_cont_cond_1}-\ref{eq:ccsi_cont_cond_3}), and (\ref{eq:ccsi_cont_cond_4}-\ref{eq:ccsi_cont_cond_9}) 
  are satisfied.
\item $\int_{\mathsf{E}_S} \tilde{\kappa}(s,c) P_S(ds) <\e$.
\item 
$|\tilde{\kappa}(s_1,u_2)-\tilde{\kappa}(s_2,u_2)| \leq \e$, 
for all (a)  $s_1,s_2 \in [-\ell_S,u_S]$, (b) $u_1,u_2 \in [-\ell_U,u_U]$,
(c) $|s_1-s_2| \leq \frac{1}{2^{n_S}}$, and 
(d) $|u_1-u_2| \leq \frac{1}{2^{n_U}}$.
\end{itemize}

We see that
\[
\frac{1}{m} \log \Theta \geq I(U;Y)-I(U;S) - 5 \e.
\]
Next we ensure that the cost constraint is satisfied. 
Consider for any $t \in
\{1,2,\ldots,m\}$, 
\begin{align}
\mathbb{E}\kappa(S_t,\widetilde{U}'_{t}) &=  
\sum_{i \neq *} \sum_{j \neq *} 
 P(\widehat{U}'_{t}=j, \widehat{S}_{t}=i) 
 \int_{\mathsf{A}(i)
   \backslash \mathsf{E}_S } \int_{\mathsf{A}(j) \backslash \{c\}} \tilde{\kappa}(s,u)
 \frac{P_{S}(ds)}{P(S \in \mathsf{A}(i) \backslash \mathsf{E}_S)} \frac{P_{\widetilde{U}|S}(du|s)}{Q(j,s)}  \nn \\
&  + \sum_{i \neq *}  
 P(\widehat{U}'_{t}=*, \widehat{S}_{t}=i) \int_{\mathsf{A}(i)
   \backslash \mathsf{E}_S } \tilde{\kappa}(s,c)
 \frac{P_{S}(ds)}{P(S \in \mathsf{A}(i) \backslash \mathsf{E}_S)} 
 +  P(\widehat{S}_{t}=*) \int_{\mathsf{E}_S } \tilde{\kappa}(s,c)
 \frac{P_{S}(ds)}{P(S \in \mathsf{E}_S)}   \nn \\
&\leq  \sum_{i \neq *} \sum_{j \neq *} 
 P(\widehat{U}'_{t}=j, \widehat{S}=i)
 (\tilde{\kappa}(i,j)+\e) +  \sum_{i \neq *}  P(\widehat{U}'_{t}=*,
 \widehat{S}_{t}=i) (\tilde{\kappa}(i,c)+\e) \nn \\
&  +\int_{\mathsf{E}_S}
 \tilde{\kappa}(s,c)P_s(ds) \nn \\
&= \mathbb{E} \kappa(\widehat{S}_{t},\widehat{U}'_{t})+3\e.
\end{align}
Hence, we have 
\[
\frac{1}{m} \sum_{i=1}^m \mathbb{E}\kappa(S_{t},\widetilde{U}'_{i}) \leq
\frac{1}{m} \sum_{i=1}^m
\mathbb{E}\kappa(\widehat{S}_{t},\widehat{U}'_{t})+3\e \leq \mathbb{E}\tilde{\kappa}(S,U)+7\e.
\]
The desired achievability result follows.

\qed

\section{Proof of Theorem \ref{th:7}}
\label{App:th:7}
\begin{proof}
\textbf{Step 1 (Clipping).}
In this step, we show that for all $\zeta>0$ there exists $\ell,\ell'$ for which we have:
\begin{align}
        \label{eq:2:3.1.1}&
       | I(X;  \widetilde{U}_{\ell})- I(X;U)|\leq \zeta\\
    \label{eq:2:3.1.2}&
   | I(Y; \widetilde{V}_{\ell'})-  I(Y;V)|\leq \zeta\\ 
    \label{eq:2:3.1.3}&
    |I(\widetilde{U}_{\ell}+\widetilde{V}_{\ell'}; \widetilde{U}_{\ell})- I(U+V;U)|\leq \zeta
    \\ 
    \label{eq:2:3.1.4}&
    |I( \widetilde{U}_{\ell}+ \widetilde{V}_{\ell'}; \widetilde{V}_{\ell'} )- I(U+V;V)|\leq \zeta
    \\
    \label{eq:2:3.1.5}&
|I( \widetilde{U}_{\ell}; \widetilde{V}_{\ell'})-  I(U;V)|\leq \zeta,
\end{align}

 Equations \eqref{eq:2:3.1.1}, \eqref{eq:2:3.1.2}, and \eqref{eq:2:3.1.5} follow from the data processing inequality and lower semi-continuity of mutual information. We will show the result for \eqref{eq:2:3.1.3}. The proof for \eqref{eq:2:3.1.4} follows a similar argument.


 Define $A_{U,\ell}$ as the indicator of $U\in [-\ell,\ell]$ and 
 $B_{V,\ell'}$ as the indicator of $V\in [-\ell',\ell']$. 
Consider the following arguments:
\begin{align*}
&I(\widetilde{U}_\ell;\widetilde{U}_\ell+\widetilde{V}_{\ell'})\leq I(A_{U,\ell},\widetilde{U}_\ell;\widetilde{U}_\ell+\widetilde{V}_{\ell'})
\\&
\leq 
H(A_{U,\ell})+ P(A_{U,\ell}=0)I(\widetilde{U}_{\ell};\widetilde{U}_{\ell}+\widetilde{V}_{\ell'}|A_{U,\ell}=0)+ P(A_{U,\ell}=1)I(\widetilde{U}_{\ell};\widetilde{U}_{\ell}+\widetilde{V}_{\ell'}|A_{U,\ell}=1)
\\&
\leq 
H(A_{U,\ell})+ P(A_{U,\ell}=0)I(\widetilde{U}_{\ell};\widetilde{U}_{\ell}+\widetilde{V}_{\ell'}, B_{V,\ell'}|A_{U,\ell}=0)+ P(A_{U,\ell}=1)I(\widetilde{U}_{\ell};\widetilde{U}_{\ell}+\widetilde{V}_{\ell'}|A_{U,\ell}=1)
\\& \leq 
H(A_{U,\ell})+ H(B_{V,\ell'})+ P(B_{V,\ell'}=0|A_{U,\ell}=0)P(A_{U,\ell}=0)I(\widetilde{U}_{\ell};\widetilde{U}_{\ell}
+\widetilde{V}_{\ell'}|A_{U,\ell}=0,B_{V,\ell'}=0)
\\&\hspace{0.2in}+P(B_{V,\ell'}=1|A_{U,\ell}=0)P(A_{U,\ell}=0)I(\widetilde{U}_{\ell};\widetilde{U}_{\ell}+\widetilde{V}_{\ell'}|A_{U,\ell}=0,B_{V,\ell'}=1)
\\& \hspace{0.2in} +P(A_{U,\ell}=1)I(\widetilde{U}_\ell;\widetilde{U}_\ell+\widetilde{V}_{\ell'}|A_{U,\ell}=1)
\\& \leq 
H(A_{U,\ell})+ H(B_{V,\ell'})+ P(B_{V,\ell'}=0|A_{U,\ell}=0)P(A_{U,\ell}=0)I(U';U'
+V'|A_{U,\ell}=0,B_{V,\ell'}=0)
\\& \hspace{0.2in} +P(B_{V,\ell'}=1|A_{U,\ell}=0)P(A_{U,\ell}=0)I(U';U'+V|A_{U,\ell}=0,B_{V,\ell'}=1)
\\& \hspace{0.2in}+P(A_{U,\ell}=1)I(U;U+\widetilde{V}_{\ell'}|A_{U,\ell}=1)
\end{align*}
\begin{align*}
& \stackrel{(a)}{\leq} 
P(B_{V,\ell'}=0|A_{U,\ell}=0)P(A_{U,\ell}=0)I(U';U'
+V'|A_{U,\ell}=0,B_{V,\ell'}=0)
\\& \hspace{0.2in} +P(B_{V,\ell'}=1|A_{U,\ell}=0)P(A_{U,\ell}=0)I(U';U'+V|A_{U,\ell}=0,B_{V,\ell'}=1)
\\& \hspace{0.2in}+P(A_{U,\ell}=1)I(U;U+\widetilde{V}_{\ell'}|A_{U,\ell}=1)+2\eta_1
\\&
\stackrel{(b)}{\leq} I(U;U+\widetilde{V}_{\ell'}|A_{U,\ell}=1)+ \eta_{1} +\gamma_1
\\& \stackrel{(c)}{\leq} 
I(A_{U,\ell};U+\widetilde{V}_{\ell'})+P(A_{U,\ell}=1)I(U;U+\widetilde{V}_{\ell'}|A_{U,\ell}=1)+
\\&
+P(A_{U,\ell}=0)I(U;U+\widetilde{V}_{\ell'}|A_{U,\ell}=0)
 + 2\eta_{1} +\gamma_1
\\& 
=I(U;U+\widetilde{V}_{\ell'})+ 2\eta_{1} +\gamma_1,
\end{align*}
where in (a) we have taken $\ell$ and $\ell'$ large enough such that $H(A_{U,\ell})<\eta_1$ and $H(B_{V,\ell'})<\eta_1$, and in (b) we have defined
\begin{align*}
&\gamma_1\triangleq 
P(A_{U,\ell}=0,B_{V,\ell'}=0)I(U';U'
+V'|A_{U,\ell}=0,B_{V,\ell'}=0)
\\&\qquad  +P(A_{U,\ell}=0,B_{V,\ell'}=1)I(U';U'+V|A_{U,\ell}=0,B_{V,\ell'}=1).
\end{align*}
We consider the first mutual information term as follows. 
\begin{align*}
I(U';U'+V'|A_{U,\ell}=0, B_{U,\ell'}=0)= 
I(U';U'+V') = h(U'+V')-h(V')
\end{align*}
We have:
\begin{align*}
&    h(V') =h(V|B_{V,\ell'}=1)
= \frac{1}{P(B_{V,\ell'}=1)}( h(V)- P(B_{V,\ell'}=0)h(V|B_{V,\ell'}=0))
\end{align*}
and $h(V|B_{V,\ell'}=0)$ can be bounded from above as follows:
\begin{align*}
& h(V|B_{V,\ell'}=0) \leq \frac{1}{2}\log{2\pi e \mbox{Var}(V|B_{V,\ell'}=0)}\stackrel{(a)}\leq
\frac{1}{2}\log{2\pi e \frac{1}{P(B_{V,\ell'}=0)} \mbox{Var}(V)},
\end{align*}
where in (a) we have used the law of total variance.
This implies that $h(V') \geq h(V)$ as $\ell\to \infty$. So, $h(V')$ is bounded from below. Next, we bound $h(U'+V')$ from above:
\begin{align*}
& h(U'+V')\leq \frac{1}{2}\log{2\pi e \mbox{Var}(U'+V')},
\\& \mbox{Var}(U'+V')= \mbox{Var}(U')+\mbox{Var}(V')= 
\mbox{Var}(U|A_{U,\ell}=1)+\mbox{Var}(V|B_{V,\ell'}=1)
\\\qquad &\leq
\frac{1}{P(A_{U,\ell}=1)}\mbox{Var}(U)+ 
\frac{1}{P(B_{V,\ell'}=1)}\mbox{Var}(V)
\leq \frac{1}{1-\eta_1} (\mbox{Var}(U)+\mbox{Var}(V)).
\end{align*}
As a result, we have
\begin{align*}
& 
P(A_{U,\ell}=0,B_{V,\ell'}=0)I(U';U'
+V'|A_{U,\ell}=0,B_{V,\ell'}=0)
\\&\leq  P(A_{U,\ell}=0,B_{V,\ell'}=0)\left( 
 \frac{1}{2}\log{2\pi e  \frac{1}{1-\eta_1} \left(\mbox{Var}(U)+\mbox{Var}(V)\right)})
 - h(V)\right)\to 0
\end{align*}
as $\ell \to \infty$. 

Next, we consider the second mutual information term $I(U';U'+V|A_{U,\ell}=0,B_{V,\ell'}=1)$.
Note that:
\begin{align*}
    & h(U'+V|A_{U,\ell}=0,B_{V,\ell'}=1)
     \leq \frac{1}{2}\log{2\pi e \mbox{Var}(U'+V|A_{U,\ell}=0,B_{V,\ell'}=1)}
     \\& = \frac{1}{2}\log{2\pi e \mbox(\mbox{Var}(U')+\mbox{Var}(V|A_{U,\ell}=0,B_{V,\ell'}=1))}
     \\&= \frac{1}{2}\log{2\pi e( \mbox{Var}(U|A_{U,\ell}=1)}+\mbox{Var}(V|A_{U,\ell}=0,B_{V,\ell'}=1))
     \\& \leq 
     \frac{1}{2}\log{2\pi e\left(
     \frac{\mbox{Var}(U)}{P(A_{U,\ell}=1)}+\frac{\mbox{Var}(V)}{P(A_{U,\ell}=0,B_{V,\ell'}=1)}\right)}
\end{align*}
So, $P(A_{U,\ell}=0,B_{V,\ell'}=1)h(U'+V|A_{U,\ell}=0,B_{V,\ell'}=1)\to 0$ as $\ell,\ell'\to \infty$. Also, $P(A_{U,\ell}=0,B_{V,\ell'}=1)h(V|A_{U,\ell}=0,B_{V,\ell'}=1) \to 0$ as $\ell,\ell\to \infty$ as shown above. Hence, $P(A_{U,\ell}=0,B_{V,\ell'}=1)I(U';U'+V|A_{U,\ell}=0,B_{V,\ell'}=1)\to 0$ as $\ell,\ell' \to \infty$ and consequently $\gamma_1\to 0$ as $\ell,\ell' \to \infty$.  
Next we focus on $I(U;U+\widetilde{V}_{\ell'})$:
\begin{align*}
 & I(U;U+\widetilde{V}_{\ell'})
  \leq I(U;U+\widetilde{V}_{\ell'},B_{V,\ell'}) \\&
  \leq H(B_{V,\ell'})+
  P(B_{V,\ell'}=0)I(U;U+V'|B_{V,\ell'}=0)
  +  P(B_{V,\ell'}=1)I(U;U+\widetilde{V}_{\ell'}|B_{V,\ell'}=1)
  \\&\leq P(B_{V,\ell'}=1)I(U;U+\widetilde{V}_{\ell'}|B_{V,\ell'}=1)+\gamma_2
\end{align*}
where
\begin{align*}
    \gamma_2= H(B_{V,\ell'})+
  P(B_{V,\ell'}=0)I(U;U+V'|B_{V,\ell'}=0).
\end{align*}
Note that:
\begin{align*}
& h(U+V'|B_{V,\ell'}=0)\leq 
  \frac{1}{2}\log 2\pi e \mbox{Var}(U+V'|B_{V,\ell'}=0)
  \\& \leq   \frac{1}{2}\log 2\pi e( \mbox{Var}(U|B_{V,\ell'}=0)+\mbox{Var}(V'|B_{V,\ell'}=0))
  \\& \leq 
  \frac{1}{2}\log 2\pi e\left( \frac{\mbox{Var}(U)}{P(B_{V,\ell'}=0)}+\frac{\mbox{Var(V)}}{P(B_{V,\ell'}=1)}\right)
\end{align*}
So, $P(B_{V,\ell'}=0)h(U+V'|B_{V,\ell'}=0)\to 0$ as $\ell,\ell' \to \infty$, and in turn using the 
arguments used above regarding $h(V')$ we infer that $\gamma_2\to 0$ as $\ell,\ell'\to \infty$.

\noindent \textbf{Step 2 (Smoothing).} In this step, we show that for all $\gamma>0$ there exists $\epsilon(\gamma)\in \mathbb{R}$ such that for all $\epsilon\leq \epsilon(\gamma)$, we have:
\begin{align}
   &
        |I(X;  \widetilde{U}_{\ell,\epsilon})-I(X;\widetilde{U}_\ell)|\leq \gamma\\
&
   | I(Y; \widetilde{V}_{\ell,\epsilon})- I(Y;\widetilde{V}_\ell)|\leq \gamma\\
&
    |I(\widetilde{U}_{\ell,\epsilon}+\widetilde{V}_{\ell,\epsilon}; \widetilde{U}_{\ell,\epsilon})- I(\widetilde{U}_\ell+\widetilde{V}_\ell;\widetilde{V}_\ell)|\leq \gamma
    \\ 
 &
   |I( \widetilde{U}_{\ell,\epsilon}+ \widetilde{V}_{\ell,\epsilon}; \widetilde{V}_{\ell,\epsilon} )- I(\widetilde{U}_\ell+\widetilde{V}_\ell;\widetilde{V}_\ell)|\leq \gamma
    \\
 &
|I( \widetilde{U}_{\ell,\epsilon}; \widetilde{V}_{\ell,\epsilon})- I(\widetilde{U}_\ell;\widetilde{V}_\ell)|\leq \gamma,
\end{align}
\noindent 
We argue that $F_{X,Y,\widetilde{U}_{\ell,\epsilon},\widetilde{V}_{\ell,\epsilon}} \to F_{X,Y,\widetilde{U}_{\ell},\widetilde{V}_{\ell}}$ as $\epsilon\to 0$. We show convergence for $F_{\widetilde{U}_{\ell,\epsilon}}$. The convergence for the joint distribution follows by similar arguments. To show this,
let $u\in [-\ell,\ell]$ be a point of continuity of $F_{\widetilde{U}_{\ell}}(\cdot)$. Then,
by construction, we have $F_{\widetilde{U}_{\ell,\epsilon}}(u)\to F_{\widetilde{U}_{\ell}}(u)$ as $\epsilon\to 0$ as shown below:
\begin{align*}
   & P(\widetilde{U}_\ell \leq  u-\epsilon)\leq P(\widetilde{U}_{
   \ell,\epsilon}\leq u) \leq     P(\widetilde{U}_\ell\leq u+\epsilon)
   \\& 
      \Rightarrow  P(\widetilde{U}_\ell \leq  u-\epsilon)-F_{\widetilde{U}_{\ell}}(u)\leq P(\widetilde{U}_{
   \ell,\epsilon}\leq u) -F_{\widetilde{U}_\ell}(u)\leq     P(\widetilde{U}_\ell\leq u+\epsilon)-F_{\widetilde{U}_\ell}(u)
      \\&\Rightarrow  
      P(\widetilde{U}_\ell \leq  u-\epsilon)-F_{\widetilde{U}_\ell}(u+\epsilon)\leq P(\widetilde{U}_{
   \ell,\epsilon}\leq u) -F_{\widetilde{U}_{\ell}}(u)\leq     P(\widetilde{U}_\ell\leq u+\epsilon)-F_{\widetilde{U}_{\ell}}(u-\epsilon)
    \\&
\Rightarrow    |F_{\widetilde{U}_{\ell,\epsilon}}(u)- F_{
\widetilde{U}_\ell}(u)| \leq P(u-\epsilon <  \widetilde{U}_{\ell} \leq u + \epsilon).
\end{align*}

This convergence result along with the lower semi-continuity of mutual information implies \eqref{eq:3:5.1.1}, \eqref{eq:3:5.1.2}, and \eqref{eq:3:5.1.5}. 
Towards showing \eqref{eq:3:5.1.3}, note that 
\begin{align*}
    &I(\widetilde{U}_\ell;\widetilde{U}_\ell+\widetilde{V}_\ell)= I(\widetilde{U}_\ell;\widetilde{U}_\ell+\widetilde{V}_\ell|\widetilde{N}_{\ell,\epsilon},\widetilde{N}'_{\ell,\epsilon})
   \\& = I(\widetilde{U}_{\ell,\epsilon};\widetilde{U}_{\ell,\epsilon}+\widetilde{V}_{\ell,\epsilon}|\widetilde{N}_{\ell,\epsilon},\widetilde{N}'_{\ell,\epsilon})
\\&=  I(\widetilde{U}_{\ell,\epsilon},\widetilde{N}_{\ell,\epsilon},\widetilde{N}'_{\ell,\epsilon};\widetilde{U}_{\ell,\epsilon}+\widetilde{V}_{\ell,\epsilon})-
I(\widetilde{N}_{\ell,\epsilon},\widetilde{N}'_{\ell,\epsilon};\widetilde{U}_{\ell,\epsilon}+\widetilde{V}_{\ell,\epsilon})
\\&\geq  I(\widetilde{U}_{\ell,\epsilon},\widetilde{N}_{\ell,\epsilon},\widetilde{N}'_{\ell,\epsilon};\widetilde{U}_{\ell,\epsilon}+\widetilde{V}_{\ell,\epsilon})-
I(\widetilde{N}_{\ell,\epsilon},\widetilde{N}'_{\ell,\epsilon};\widetilde{U}_{\ell,\epsilon},\widetilde{V}_{\ell,\epsilon})
\\&= I(\widetilde{U}_{\ell,\epsilon},\widetilde{N}_{\ell,\epsilon},\widetilde{N}'_{\ell,\epsilon};\widetilde{U}_{\ell,\epsilon}+\widetilde{V}_{\ell,\epsilon})-
I(\widetilde{N}_{\ell,\epsilon},\widetilde{N}'_{\ell,\epsilon};\widetilde{U}_{\ell,\epsilon})
- I(\widetilde{N}_{\ell,\epsilon},\widetilde{N}'_{\ell,\epsilon};\widetilde{V}_{\ell,\epsilon}|\widetilde{U}_{\ell,\epsilon})
\\& = I(\widetilde{U}_{\ell,\epsilon},\widetilde{N}_{\ell,\epsilon},\widetilde{N}'_{\ell,\epsilon};\widetilde{U}_{\ell,\epsilon}+\widetilde{V}_{\ell,\epsilon})-I(\widetilde{N}_{\ell,\epsilon};\widetilde{U}_{\ell,\epsilon})-
I(\widetilde{N}_{\ell,\epsilon},\widetilde{N}'_{\ell,\epsilon};\widetilde{V}_{\ell,\epsilon}|\widetilde{U}_{\ell,\epsilon})
\\& = I(\widetilde{U}_{\ell,\epsilon};\widetilde{U}_{\ell,\epsilon}+\widetilde{V}_{\ell,\epsilon})+ I(\widetilde{N}_{\ell,\epsilon},\widetilde{N}'_{\ell,\epsilon};\widetilde{U}_{\ell,\epsilon}+\widetilde{V}_{\ell,\epsilon}| \widetilde{U}_{\ell,\epsilon} ) -I(\widetilde{N}_{\ell,\epsilon};\widetilde{U}_{\ell,\epsilon})-
I(\widetilde{N}_{\ell,\epsilon},\widetilde{N}'_{\ell,\epsilon};\widetilde{V}_{\ell,\epsilon}|\widetilde{U}_{\ell,\epsilon})
  \\&= 
 I(\widetilde{U}_{\ell,\epsilon};\widetilde{U}_{\ell,\epsilon}+\widetilde{V}_{\ell,\epsilon})-I(\widetilde{N}_{\ell,\epsilon};\widetilde{U}_{\ell,\epsilon}),
\end{align*}
where $I(\widetilde{N}_{\ell,\epsilon};\widetilde{U}_{\ell,\epsilon})\to 0$ as $\epsilon \to 0$ by Lemma \ref{lem:smoothing}. The proof of \eqref{eq:3:5.1.3} follows by lower semi-continuity of mutual information. 
Equation \eqref{eq:3:5.1.4} can be proved using a similar argument. Note that the joint PDF of $(\widetilde{U}_{\ell,\epsilon},\widetilde{V}_{\ell,\epsilon})$ is jointly continuous on a compact support and hence uniformly continuous.

\textbf{Step 3 (Discretization).} 
In this step, we discretize $\widetilde{U}_{\ell,\epsilon}, \widetilde{V}_{\ell,\epsilon}$ to $\widehat{U}_{\ell,\epsilon,n}, \widehat{V}_{\ell,\epsilon,n}$ by applying $Q_n(\cdot)$. We show that:
\begin{align}
        \label{eq:3:5.1.1}&
        |I(X;  \widehat{U}_{\ell,\epsilon,n})-I(X;\widetilde{U}_{\ell,\epsilon})|\leq \gamma\\
    \label{eq:3:5.1.2}&
   | I(Y; \widehat{V}_{\ell,\epsilon,n})- I(Y;\widetilde{V}_{\ell,\epsilon})|\leq \gamma\\
    \label{eq:3:5.1.3}&
    |I(\widehat{U}_{\ell,\epsilon,n}+\widehat{V}_{\ell,\epsilon,n}; \widehat{U}_{\ell,\epsilon,n})- I(\widetilde{U}_{\ell,\epsilon}+\widetilde{V}_{\ell,\epsilon};\widetilde{V}_{\ell,\epsilon})|\leq \gamma
    \\ 
    \label{eq:3:5.1.4}&
   |I( \widehat{U}_{\ell,\epsilon,n}+ \widehat{V}_{\ell,\epsilon,n}; \widehat{V}_{\ell,\epsilon,n} )- I(\widetilde{U}_{\ell,\epsilon}+\widetilde{V}_{\ell,\epsilon};\widetilde{V}_{\ell,\epsilon})|\leq \gamma
    \\
    \label{eq:3:5.1.5}&
|I( \widehat{U}_{\ell,\epsilon,n}; \widehat{V}_{\ell,\epsilon,n})- I(\widetilde{U}_{\ell,\epsilon};\widetilde{V}_{\ell,\epsilon})|\leq \gamma,
\end{align}

Note that Equations \eqref{eq:3:5.1.1}, \eqref{eq:3:5.1.2}, and \eqref{eq:3:5.1.5} follow from data processing inequality and lower semi-continuity of mutual information. We will show Equation \eqref{eq:3:5.1.3}. The proof of \eqref{eq:3:5.1.4} follows by a similar argument.

The next step uses the approach taken in \cite{makkuva2018equivalence} to study entropy of linear combinations of continuous variables.
In the following, 
we drop the subscript on $Q_n(\cdot)$ when there is no ambiguity.
 Define $\mod_{\!\!Q}(\widetilde{U}_{\ell,\epsilon}) \triangleq \widetilde{U}_{\ell,\epsilon}- \widehat{U}_{\ell,\epsilon,n}$ and $\mod_{\!\!Q}(\widetilde{V}_{\ell,\epsilon}) \triangleq \widetilde{V}_{\ell,\epsilon}- \widehat{V}_{\ell,\epsilon,n}$, and the variables $C\triangleq Q(\widetilde{U}_{\ell,\epsilon}+\widetilde{V}_{\ell,\epsilon})$, $D\triangleq \widehat{U}_{\ell,\epsilon,n}+\widehat{V}_{\ell,\epsilon,n}$, $E=Q(\!\! \mod_{\!\!Q}(\widetilde{U}_{\ell,\epsilon})+\!\! \mod_{\!\!Q}(\widetilde{V}_{\ell,\epsilon}))$. Note that $E\in \{\frac{-1}{N},0,\frac{1}{N}\}$ by construction, where $N\triangleq 2^n$. We will show that i) $H(C)- H(D)\to 0$ as $N\to \infty$, and ii)   $H(\widehat{U}_{\ell,\epsilon,n},C)- H(\widehat{U}_{\ell,\epsilon,n},D)\to 0$ as $N\to \infty$.
  This implies that $I(\widehat{U}_{\ell,\epsilon,n};\widehat{U}_{\ell,\epsilon,n}+\widehat{V}_{\ell,\epsilon,n})-I(\widehat{U}_{\ell,\epsilon,n};Q(\widetilde{U}_{\ell,\epsilon}+\widetilde{V}_{\ell,\epsilon}))\to 0$ as $n \to \infty$. Consequently, by data processing inequality and lower semi-continuity of mutual information, we have
 $I(\widehat{U}_{\ell,\epsilon,n};\widehat{U}_{\ell,\epsilon,n}+\widehat{V}_{\ell,\epsilon,n})-I(\widetilde{U}_{\ell,\epsilon};\widetilde{U}_{\ell,\epsilon}+\widetilde{V}_{\ell,\epsilon})\to 0$, and similarly, $I(\widehat{V}_{\ell,\epsilon,n};\widehat{U}_{\ell,\epsilon,n}+\widehat{V}_{\ell,\epsilon,n})-I(\widetilde{V}_{\ell,\epsilon};\widetilde{U}_{\ell,\epsilon}+\widetilde{V}_{\ell,\epsilon})\to 0$.
 
 First, we will show that  $H(C)- H(D)\to 0$ as $N\to \infty$. Note that $C=D+E$ using the distributive property of lattices \cite{zamir2014lattice}. As a result,
 \begin{align*}
     H(C)-H(D)= I(C;E)-I(D;E).
 \end{align*}
 First, let us consider $I(D;E)$ as follows: 
\begin{align*}
    I(D;E)& = I(\widehat{U}_{\ell,\epsilon,n}+\widehat{V}_{\ell,\epsilon,n};  Q(\!\!\!\! \mod_{\!\!Q}(\widetilde{U}_{\ell,\epsilon})+\!\!\!\!
    \mod_{\!\!Q}(\widetilde{V}_{\ell,\epsilon})))
    \\& \leq 
    I(\widehat{U}_{\ell,\epsilon,n}+\widehat{V}_{\ell,\epsilon,n};  \!\! \mod_{\!\!Q}(\widetilde{U}_{\ell,\epsilon})+\!\!
    \mod_{\!\!Q}(\widetilde{V}_{\ell,\epsilon}))
    \\& \leq     I(\widehat{U}_{\ell,\epsilon,n},\widehat{V}_{\ell,\epsilon,n};  \!\! \mod_{\!\!Q}(\widetilde{U}_{\ell,\epsilon}),\!\!
    \mod_{\!\!Q}(\widetilde{V}_{\ell,\epsilon})) ,
\end{align*}
which goes to $0$ as $N\to \infty$ using Lemma 6 in \cite{makkuva2018equivalence}. Next, let us consider $I(C;E)$. Using Proposition 12 in \cite{polyanskiy2015dissipation}, we have:
\begin{align*}
I(C;E)\leq (\log{3}-1) T(E;C)+h_b(T(E;C)),
\end{align*}
where $h_b$ is the binary entropy function, and $T(E;C)\triangleq V(P_{E,C},P_{E}P_{C})$, where $P_{E,C}$ is the probability mass function of the pair $(E,C)$. So, it suffices to show that $T(E;C)\to 0$ as $N\to \infty$.
Note that 
\begin{align}
\nn  &  V(P_{E,C},P_{E}P_{C}) =
    \sum_{e\in\{\frac{-1}{N},0,\frac{1}{N}\}}
P(E=e) V(P_{C},P_{C|E}(\cdot|e))
\leq \sum_{e\in\{\frac{-1}{N},0,\frac{1}{N}\}}V(P_{C},P_{C|E}(\cdot |e))
\\&\nn\leq  \sum_{e,e'\in\{\frac{-1}{N},0,\frac{1}{N}\}}
V(P_{C|E}(\cdot|e'),P_{C|E}(\cdot|e))
\\&\nn
=\sum_{e,e'\in\{\frac{-1}{N},0,\frac{1}{N}\}}
\sum_{d}
|P_{D|E}(d-e'|e')-P_{D|E}(d-e|e)|
\\&\nn\leq
\sum_{e,e'\in\{\frac{-1}{N},0,\frac{1}{N}\}}
\sum_{d}
|P_{D|E}(d-e'|e')-P_{D|E}(d-e'|e)|
+
|P_{D|E}(d-e'|e)-P_{D|E}(d-e|e)|
\\\label{eq:th7:sum}&
=
\sum_{e,e'\in\{\frac{-1}{N},0,\frac{1}{N}\}}
\sum_{d}
|P_{D|E}(d|e')-P_{D|E}(d|e)|
+
|P_{D|E}(d-e'|e)-P_{D|E}(d-e|e)|
\end{align}
We investigate the first term in the summation in 
Equation \eqref{eq:th7:sum}.
\begin{align*}
 &   \sum_{e,e'\in\{\frac{-1}{N},0,\frac{1}{N}\}}
\sum_{d}
|P_{D|E}(d|e')-P_{D|E}(d|e)|
\\&\leq 
\sum_{e,e'\in\{\frac{-1}{N},0,\frac{1}{N}\}}
\sum_{d}
|P_{D|E}(d|e')-P_{D}(d)|
+|P_{D}(d)-P_{D|E}(d|e)|,
\end{align*} 
which goes to 0 as $N\to \infty$ due to Pinsker's inequality. To see this, fix $\eta>0$, and let $N$ be large enough so that $I(D;E)<\eta$. Note that such $N$ exists since $\lim_{N\to\infty}I(D;E)= 0$ as shown above. Due to Pinsker's inequality, we have:
\begin{align*}
    &\eta\geq I(D;E)= \sum_{e}P(E=e) D(P_{D|E}(\cdot|e)||P_{D} )
    \\&\geq 2(\ln{2})\sum_{e} P(E=e)V^2(P_{D|E}(\cdot|e), P_{D})
    \\& \geq {2(\ln{2})}P(E=e')V^2(P_{D|E}(\cdot|e'), P_{D}), \qquad \forall e'\in \{\frac{-1}{N},0,\frac{1}{N}\}. 
\end{align*}
Furthermore, we show that $|P(E=0)-\frac{3}{4}|\to 0$ and $|P(E=\frac{1}{N})-\frac{1}{8}|\to 0$, and $|P(E=\frac{-1}{N})-\frac{1}{8}|\to 0$ as $N \to \infty$:
\begin{align*}
  &  P(E=\frac{1}{N})=
    \int_{u,v: Q(\!\!\!\!\!\!\mod_{\!\!Q} {(u)}+\!\!\!\!\!\!\mod _{\!\!Q}{(v)}))=\frac{1}{N}}f_{\widetilde{U}_\ell,\widetilde{V}_\ell}(u,v)dudv
   \\& =\sum_{i=1}^N\sum_{j=1}^N
    \int_{u,v\in E_i\times E_j: Q(\!\!\!\!\!\!\mod_{\!\!Q}{(u)}+\!\!\!\!\!\!\mod_{\!\!Q}(v)))=\frac{1}{N}}f_{\widetilde{U}_{\ell,\epsilon},\widetilde{V}_{\ell,\epsilon}}(u,v)dudv
       \\& =\sum_{i=1}^N\sum_{j=1}^N
    \int_{u,v\in E_i\times E_j: u+v>e_i+e_j+\frac{1}{2N}}f_{\widetilde{U}_{\ell,\epsilon},\widetilde{V}_{\ell,\epsilon}}(u,v)dudv
    \\&
\leq     \sum_{i=1}^N\sum_{j=1}^N
    \int_{u,v\in E_i\times E_j: u+v<e_i+e_j-\frac{1}{2N}}f_{\widetilde{U}_{\ell,\epsilon},\widetilde{V}_{\ell,\epsilon}}(u,v)dudv+\delta_N
    =P(E=-\frac{1}{N})+\delta_N,
\end{align*}
where in the last inequality we have used the fact that $f_{\widetilde{U}_{\ell,\epsilon},\widetilde{V}_{\ell,\epsilon}}$ is continuous over a compact support, and hence uniformly continuous, to argue the existence of $\delta_N$ such that $\delta_N\to 0$ as $N\to 0$. 
Similarly $P(E=\frac{-1}{N})\leq P(E=\frac{1}{N})+\delta_N$.
Furthermore,
\begin{align*}
   & P(E=0)=
    \int_{u,v: Q(\!\!\!\!\!\!\mod_Q(u)+\!\!\!\!\!\!\mod_Q(v)))=0}f_{\widetilde{U}_{\ell,\epsilon},\widetilde{V}_{\ell,\epsilon}}(u,v)dudv
   \\& =\sum_{i=1}^N\sum_{j=1}^N
    \int_{u,v\in E_i\times E_j: Q(\!\!\!\!\!\!\mod_Q(u)+\!\!\!\!\!\!\mod_Q(v)))=0}f_{\widetilde{U}_{\ell,\epsilon},\widetilde{V}_{\ell,\epsilon}}(u,v)dudv
   \\& =\sum_{i=1}^N\sum_{j=1}^N
    \int_{u,v\in E_i\times E_j: e_i+e_j+\frac{-1}{2N}<u+v<e_i+e_j+\frac{1}{2N}}f_{\widetilde{U}_{\ell,\epsilon},\widetilde{V}_{\ell,\epsilon}}(u,v)dudv
    \\& 
    \leq    6 \sum_{i=1}^N\sum_{j=1}^N
    \int_{u,v\in E_i\times E_j: u+v<e_i+e_j-\frac{1}{2N}}f_{\widetilde{U}_{\ell,\epsilon},\widetilde{V}_{\ell,\epsilon}}(u,v)dudv+\delta_N
    = 6P(E=\frac{-1}{N})+\delta_N.
\end{align*}
Similarly, $6P(E=\frac{-1}{N})\leq P(E=0)+\delta_N$. So, $|P(E=0)-\frac{3}{4}|\leq \delta_N$ and $|P(E=\frac{1}{N})-\frac{1}{8}|\leq \delta_N$, and $|P(E=\frac{-1}{N})-\frac{1}{8}|\leq \delta_N$. Consequently, $V(P_{D|E}(\cdot|e'), P_{D})\to 0$ as $N \to \infty$ for all $e'\in \{\frac{-1}{N},0,\frac{1}{N}\}$. 

Next, we consider the second term in Equation   \eqref{eq:th7:sum}. Fix $ e,e' \in \{\frac{-1}{N},0 ,\frac{1}{N}\}$. Note that
\begin{align*}
    D+e-e' = Q(\widetilde{U}_{\ell,\epsilon})+Q(\widetilde{V}_{\ell,\epsilon})+e-e' = Q(\widetilde{U}_{\ell,\epsilon}+e)+Q(\widetilde{V}_{\ell,\epsilon}-e'),
\end{align*}
So, 
\begin{align}
 &\nn  \sum_{d}
|P_{D|E}(d-e'|e)-P_{D|E}(d-e|e)|
\\& \nn
=\sum_{d}
|P_{D|E}(d+e-e'|e)-P_{D|E}(d|e)|
\\&\nn= \sum_{d}
|P(Q(\widetilde{U}_{\ell,\epsilon}+e)+ Q(\widetilde{V}_{\ell,\epsilon}-e')=d|E=e)-P(Q(\widetilde{U}_{\ell,\epsilon})+Q(\widetilde{V}_{\ell,\epsilon})=d)|E=e)|
\\&\nn\stackrel{(a)}{\leq} 
2\sup_{\mathsf{A}} |P_{\widetilde{U}_{\ell,\epsilon},\widetilde{V}_{\ell,\epsilon}|E}(\mathsf{A}+\{(e,-e')\}|e), -P_{\widetilde{U}_{\ell,\epsilon},\widetilde{V}_{\ell,\epsilon}|E}(\mathsf{A}|e)|
\\& \nn=
2\sup_{\mathsf{A}} \frac{1}{P(E=e)}|P_{\widetilde{U}_{\ell,\epsilon},\widetilde{V}_{\ell,\epsilon}}((\mathsf{A}+ \{(e,-e')\})\cap \{E=e\}) -P_{\widetilde{U}_{\ell,\epsilon},\widetilde{V}_{\ell,\epsilon}}(\mathsf{A}\cap \{E=e\})|
\\& \nn\stackrel{(b)}{\leq} 2\sup_{\mathsf{A}} \frac{1}{P(E=e)}|P_{\widetilde{U}_{\ell,\epsilon},\widetilde{V}_{\ell,\epsilon}}(\mathsf{A}+ \{(e,-e')\})-P_{\widetilde{U}_{\ell,\epsilon},\widetilde{V}_{\ell,\epsilon}}(\mathsf{A})|
\\&=\label{eq:th7:sup} \frac{1}{P(E=e)}\int_{u,v}|f_{\widetilde{U}_{\ell,\epsilon},\widetilde{V}_{\ell,\epsilon}}(u+e,v-e')-f_{\widetilde{U}_{\ell,\epsilon},\widetilde{V}_{\ell,\epsilon}}(u,v)|dudv,
\end{align}
where in (a) we have used the data processing inequality for variational distance, and in (b) we have used the fact that the supremum over $\mathsf{A}$ is larger than that over $\mathsf{A}\cap \{E=e\}$. Note that the last term goes to $0$ as $N\to \infty$ due to uniform continuity of $f_{\widetilde{U}_{\ell,\epsilon},\widetilde{V}_{\ell,\epsilon}}$. 

We have thus shown that $T(C;E)\to 0$ as $N\to \infty$, and hence
$I(C;E) \to 0$, and consequently, 
$H(C)\to H(D)$ as $N\to \infty$. Next, we will show that $H(\widehat{U}_{\ell,\epsilon,n},C)- H(\widehat{U}_{\ell,\epsilon,n},D)\to 0$. Similar to the previous part, we have: 
 \begin{align*}
     H(\widehat{U}_{\ell,\epsilon,n},C)-H(\widehat{U}_{\ell,\epsilon,n},D)= I(\widehat{U}_{\ell,\epsilon,n},C;E)-I(\widehat{U}_{\ell,\epsilon,n},D;E).
 \end{align*}
The second term $I(\widehat{U}_{\ell,\epsilon,n},D;E)$ goes to $0$ as $N\to \infty$ by a similar argument as in the previous case. For the first term, similarly it suffices to show that $V(P_{E,C, \widehat{U}_{\ell,\epsilon,n}}, P_EP_{C, \widehat{U}_{\ell,\epsilon,n}}) \to 0$ as $N\to \infty$. We have: 
 
\begin{align}
\nn  &  V(P_{E,C, \widehat{U}_{\ell,\epsilon,n}}, P_EP_{C, \widehat{U}_{\ell,\epsilon,n}}) =
    \sum_{e\in\{\frac{-1}{N},0,\frac{1}{N}\}}
P(E=e) V(P_{C,\widehat{U}_{\ell,\epsilon,n}},P_{C,\widehat{U}_{\ell,\epsilon,n}|E}(\cdot|e))
\\&
\nn\leq  \sum_{e,e'\in\{\frac{-1}{N},0,\frac{1}{N}\}}
V(P_{C,\widehat{U}_{\ell,\epsilon,n}|E}(\cdot|e'),P_{C,\widehat{U}_{\ell,\epsilon,n}|E}(\cdot|e))
\\&\nn
=\sum_{e,e'\in\{\frac{-1}{N},0,\frac{1}{N}\}}
\sum_{d}\sum_{u}
|P_{D,\widehat{U}_{\ell,\epsilon,n}|E}(d-e',u|e')-P_{D,\widehat{U}_{\ell,\epsilon,n}|E}(d-e,u|e)|
\\&\nn\leq
\sum_{e,e'\in\{\frac{-1}{N},0,\frac{1}{N}\}}
\sum_{d} \sum_{u}
|P_{\widehat{U}_{\ell,\epsilon,n},\widehat{V}_{\ell,\epsilon,n}|E}(u,d-e'-u|e')-P_{\widehat{U}_{\ell,\epsilon,n},\widehat{V}_{\ell,\epsilon,n}|E}(u,d-e'-u|e)|
\\&\nn
\hspace{0.5in} +
|P_{\widehat{U}_{\ell,\epsilon,n},\widehat{V}_{\ell,\epsilon,n},|E}(u,d-e'-u|e)-P_{\widehat{U}_{\ell,\epsilon,n},\widehat{V}_{\ell,\epsilon,n}|E}(u,d-e-u|e)|
\\&
=\nn
\sum_{e,e'\in\{\frac{-1}{N},0,\frac{1}{N}\}}
\sum_{v}\sum_{u}
|P_{\widehat{U}_{\ell,\epsilon,n},\widehat{V}_{\ell,\epsilon,n}|E}(u,v|e')-P_{\widehat{U}_{\ell,\epsilon,n},\widehat{V}_{\ell,\epsilon,n}|E}(u,v|e)|
\\& \label{eq:th7:sum2} \hspace{0.5in} +
|P_{\widehat{U}_{\ell,\epsilon,n},\widehat{V}_{\ell,\epsilon,n}|E}(u,v-e'|e)-P_{\widehat{U}_{\ell,\epsilon,n},\widehat{V}_{\ell,\epsilon,n}|E}(u,v-e|e)|.
\end{align} 
 We will focus on the first term in equation \eqref{eq:th7:sum2}:
 \begin{align*}
  &  \sum_{e,e'\in\{\frac{-1}{N},0,\frac{1}{N}\}}
\sum_{v}\sum_{u}
|P_{\widehat{U}_{\ell,\epsilon,n},\widehat{V}_{\ell,\epsilon,n}|E}(u,v|e')-P_{\widehat{U}_{\ell,\epsilon,n},\widehat{V}_{\ell,\epsilon,n}|E}(u,v|e)|
\\& \leq 
\sum_{e,e'\in\{\frac{-1}{N},0,\frac{1}{N}\}} 
|P_{\widehat{U}_{\ell,\epsilon,n},\widehat{V}_{\ell,\epsilon,n}|E}(u,v|e')-P_{\widehat{U}_{\ell,\epsilon,n},\widehat{V}_{\ell,\epsilon,n}}(u,v)|
\\& \hspace{0.5in} +|P_{\widehat{U}_{\ell,\epsilon,n},\widehat{V}_{\ell,\epsilon,n}}(u,v)-P_{\widehat{U}_{\ell,\epsilon,n},\widehat{V}_{\ell,\epsilon,n}|E}(u,v|e)|,
 \end{align*}
 where the last two terms go to 0 as $N\to \infty$ due to Pinsker's inequality and the fact that $I(\widehat{U}_{\ell,\epsilon,n},D;E)\to 0$ as $N\to \infty$. Next we note the second term in \eqref{eq:th7:sum2} goes to 0 by a similar argument as in Equation \eqref{eq:th7:sup} and uniform continuity of $f_{\widetilde{U}_{\ell,\epsilon},\widetilde{V}_{\ell,\epsilon}}$. As a result, 
$H(\widehat{U}_{\ell,\epsilon,n},C)- H(\widehat{U}_{\ell,\epsilon,n},D)\to 0$. 

Lastly, Equations \eqref{eq:5.1.6} and \eqref{eq:5.1.7} follow by convergence in distribution of $(X,\widehat{U}_{n,\ell,\epsilon},  \widehat{V}_{n,\ell',\epsilon})$ to $(X,U,V)$ and $(Y,\widehat{U}_{n,\ell,\epsilon},  \widehat{V}_{n,\ell',\epsilon})$ to $(Y,U,V)$ along with the Portmanteau theorem \cite{achim2014probability,billingsley2013convergence} and Lemma \ref{thm:quantize_dist_cost}.
This completes the proof. 
\end{proof}
\section{Proof of Theorem \ref{th:9}}
\label{App:th:9}

\noindent \textbf{Step 1 (Clipping $X$ and $Y$, and Generating $\overline{U}_{\ell}$ and $\overline{V}_{\ell}$) :} Let $Z$, $W$, $\widetilde{X}_{\ell}$, and $\widetilde{Y}_{\ell'}$ be as defined in Section \ref{sec:th:9}. 
Let $(\overline{U}_{\ell},\overline{V}_{\ell'})$ be random variables that are correlated with
$(X,Y,Z,W)$ such that the distribution of $\overline{U}_{\ell}$ given $\widetilde{X}_{\ell}$ 
is given by 
\[ P_{\overline{U}_{\ell}|\widetilde{X}_\ell}(\cdot|x) = P_{U|X}(\cdot|x), \quad  x\in [-\ell,\ell]
\]
and  the distribution of $\overline{V}_{\ell'}$ given $\widetilde{Y}_{\ell'}$ 
is given by 
\[
P_{\overline{V}_{\ell'}|\widetilde{Y}_{\ell'}}(\cdot|y) = P_{V|Y}(\cdot|y),\quad y\in [-\ell',\ell'],
\]
One can check that the following Markov chain holds:
$\overline{U}_{\ell}- \widetilde{X}_{\ell} -  (X,Y,Z,W) - \widetilde{Y}_{\ell'} - \overline{V}_{\ell'}$. 
Furthermore,  for any quadruple of events $\mathsf{A},\mathsf{B},\mathsf{C}$, and $\mathsf{D}$, the random variables $(\widetilde{X}_{\ell},\widetilde{Y}_{\ell'},\overline{U}_{\ell},\overline{V}_{\ell'})$
have the following distribution: 
\begin{align*}
& P_{\widetilde{X}_{\ell},\widetilde{Y}_{\ell'},\overline{U}_{\ell},\overline{V}_{\ell'}}(\mathsf{A},\mathsf{B},\mathsf{C},\mathsf{D})= P_{\widetilde{X}_{\ell},\widetilde{Y}_{\ell'},\overline{U}_{\ell},\overline{V}_{\ell'},X,Y}(\mathsf{A},\mathsf{B},\mathsf{C},\mathsf{D},\mathbb{R},\mathbb{R})
\\&=  P_{\widetilde{X}_{\ell},\widetilde{Y}_{\ell'},\overline{U}_{\ell},\overline{V}_{\ell'},X,Y}(\mathsf{A},\mathsf{B},\mathsf{C},\mathsf{D},\mathbb{R}\cap [-\ell,\ell], \mathbb{R}\cap [-\ell',\ell'])
\\&+ P_{\widetilde{X}_{\ell},\widetilde{Y}_{\ell'},\overline{U}_{\ell},\overline{V}_{\ell'},X,Y}(\mathsf{A},\mathsf{B},\mathsf{C},\mathsf{D}, \mathbb{R}\cap [-\ell,\ell]^c, \mathbb{R}\cap [-\ell',\ell'])
\\&+
P_{\widetilde{X}_{\ell},\widetilde{Y}_{\ell'},\overline{U}_{\ell},\overline{V}_{\ell'},X,Y}(\mathsf{A},\mathsf{B},\mathsf{C},\mathsf{D}, \mathbb{R}\cap [-\ell,\ell], \mathbb{R}\cap [-\ell',\ell']^c)
\\&+ P_{\widetilde{X}_{\ell},\widetilde{Y}_{\ell'},\overline{U}_{\ell},\overline{V}_{\ell'},X,Y}(\mathsf{A},\mathsf{B},\mathsf{C},\mathsf{D}, \mathbb{R}\cap [-\ell,\ell]^c, \mathbb{R}\cap [-\ell',\ell']^c)
\\&
=P_{{X},{Y},U,V}(\mathsf{A} \cap [-\ell,\ell],\mathsf{B} \cap [-\ell',\ell'],\mathsf{C},\mathsf{D})
\\&+ P_{\widetilde{X}_{\ell},\widetilde{Y}_{\ell'},\overline{U}_{\ell},\overline{V}_{\ell'},X,Y}(\mathsf{A},\mathsf{B},\mathsf{C},\mathsf{D}, \mathbb{R}\cap [-\ell,\ell]^c, \mathbb{R}\cap [-\ell',\ell'])
\\&+
P_{\widetilde{X}_{\ell},\widetilde{Y}_{\ell'},\overline{U}_{\ell},\overline{V}_{\ell'},X,Y}(\mathsf{A},\mathsf{B},\mathsf{C},\mathsf{D}, \mathbb{R}\cap [-\ell,\ell], \mathbb{R}\cap [-\ell',\ell']^c)
\\&+ P_{\widetilde{X}_{\ell},\widetilde{Y}_{\ell'},\overline{U}_{\ell},\overline{V}_{\ell'},X,Y}(\mathsf{A},\mathsf{B},\mathsf{C},\mathsf{D}, \mathbb{R}\cap [-\ell,\ell]^c, \mathbb{R}\cap [-\ell',\ell']^c)
\end{align*}
Note that 
\begin{align*}
 0 \leq &   P_{\widetilde{X}_{\ell},\widetilde{Y}_{\ell'},\overline{U}_{\ell},\overline{V}_{\ell'},X,Y}(\mathsf{A},\mathsf{B},\mathsf{C},\mathsf{D}, \mathbb{R}\cap [-\ell,\ell]^c, \mathbb{R}\cap [-\ell',\ell'])
\\&+
P_{\widetilde{X}_{\ell},\widetilde{Y}_{\ell'},\overline{U}_{\ell},\overline{V}_{\ell'},X,Y}(\mathsf{A},\mathsf{B},\mathsf{C},\mathsf{D}, \mathbb{R}\cap [-\ell,\ell], \mathbb{R}\cap [-\ell',\ell']^c)
\\&+ P_{\widetilde{X}_{\ell},\widetilde{Y}_{\ell'},\overline{U}_{\ell},\overline{V}_{\ell'},X,Y}(\mathsf{A},\mathsf{B},\mathsf{C},\mathsf{D}, \mathbb{R}\cap [-\ell,\ell]^c, \mathbb{R}\cap [-\ell',\ell']^c)
\\& \leq 
1- P_{X,Y}( \mathbb{R}\cap [-\ell,\ell], \mathbb{R}\cap [-\ell',\ell']),
\end{align*}
which approaches $0$ as $\ell,\ell' \to \infty$.
As a result, 
\begin{align}
\label{eq:port}
&\lim_{\ell,\ell'\to \infty} P_{\widetilde{X}_{\ell},\widetilde{Y}_{\ell'},\overline{U}_{\ell},\overline{V}_{\ell'}}(\mathsf{A},\mathsf{B},\mathsf{C},\mathsf{D})= P_{X,Y,U,V}(\mathsf{A},\mathsf{B},\mathsf{C},\mathsf{D})
\end{align}
for all $\mathsf{A},\mathsf{B},\mathsf{C},\mathsf{D}$.
Hence as $\ell,\ell'\to \infty$,  $P_{\widetilde{X}_{\ell},\widetilde{Y}_{\ell'},
 \overline{U}_\ell,\overline{V}_{\ell'}}$  converges strongly to $P_{X,Y,U,V}$. Note that using Equation \eqref{eq:port} along with the Portmanteau Theorem we have:
\begin{align}
&\lim_{\ell  \rightarrow \infty}
|\mathbb{E}d_1(\widetilde{X}_{\ell},g_1(\overline{U}_{\ell},\overline{V}_{\ell'}))
-\mathbb{E}\tilde{d}_1(X,g_1(U,V))|=0,
\\&
|\mathbb{E}d_2(\widetilde{Y}_{\ell},g_2(\overline{U}_{\ell},\overline{V}_{\ell'}))
-\mathbb{E}\tilde{d}_1(Y,g_2(U,V))|=0.
\end{align}
Next using the arguments similar to that used in source coding with
side information, we  show that $I(\widetilde{X}_{\ell};\overline{U}_{\ell}) \to I(X;U)$
 as $\ell\to \infty$. By the data processing inequality:  
 \begin{align*}
     I(\widetilde{X}_{\ell};\overline{U}_{\ell}) =    I(\widetilde{X}_{\ell};{U}) \leq 
     I({X};{U}),
 \end{align*}
  and convergence follows by lower semi-continuity of mutual information.
Similarly, $I(\widetilde{Y}_{\ell'};\overline{U}_{\ell}) \to I(Y;U)$, and for sufficiently large $\ell,\ell'$. 
Also,  note that $U,V,\overline{U}_{\ell}$ and $\overline{V}_{\ell'}$ are discrete variables defined on fixed finite alphabets. As a result,  \[
\lim_{\ell,\ell', \rightarrow \infty} I(\overline{U}_{\ell};\overline{V}_{\ell'}) = I(U;V).
\]
Morever, 
\begin{align*}
&\lim_{\ell,\ell' \rightarrow \infty} I(\overline{U}_{\ell};\overline{U}_{\ell}+ \overline{V}_{\ell'}) = I(U;U+V),
\\& \lim_{\ell,\ell' \rightarrow \infty} I(\overline{V}_{\ell'};\overline{U}_{\ell}+ \overline{V}_{\ell'}) = I(V;U+V).
\end{align*}
\noindent \textbf{Step 2 (Discretizing $X$ and $Y$):} 
Next we quantize $\widetilde{X}_{\ell}$ and $\widetilde{Y}_{\ell'}$ into
$\widehat{X}_{n,\ell}$ and $\widehat{Y}_{n,\ell'}$ and enforce the Markov chain. 
Now using 
\[
I(\widetilde{X}_{\ell}\widetilde{Y}_{\ell'}
\overline{V}_{\ell};\overline{U}_{\ell}|\widehat{X}_{n,\ell})=I(\widetilde{X}_{\ell},\widetilde{Y}_{\ell'}
\overline{V}_{\ell}; \overline{U}_{\ell})
-I(\widehat{X}_{n,\ell};\overline{U}_{\ell}), 
\]
and  Theorem \ref{thm:quantize0} we have
\beq
\label{eq:step3a_inter}
   \lim_{n \rightarrow \infty}  I(\widetilde{X}_{\ell}\widetilde{Y}_{\ell'}\overline{V}_{\ell'};\overline{U}_{\ell}|\widehat{X}_{n,\ell})
=I(\widetilde{Y}_{\ell'} \overline{V}_{\ell'};\overline{U}_{\ell}|\widetilde{X}_{\ell})=0,
\eeq
and similarly, 
\beq
\label{eq:step3a_inter1}
   \lim_{n \rightarrow \infty}  I(\widetilde{X}_{\ell}\widetilde{Y}_{\ell'}\overline{U}_{\ell};\overline{V}_{\ell'}|\widehat{Y}_{n,\ell'})
=I(\widetilde{X}_{\ell}, \overline{U}_{\ell};\overline{V}_{\ell'}|\widetilde{Y}_{\ell'})=0.
\eeq

Define $\overline{U}_{n,\ell}$ and $\overline{V}_{n,\ell'}$ as random variables having the same alphabet as
$\overline{U}_{\ell}$ and $\overline{V}_{\ell'}$, and  that are jointly correlated with
$(\widehat{X}_{n,\ell},\widetilde{X}_{\ell},\widetilde{Y}_{\ell'},\widehat{Y}_{n,\ell'})$ according to the probability
distribution that satisfies (i) the Markov chain 
$\overline{V}_{n,\ell'}  - \widehat{Y}_{n,\ell'} - \widetilde{Y}_{\ell} - \widetilde{X}_{\ell} - \widehat{X}_{n,\ell} - \overline{U}_{n,\ell}$,  
(ii) the pair $(\widehat{X}_{n,\ell},\overline{U}_{n,\ell})$ has the same distribution as 
the pair $(\widehat{X}_{n,{\ell}},\overline{U}_{\ell})$, 
and (iii) the pair $(\widehat{Y}_{n,
\ell'},\overline{V}_{n,\ell'})$ has the same distribution as 
the pair $(\widehat{Y}_{n,\ell'},\overline{V}_{\ell'})$. We use Lemma \ref{lem:mc_forced1} as follows.
From Equations \eqref{eq:step3a_inter} and \eqref{eq:step3a_inter1}, by taking the quintuple
$A=\overline{U}_{\ell}$, $B=\widehat{X}_{n,\ell}$, $C=(\widetilde{X}_{\ell},\widetilde{Y}_{\ell})$, $D=\widehat{Y}_{n,\ell'}$, and $E=\overline{V}_{\ell'}$, we have 
\[
   \lim_{n \rightarrow \infty}   V \left( P_{\widehat{X}_{n,\ell}\widehat{Y}_{n,\ell'}\overline{U}_{n,\ell}\overline{V}_{n,\ell'}}
  ,  P_{\widetilde{X}_{\ell}\widetilde{Y}_{\ell'}\overline{U}_{\ell},\overline{V}_{\ell'}} \right)=0,
\]
and hence using Theorem \ref{thm:quantize0}, we have 
\begin{align*}
&\lim_{n  \rightarrow \infty}
I(\widehat{X}_{n,\ell};\overline{U}_{n,\ell})=\lim_{n \to \infty}I(\widehat{X}_{n,\ell};\overline{U}_{\ell})= I(\widetilde{X}_{\ell};\overline{U}_{\ell}),
\\&\lim_{n_Y \rightarrow \infty}
I(\widehat{Y}_{n,\ell'};\overline{V}_{n,\ell'})=
\lim_{n \rightarrow \infty}I(\widehat{Y}_{n,\ell'};\overline{V}_{\ell'})
=I(\widetilde{Y}_{\ell'};\overline{V}_{\ell'}),
\end{align*}
and using the continuity of mutual information for finite alphabets, we have 
\[ 
\lim_{n  \rightarrow \infty}
I(\overline{U}_{n,\ell};\overline{V}_{n,\ell'})=
I(\overline{U}_{\ell};\overline{V}_{\ell'}),
\]
\[ 
\lim_{n \rightarrow \infty}
I(\overline{U}_{n,\ell};\overline{U}_{n,\ell}+\overline{V}_{n,\ell'})=
I(\overline{U}_{\ell};\overline{U}_{\ell}+\overline{V}_{\ell'}),
\qquad 
\lim_{n  \rightarrow \infty}
I(\overline{V}_{n,\ell'};\overline{U}_{n,\ell}+\overline{V}_{n,\ell'})=
I(\overline{V}_{\ell'};\overline{U}_{\ell}+\overline{V}_{\ell'}).
\]

Since convergence in variational distance implies convergence in
distribution, and using  Portmanteau theorem \cite{achim2014probability,billingsley2013convergence} and Lemma \ref{thm:quantize_dist_cost} Equations \eqref{eq:9.1.6} and \eqref{eq:9.1.7} follow. This completes the proof.

\section{Proof of Theorem \ref{thm:DBC}}
\label{App:thm:DBC}
Consider the probability measure $P_{XU}P_{Y|X}P_{Z|Y}$ as given in the theorem statement. 
\\\textbf{Step 1 (Clipping X and U).}  
Let $\ell_i,i\in \mathbb{N}$ be a sequence of positive numbers approaching infinity such that $\mathbb{E}(\kappa(C_{\ell_i}(X)))\to  \mathbb{E}(\kappa(X))$ as $i\to \infty$. Note that such sequence always exists as shown in Lemma \ref{thm:quantize_dist_cost}. 
Fix $i$, let $\ell=\ell_i$, and define 
\begin{align*}
    (\widetilde{X},\widetilde{U})= 
    \begin{cases}
(X,U)\qquad & \text{ if }  U,X \in [-\ell,\ell]\\
(C_{\ell}(X),\ell)& \text{otherwise.}
\end{cases}
\end{align*}
Assume that we use $\widetilde{X}$ as the input to the channel $P_{Y,Z|X}$, and let $Y'$ and $Z'$ denote the output of the corresponding channel. We note that that for any $\epsilon> 0$ we have $I(\widetilde{U};Z')\geq  I(U;Z)-\epsilon$, $I(\widetilde{U};Y')\geq I(U;Y)-\epsilon$ and $I(\widetilde{X},\widetilde{U};Y')\geq I(X,U;Y)-\epsilon$  for large enough $\ell$ due to convergence in distribution and lower-semi continuity of mutual information. Furthermore, 
\begin{align*}
    &I(\widetilde{U};Y')\leq I(\widetilde{U}, \mathbbm{1}(U,X\in [-\ell,\ell]);Y')
    \\& \leq H(P(\mathbbm{1}(U,X\in [-\ell,\ell]))
    + P(U,X\in [-\ell,\ell]) I(U;Y|U,X\in [-\ell,\ell])
   \\& + P(U\notin [-\ell,\ell] \text{ or } X\notin [-\ell,\ell]) I(\widetilde{U};Y'|U\notin [-\ell,\ell]\text{ or } X\notin [-\ell,\ell])
   \\& \stackrel{(a)}{\leq}  P(U,X\in [-\ell,\ell]) I(U;Y|U,X\in [-\ell,\ell])+\epsilon
   \\& \leq I(U, \mathbbm{1}(U,X\in [-\ell,\ell]);Y)+\epsilon,
   \\& \leq I(U;Y)+ H(\mathbbm{1}(U,X\in [-\ell,\ell])|U)+\epsilon
   \\& \stackrel{(b)}{\leq} I(U;Y)+2\epsilon,
\end{align*}
 where (a) follows by the fact that 
 \begin{align*}
& P(U\notin [-\ell,\ell] \text{ or } X\notin [-\ell,\ell]) I(\widetilde{U};Y|U\notin [-\ell,\ell]\text{ or } X\notin [-\ell,\ell])=0
  \end{align*}
  by construction, and (b) follows by the fact that $ H(\mathbbm{1}(U,X\in [-\ell,\ell])|U)\leq P(U\notin [-\ell,\ell]) \leq \epsilon$ for large enough $\ell$. Consequently, $I(\widetilde{U};Y')\to I(U;Y)$ as $\ell\to \infty$. So, there exists $\ell$ large enough such that:
  \begin{align*}
&      I(\widetilde{U};Z)\geq I(U;Z)-\epsilon,\\
& I(\widetilde{X};Y|\widetilde{U})\geq I(X;Y|U)-\epsilon. 
  \end{align*}
  
\noindent\textbf{Step 2 (Discretizing $U,X$).} In this step, we discretize $\widetilde{U}$ and $\widetilde{X}$ with discretization step $2^{-n}$. Let $\widehat{X}$, $\widehat{U}$ denote the corresponding discretized variables. Then, by lower semi-continuity of mutual information and the data processing inequality, we have $I(\widehat{U};Z')\to I(\widetilde{U};Z')$,  $I(\widehat{X},\widehat{U};Y')\to I(\widetilde{X},\widetilde{U};Y')$, and $I(\widehat{U};Y)\to I(\widetilde{U};Y')$ as $n \to \infty$, and hence $I(\widehat{X};Y'|\widehat{U})\to I(\widetilde{X};Y'|\widetilde{U})$ as $n\to \infty$. Furthermore $\mathbb{E}(\kappa(\widehat{X}))\to \mathbb{E}(\kappa(\widetilde{X}))$ as $n\to \infty$ due to the Portmanteau Theorem  and since $\widetilde{X}$ is a bounded variable.

\noindent \textbf{Step 3 (Discretizing $Y,Z$).} Next, we clip and discretize $Y'$ and $Z'$ with clipping interval $[-\ell,\ell]$ and discretization step $n$. Let $\widehat{Y}$ and $\widehat{Z}$ be the resulting discrete variables. By the lower semi-continuity of mutual information and data processing inequality, $I(\widehat{U};\widehat{Z})\to I(\widehat{U};Z')$ and $I(\widehat{X};\widehat{Y}|\widehat{U})\to I(\widehat{U};Y|\widehat{U})$ as $\ell\to \infty$ and $n\to 0$. Note that the channel between $\widehat{X}$ and $\widehat{Y},\widehat{Z}$ is:
 \begin{align*}
     P(\widehat{Y}=k_1,\widehat{Z}=k_2|\widehat{X}=j)= \int_{A_{\ell,n}(j)}\frac{P_{\widetilde{X}}(dx)}{Q(j)}\int_{A_{\ell,n}(k_1)\times A_{\ell,n}(k_2)} P_{Y,Z|X}(dy,dz|x),
 \end{align*}
 Note that the induced discrete broadcast channel $P_{\widehat{Y},\widehat{Z}|\widehat{X}}$ may not be degraded. 
 
\noindent \textbf{Step 4 (Transmission Scheme). }
 Given the discrete broadcast channel $P_{\widehat{Z},\widehat{Y}|\widehat{X}}$, which is not necessarily degraded, the following inner-bound to the Marton's rate region \cite{Marton} is achievable:
 \begin{align}
\label{eq:DBC:1}     &R_1\leq I(\widehat{U};\widehat{Z})\\
\label{eq:DBC:2}      &R_2\leq I(\widehat{X},\widehat{U};\widehat{Y})\\
\label{eq:DBC:3}      &R_1+R_2\leq \min\{I(\widehat{U};\widehat{Z})+I(\widehat{X};\widehat{Y}|\widehat{U}), I(\widehat{U},\widehat{X};\widehat{Y})\}
 \end{align}
Let $\epsilon$ be small enough such that $\epsilon\leq \frac{1}{6}(I(U;Y)-I(U;Z))$, and take $n$ and $\ell$ large enough so that
\begin{align}
\label{eq:DBC:4}     &|I(U;Z)-I(\widehat{U};\widehat{Z})|\leq \epsilon\\
 \label{eq:DBC:5}    &|I(X;Y|U)-I(\widehat{X};\widehat{Y}|\widehat{U})|\leq \epsilon \\
      &|I(X,U;Y)-I(\widehat{X},\widehat{U};\widehat{Y})|\leq \epsilon
      \\
      \label{eq:DBC:6} & \mathbb{E}(\kappa(\widehat{X}))\leq \tau+\epsilon.
\end{align}
Note that such $n$ and $\ell$ exists as shown in the prior steps. Then, we have $$\min\{I(\widehat{U};\widehat{Z})+I(\widehat{X};\widehat{Y}|\widehat{U}), I(\widehat{U},\widehat{X};\widehat{Y})\}= I(\widehat{U};\widehat{Z})+I(\widehat{X};\widehat{Y}|\widehat{U})$$ since $I(U;Z)\leq I(U;Y)$ by degradedness of the original channel. So that the bounds given in \eqref{eq:DBC:1}-\eqref{eq:DBC:3} become:
 \begin{align*}
    &R_1\leq I(\widehat{U};\widehat{Z}),\quad
      &R_1+R_2\leq  I(\widehat{U};\widehat{Z})+I(\widehat{X};\widehat{Y}|\widehat{U}).
 \end{align*}
 So, the following are achievable:
 \begin{align}
 \label{eq:DBC:7}     &R_1= I(\widehat{U};\widehat{Z}),\quad 
     &R_2= I(\widehat{X};\widehat{Y}|\widehat{U}).
 \end{align}
 Let TS$_d$ be a transmission system with parameters $(m,\Theta)$ and rates satisfying the above and the cost constraint $ \frac{1}{m}\sum_{i=1}^m \mathbb{E}(\kappa(X_i))\leq \tau+\epsilon$, where $X^m=e_d(M_1,M_2)$, and $M_1,M_2$ are messages chosen randomly, independently, and uniformly from $[\Theta_1]\times [\Theta_2]$. The continuous transmission system TS$_c$ is constructed as follows. Given a message pair $M_1,M_2$, the codeword $\widehat{X}^m=e_d(M_1,M_2)$ is produced. For each letter $\widehat{X}_i, i\in [m]$, a continuous variable is produced using the reverse test-channel $P_{\widetilde{X}|\widehat{X}}=\frac{P_{\widetilde{X}}(\cdot)}{Q(j)}$, where $\widehat{X}=j$ and $Q(j)\triangleq P(\widehat{X}=j)$. The resulting sequence $\widetilde{X}^m$ is transmitted and the pair $Z^m$ and $Y^m$ are received at decoders 1 and 2, respectively. The decoders clip and discretize the received sequences to $\widehat{Z}^m$ and $\widehat{Y}^m$, respectively, and decode the message using the discrete decoding function. The channel seen by TS$_d$ for $\widehat{X}=j$ is given by:
 \begin{align*}
     P(\widehat{Y}=k_1,\widehat{Z}=k_2|\widehat{X}=j)= \int_{A_{\ell,n}(j)}\frac{P_{\widetilde{X}}(dx)}{Q(j)}\int_{A_{\ell,n}(k_1)\times A_{\ell,n}(k_2)} P_{Y,Z|x}(dy,dz|x),
 \end{align*}
 which is the intended channel. 
 
 \noindent \textbf{Step 5 (Checking the Constraints). } Fix $\epsilon>0$. Assume that $\ell,n,m,\Theta_1,\Theta_2$ are such that:
 \begin{itemize}
     \item \eqref{eq:DBC:4}-\eqref{eq:DBC:7} are satisfied. 
     \item $|\kappa(x_1)-\kappa(x_2)|\leq \epsilon$, for all (a) $x_1,x_2\in [-\ell,\ell]$, (b) $|x_1-x_2|\leq \frac{1}{2^{n_q}}$.
     \end{itemize}
 We see that:
 \begin{align*}
     &\frac{1}{m}\log{\Theta_1}\geq I(U;Z)-\epsilon\\
     &\frac{1}{m}\log{\Theta_2}\geq I(X;Y|U)-\epsilon.
 \end{align*}
 Next, we ensure that the cost constraint is satisfied. Consider for any $t\in \{1,2,\cdots,m\}$, 
 \begin{align*}
     \mathbb{E}(\kappa(\widetilde{X}_t))= \sum_{j}P(\widehat{X}=j)\int_{A_{\ell,n_q}(j)}\kappa(x) \frac{P_{\widetilde{X}}(dx)}{Q(j)}
\leq \sum_{j} P(\widehat{X}=j) (\kappa(j)+\epsilon)= \mathbb{E}(\widehat{X})+\epsilon.
 \end{align*}
 Hence,
 \begin{align*}
     \frac{1}{m}\sum_{i=1}^m \mathbb{E}(\kappa(\widetilde{X}))\leq 
      \frac{1}{m}\sum_{i=1}^m \mathbb{E}(\kappa(\widehat{X}))+\epsilon
      \leq 
            \frac{1}{m}\sum_{i=1}^m \mathbb{E}(\kappa({X}))+2\epsilon.
 \end{align*}
 This completes the proof.
 \section{Proof of Theorem \ref{thm:bergertung_cont}}
\label{App:bergertung_cont}
\noindent \textbf{Proof:} 
Fix $\xi>0$, and define $d'_1(x,u,v)\triangleq d_1(x,g(u,v))$ and $d'_2(y,u,v)\triangleq d_2(y,g(u,v))$ for all $x,y,u,v, \in \mathbb{R}^4$.
Using the procedure described Section \ref{sec:th:7}, 
we perform clipping, smoothing, and discretization of random variables $(U,V)$ into $\widehat{U}_{n',\ell',\epsilon'}$ and  
$\widehat{V}_{n',\ell',\epsilon}$ by choosing the parameters $(n',\ell',\epsilon)$ appropriately. Furthermore, using the procedure described in Section \ref{sec:th:9}, we perform clipping, and discretization of the source variables $(X,Y)$ to produce the quadruple $(\widehat{X}_{n,\ell}, \widehat{Y}_{n,\ell}, \overline{U}'_{n',\ell',\epsilon'}, \overline{V}'_{n',\ell',\epsilon'})$ by choosing the parameters $(n,\ell)$ appropriately. For ease of notation, we will drop the subscripts from the random variables when there is no ambiguity. Using Theorems \ref{th:7} and \ref{th:9}, we have:
\begin{align}
& \label{eq:App7:1}       |I(\widehat{X};  \overline{U})
        - I(X;U)|\leq 2\xi,\\
    & \label{eq:App7:2}  
    |I(\widehat{Y}; \overline{V})-I(Y;V)|\leq 2\xi,
\\& \label{eq:App7:3}   |I( \overline{U};  \overline{V})-  I(U;V)|\leq 2\xi. 
\\& 
\mathbb{E}d'_1({X},{U},{V})+\xi \geq  \mathbb{E}d'_1(\widehat{X},\overline{U},\overline{V})  \\
&\mathbb{E}d'_2({Y},{U},{V})+\xi \geq  \mathbb{E}d'_2(\widehat{Y},\overline{U},\overline{V}). 
\end{align}

We can use  \cite{berger1978multiterminal} to show that the rate-distortion tuple
given  by 
\begin{align}
R_1 &\geq I(\widehat{X};\overline{U})-I(\overline{U};\overline{V}), \nn \\
R_2 &\geq I(\widehat{Y};\overline{V})-I(\overline{U};\overline{V}), \nn \\
R_1+R_2 &\geq I(\widehat{X};\overline{U}) + I(\widehat{Y};\overline{V})
-I(\overline{U};\overline{V}) \nn \\
D_1 &\geq  \mathbb{E}d'_1(\widehat{X},\overline{U},\overline{V}) \nn \\
D_2 &\geq  \mathbb{E}d'_2(\widehat{Y},\overline{U},\overline{V}) \nn
\end{align}
 is achievable for the finite-alphabet source $(\widehat{X},\widehat{Y},d_1,d_2)$. Based on Equations \eqref{eq:App7:1}, \eqref{eq:App7:2}, and \eqref{eq:App7:3}, the rates are within $2\xi$ of the ones in the theorem statement. 
To show that the claimed distortions for the reconstruction of continuous source are achievable, consider a transmission system
with parameter $(m,\Theta_1,\Theta_2)$ for compressing the finite-alphabet source
such that 
\begin{equation}
\frac{1}{m}\sum_{i=1}^m \mathbb{E}d'_1(\widehat{X}_{i},\overline{U}'_i,\overline{V}'_i)
\leq
\mathbb{E}d'_1(\widehat{X},\overline{U},\overline{V})
+\e,
\label{eq:DSC_4}
\end{equation}
where $(\overline{U}'^{m},\overline{V}'^{m})\triangleq f(e_1(\widehat{X}^m),e_2(\widehat{Y}^m))$, and $e_1$ and 
$e_2$ denote the encoders and $f$ the decoder.

For the source $(X,Y,d_1,d_2)$ we obtain an $(m,\Theta_1',\Theta'_2)$  transmission system 
$\mbox{TS}_c$ as follows. We assume that the encoder and decoder share
common randomness.  From $(X,Y)$ we create
$(\widehat{X},\widehat{Y})$ and use $\mbox{TS}_d$.
Let $T_1=\mathbbm{1}_{\{X\in [-\ell,\ell] \}}$. The encoder of $X$ in 
$\mbox{TS}_c$ sends information to the decoder in two parts.
The first part is $e(\widehat{X}^m)$ and the second part is a  compressed (almost
lossless) version of $T_1^m$.  If the sequence $T_1^m$ is typical, the
encoder sends the index of the sequence in the typical set, otherwise,
it sends the index $0$. This extra piece of information requires
$h_b(P_X([-\ell,\ell])+\xi$ bits per sample. Similar encoding strategy is used at the other encoder. 
The decoder is constructed as follows. Let $(\check{U}^m,\check{V}^m)$ denote the reconstruction vector.
If the decoder receives index $0$ in the second part from either of the encoders, then
it uses an arbitrary constant $c$ as a reconstruction, i.e.,
$\check{U}_i=c, \check{V}_i=c'$ for all $i$. Otherwise, it can
reconstruct $(T_1^m,T_2^m)$ reliably. If $T_{1i}=T_{2i}=1$, then the reconstruction is
$\check{U}_i=\overline{U}'_i$ and $\check{V}_i=\overline{V}'_i$, otherwise it is arbitrary constants $(c,c')$,
i.e., $\check{U}_i=c, \check{V}_i=c'$. 

 Assume that the parameters of the transmission system that of the discrete source are such that
\begin{itemize}
\item $h_b(P_X([-\ell,\ell])) \leq \xi$, and $h_b(P_Y([-\ell,\ell])) \leq \xi$
\item $P(\mathsf{A}^c) \mathbb{E}d'_1(X_i,c,c'|\mathsf{A}^c) \leq \xi$, and $P(\mathsf{A}^c) \mathbb{E}d'_2(Y_i,c, c'|\mathsf{A}^c) \leq \xi$ for all $1 \leq i
  \leq m$, where $\mathsf{A}$ denotes the event
  that $(T_1^m,T_2^m)$ is jointly typical.
\item $P(\mathsf{B}_i^c) \mathbb{E}d'_1(X_i,c,c'|\mathsf{B}_i^c) \leq \xi$, and $P(\mathsf{B}_i^c) \mathbb{E}d'_2(Y_i,c, c'|\mathsf{B}_i^c) \leq \xi$ for all $1 \leq i
  \leq m$, where 
 $\mathsf{B}_i$ denote the event $(X_i,Y_i)
\in [-\ell,\ell] \times [-\ell,\ell]$.
\item 
$|d'_1(x_1,b_1,b_2)-d'_1(x_2,b_1,b_2)| \leq \xi$, 
for all (a)  $x_1,x_2 \in [-\ell,\ell]$, 
(b)  $b_1 \in [-\ell',\ell']$, (c) $b_2 \in [-\ell',\ell']$, and (d) $|x_1-x_2| \leq \frac{1}{2^{n}}$.
\item 
$|d'_2(y_1,b_1,b_2)-d'_2(y_2,b_1,b_2)| \leq \xi$, 
for all (a)  $x_1,x_2 \in [-\ell,\ell]$, 
(b)  $b_1 \in [-\ell',\ell']$, (c) $b_2 \in [-\ell',\ell']$, and (d) $|x_1-x_2| \leq \frac{1}{2^{n}}$.
\end{itemize}
Note that the constants $c,c'$ satisfying the conditions in the second and third bullets exist by arguments similar to those in Appendix \ref{app:wynerziv_c}.
 Let $P_{X_i,Y_i|\mathsf{A},\mathsf{B}_i}$ denote the
probability distribution of $(X_i,Y_i)$ given the event $\mathsf{A}$ and $\mathsf{B}_i$. Consider for any $i \in
\{1,2,\ldots,m\}$, 
\begin{align}
\mathbb{E}d'_1(X_i,\check{U}_i,\check{V}_i) &\leq  P(\mathsf{A}^c) \mathbb{E}d'_1(X_i,c,c'|\mathsf{A}^c)+
P(B_i^c) \mathbb{E}d'_1(X_i,c,c'|\mathsf{B}_i^c)  \nn \\
&\hspace{1in} + P(\mathsf{A} \cap \mathsf{B}_i)
\mathbb{E}d'_1(X_i,\check{U}_i,\check{V}_i|\mathsf{A},\mathsf{B}_i) \nn \\
&\overset{a}{\leq} 2\xi+
P(\mathsf{A} \cap \mathsf{B}_i) \mathbb{E}d'_1(X_i,\check{U}_i,\check{V}_i|\mathsf{A} \cap \mathsf{B}_i) \nn \\
&= 2\xi+ P(\mathsf{A} \cap \mathsf{B}_i) \mathbb{E}d'_1(X_i,\overline{U}'_i,\overline{V}'_i|\mathsf{A} \cap \mathsf{B}_i) \nn \\
&\leq  2\xi+ P(\mathsf{B}_i) \mathbb{E}d'_1(X_i,\overline{U}
'_i,\overline{V}
'_i|\mathsf{B}_i) \nn \\ 
&\overset{(b)}{=} 2\xi+ P(\mathsf{B}_i) \sum_{i,b,b'} P_i(b,b'|\zeta(i)) \int_{A(i)  } d'_1(x,b,b') \frac{dP_{X_i}(x)}{P(\mathsf{B}_i)}  \nn \\
&\overset{(c)}{\leq} 2\xi+P(\mathsf{B}_i) \left[
  \mathbb{E}d'_1(\widehat{X}_{i},\overline{U}_i,\overline{V}_i) +
  \xi \right], \nn 
\end{align}
where we have following arguments:
(a) follows from third and fourth bullets from the previous
page. In  (b) we have denoted the conditional probability of
$\overline{U}_i,\overline{V}_i$ given $\widehat{X}_{i}$ as
$P_i$. 
(c) follows from the fifth bullet from the previous page.
Finally, we have 
\begin{align}
\frac{1}{m} \sum_{i=1}^m \mathbb{E}d'_1(X_i,\check{U}_i,\check{V}_i)
&\leq 2\xi+\left[
  \mathbb{E}d'_1(\widehat{X},\overline{U},\overline{V})) +
  2\xi \right]  \nn \\
&\leq 5\xi+ \mathbb{E}d'_1(X,U,V). \nn
\end{align}
The proof for the distortion for reconstructing $Y$ follows by similar arguments.
 This completes the desired proof.
\qed

\section{Proof of Proposition \ref{th:MD:Ex1}}
\label{App:th:MD:Ex1}

\begin{figure}[!h]
\centering
\includegraphics[width=3in,height=1.8in]{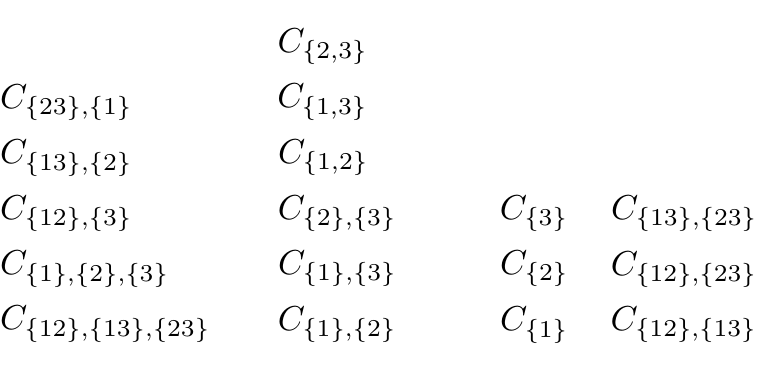}
\caption{Three-Descriptions Codebook Structure}
\label{code}
\end{figure}

 The proof follows similar similar steps as the one given in Example 3 in \cite{shirani2018achievable}. We provide an outline in the following. The codebook structure for in for SSC strategy with unstructured random codes is shown in Figure \ref{code}. Note that Decoders $\{1\}$,$\{2\}$ and $\{1,2\}$ operate at optimal PtP rate-distortion. So by the same arguments as in steps one through five in Example 3 in \cite{shirani2018achievable}, all codebooks except $C_{\{1\},\{3\}}, C_{\{2\},\{3\}}, C_{\{1\}}, C_{\{2\}}, C_{\{3\}}$ can be eliminated. For ease of notation, we denote the corresponding random variables for these codebooks as $U_{1,3},U_{2,3},U_1,U_2,U_3$, respectively.
 Note that due to optimality at Decoders $\{1\}$,$\{2\}$ and $\{1,2\}$ we must have:
 \begin{align*}
     U_{1,3},U_1 - X - Z -   U_{2,3},U_2,
 \end{align*}
In order to evaluate the achievable rates using Gaussian test-channels, let $U_{1,3}= X+Q_{1,3}$ and $U_{2,3}=Z+Q_2$, where $Q_{1,3},Q_{2,3}$ are two Gaussian variables which are independent of each other and of $X,Z$ with variances $\theta_1,\theta_2>0$, respectively. Note that $E(X|U_{1,3})= \frac{1}{1+\theta_1}U_{1,3}$ and $E(Z|U_{2,3})= \frac{1}{1+\theta_2}U_{2,3}$. In order to achieve the desired distortion at Decoders $\{1\}$ and $\{2\}$, we must have $U_1=X-\frac{1}{1+\theta_1}U_{1,3}+Q_{1}$ and  
$U_2=Z-\frac{1}{1+\theta_1}U_{2,3}+Q_{2}$, where $Q_1$ and $Q_2$ are Gaussian variables with zero mean and variance $P$, independent of each other and all other variables. Then, the reconstruction $\hat{X}= U_{1,3}+U_1$ and $\hat{Z}= U_{2,3}+U_2$ satisfy the distortion constraints at Decoder $\{1\}$ and Decoder $\{2\}$, respectively. The Gaussian variable $U_3$ can be decomposed in terms of $X,Z, Q_{1,3},Q_{2,3},Q_{1},Q_2,Q_3$, where $Q_3$ is an independent Gaussian variable with zero mean and unit variance, so that $U_3= \alpha_1 X+\alpha_2 Z + \alpha_3 Q_{1,3} + \alpha_4 Q_{2,3}+\alpha_5 Q_1 +\alpha_6 Q_2 +\alpha_7 Q_3$ for some $\alpha_i\in \mathbb{R}, i\in [7]$. Then, the reconstruction at Decoder $\{3\}$ of $X+Z$ which minimizes the distortion is given by:
\begin{align*}
    \widehat{(X+Z)}_{3} \triangleq\mathbb{E}(X+Z|U_{1,3},U_{2,3},U_3)&= \Sigma_{X+Z,U_{1,3}U_{2,3}U_3}\Sigma^{-1}_{U_{1,3},U_{2,3},U_3}
    \begin{bmatrix}
    U_{1,3}&U_{2,3}&U_3
    \end{bmatrix}^t
    \\&= \frac{1}{1+\theta_1}U_{1,3}+\frac{1}{1+\theta_2}W_{2,3}+\frac{\alpha_1+\alpha_2}{Var(U_3)}U_3,
\end{align*}
where  $\Sigma_{X+Z,U_{1,3}U_{2,3}U_3}\triangleq  \mathbb{E}((X+Z)[U_{1,3},U_{2,3},U_3])$ and  $\Sigma_{U_{1,3},U_{2,3},U_3}$ is the covariance matrix of ${U_{1,3},U_{2,3},U_3}$. Similarly, the reconstructions at Decoder $\{1,3\}$ are:
\begin{align*}
&  \widehat{X}_{1,3}  \triangleq \mathbb{E}(X|U_{1,3},U_{2,3},U_1,U_3)= \Sigma_{X,U_{1,3}U_{2,3}U_1U_3}\Sigma^{-1}_{U_{1,3},U_{2,3},U_1,U_3}
    \begin{bmatrix}
    U_{1,3}&U_{2,3}&U_1&U_3
    \end{bmatrix}^t,
    \\& 
 \widehat{Z}_{1,3}  \triangleq \mathbb{E}(Z|U_{1,3},U_{2,3},U_1,U_3)= \Sigma_{Z,U_{1,3}U_{2,3}U_1U_3}\Sigma^{-1}_{U_{1,3},U_{2,3},U_1,U_3}
    \begin{bmatrix}
    U_{1,3}&U_{2,3}&U_1&U_3
    \end{bmatrix}^t  \end{align*}
  The reconstructions $\widehat{X}_{2,3}, \widehat{Z}_{2,3}$ at Decoder $\{2,3\}$ can be written in a similar fashion. Furthermore, using the covering and packing bounds in Theorem \ref{thm:RDMD}, we have:
  \begin{align*}
      R_3&\geq I(X,Z;U_3,U_{1,3},U_{2,3})+I(U_1,U_2;U_3|U_{1,3},U_{2,3},X,Z)
      \\&= \frac{1}{2}\log{\frac{|\Sigma_{X,Z}||\Sigma_{U_{1,3}U_{2,3}U_3}|}{|\Sigma_{X,Z,U_{1,3},U_{2,3}U_3}|}}+I(Q_1,Q_2;\alpha_5Q_1+\alpha_6Q_2+\alpha_7Q_3)
      \\&= \frac{1}{2}\log{\frac{|\Sigma_{U_{1,3}U_{2,3}U_3}|}{|\Sigma_{X,Z,U_{1,3},U_{2,3}U_3}|}}+\frac{1}{2}\log{(\alpha_5^2P+\alpha_6^2P+\alpha_7^2)}.
  \end{align*}
  Optimizing over $\alpha_i, i\in [7], \theta_1,\theta_2$ yields the achievable rates plotted in Figure \ref{Fig:MD:Ex1}.
 \end{appendices}
\bibliographystyle{unsrt}
  \bibliography{references}

\end{document}